\documentclass[english, 11pt]{article}
\usepackage{lmodern}

\usepackage[T1]{fontenc}
\UseRawInputEncoding
\usepackage{xcolor}
\usepackage{babel}
\usepackage{mathtools}
\usepackage{enumitem}
\usepackage{amsmath}
\usepackage{amsthm}
\usepackage{amssymb}
\usepackage{environ}
\usepackage{graphicx}
\usepackage{authblk}

\usepackage{moreenum}
\usepackage{latexsym}
\usepackage{footnotehyper}   
\makesavenoteenv{minipage}

\usepackage{hyperref}
\hypersetup{breaklinks=false, pdfborder={0 0 0}, pdfborderstyle={}, colorlinks=true, urlcolor=NavyBlue,  linkcolor=NavyBlue,  citecolor=NavyBlue}

\usepackage{geometry}
\PassOptionsToPackage{normalem}{ulem}
\usepackage{ulem}

\usepackage{tikz}
\usetikzlibrary{matrix,calc}
\usepackage{pifont}

\makeatletter

\providecommand{\tabularnewline}{\\}

\numberwithin{figure}{section}
\theoremstyle{definition}
\newtheorem{claim}{\protect\claimname}[section]
\newtheorem{prop}{\protect\propositionname}[section]
\newtheorem{lem}{\protect\lemmaname}[section]
\newtheorem{defn}{\protect\definitionname}[section]
\theoremstyle{plain}
\newtheorem{thm}{\protect\theoremname}
\theoremstyle{remark}
\newtheorem{rem}[thm]{\protect\remarkname}
\newlist{casenv}{enumerate}{4}
\setlist[casenv]{leftmargin=*,align=left,widest={iiii}}
\setlist[casenv,1]{label={{\itshape\ \casename} \arabic*.},ref=\arabic*}
\setlist[casenv,2]{label={{\itshape\ \casename} \roman*.},ref=\roman*}
\setlist[casenv,3]{label={{\itshape\ \casename\ \alph*.}},ref=\alph*}
\setlist[casenv,4]{label={{\itshape\ \casename} \arabic*.},ref=\arabic*}

\usepackage[dvipsnames]{xcolor}
\usepackage{bbm}      
\usepackage{upgreek}
\usepackage{amsthm}
\usepackage{wrapfig}
\usepackage{ragged2e}

\let\myTOC\tableofcontents
\renewcommand{\tableofcontents}{%
 { 
 \hypersetup{linkcolor = black}
 \pagenumbering{roman}
 \pdfbookmark[1]{\contentsname}{}
 \myTOC
 \cleardoublepage
 }
 \pagenumbering{arabic}}

\@ifundefined{qedsymbol}{\relax}{}

\DeclareSymbolFont{bbold}{U}{bbold}{m}{n}
\DeclareSymbolFontAlphabet{\mathbbold}{bbold}

\relax
\relax

\theoremstyle{definition}
\newtheorem{notation}{Notation}[section]

\newcommand{\minipar}[1]{%
  \begin{minipage}[m]{0.81\linewidth}%
    #1\unskip%
  \end{minipage}%
}

\newcommand{\modifiedminipar}[2]{%
  \begin{minipage}[m]{#1\linewidth}%
    #2\unskip%
  \end{minipage}%
}

\newcommand{\justifiedminipage}[1]{%
  \minipar{%
    \justifying%
    \setlength{\parfillskip}{0pt}%
    \setlength{\parindent}{0pt}
    \ignorespaces#1\unskip%
  }%
}
\newcommand{\modifiedjustifiedminipage}[2]{%
	\begin{minipage}[m]{#1\linewidth}%
    \justifying%
    \setlength{\parfillskip}{0pt}%
    \setlength{\parindent}{0pt}
    \ignorespaces#2\unskip
  \end{minipage}%
}

\AtBeginDocument{%
  \setlist[casenv,3]{label={{\itshape\ \casename} \greek*.},ref=\greek*}
  }%

\makeatother

\usepackage{listings}
\providecommand{\casename}{Case}
\providecommand{\remarkname}{Remark}
\providecommand{\theoremname}{Theorem}
\providecommand{\definitionname}{Definition}
\providecommand{\lemmaname}{Lemma}
\providecommand{\propositionname}{Proposition}
\providecommand{\claimname}{Claim}

\newcommand{\cmark}{\textcolor{green!60!black}{\ding{51}}}
\newcommand{\xmark}{\textcolor{red}{\ding{55}}}

\usepackage{cancel}
\usepackage{comment}

\newcommand{\lyxadded}[3]{#3}
\newcommand{\lyxdeleted}[3]{}

\usepackage{setspace}
\setlist{itemsep=2pt, topsep=4pt, parsep=0pt, partopsep=0pt}

\makeatletter
\newcommand{\LinkColorText}[1]{%
  {\color{\@linkcolor}#1}%
}
\makeatother

\usepackage{subfiles} 

\title{The Domain of RSD Characterization by\\ Efficiency, Symmetry, and Strategy-Proofness}
\author{Maor Ben Zaquen}
\author{Ron Holzman}
\affil{Department of Mathematics, Technion - Israel Institute of Technology\\ \href{mailto:mben-zaquen@campus.technion.ac.il}{\texttt{mben-zaquen@campus.technion.ac.il}}, \href{mailto:holzman@technion.ac.il}{\texttt{holzman@technion.ac.il}}}
\date{\today}

\begin{document}
\maketitle
\begingroup
\renewcommand\thefootnote{}
\footnotetext{
Useful discussions and correspondence with Itai Ashlagi, Christian Basteck, Lars Ehlers, Yannai Gonczarowski, Alexander Nesterov, David Ryz\'{a}k, Fedor Sandomirskiy, and David Sychrovsk\'{y} are gratefully acknowledged.}
\endgroup

\begin{abstract}
Given a set of $n$ individuals with strict preferences over $m$
indivisible objects, the Random Serial Dictatorship (RSD) mechanism
is a method for allocating objects to individuals in a way that is
efficient, fair, and incentive-compatible. A random order of individuals
is first drawn, and each individual, following this order, selects
their most preferred available object. The procedure continues until
either all objects have been assigned or all individuals have received
an object.

RSD is widely recognized for its application in fair allocation problems
involving indivisible goods, such as school placements and housing
assignments. Despite its extensive use, a comprehensive axiomatic
characterization has remained incomplete. For the balanced case $n=m=3$,
Bogomolnaia and Moulin \cite{bogomolnaia2001new} have shown that RSD is uniquely characterized by Ex-Post Efficiency,
Equal Treatment of Equals, and Strategy-Proofness. The possibility
of extending this characterization to larger markets had been a long-standing
open question, which Basteck and Ehlers \cite{basteck2025constrained} recently answered in the negative for all markets with $n,m\geq5$.

This work completes the picture by identifying exactly for which pairs
$\left(n,m\right)$ these three axioms uniquely characterize the RSD
mechanism and for which pairs they admit multiple mechanisms. In the
latter cases, we construct explicit alternatives satisfying the axioms
and examine whether augmenting the set of axioms could rule out these
alternatives.
\end{abstract}

\section{Introduction}
The problem of assigning indivisible objects to agents in a fair and
efficient way is central in mechanism design. Typical applications
include school choice, assignment of students to dorm rooms, and allocation
of public housing. In such environments, monetary transfers are often
unavailable or undesirable, and the only information agents report
is a ranking of the available objects. A mechanism must then map these
preference profiles into allocations that satisfy normative criteria
such as efficiency, fairness, and incentive compatibility. Since deterministic
rules often fail to meet fairness requirements even when considered
on their own, and in particular cannot satisfy efficiency, fairness,
and incentive compatibility simultaneously, the literature has focused
on randomized mechanisms, which assign to each agent a lottery over
the objects, with an outside option of receiving nothing. Throughout,
we assume that every agent ranks this outside option below every object.

A particularly prominent mechanism in this context is Random Serial
Dictatorship (RSD), also known in the literature as Random Priority.
Given a profile of strict preferences and a set of indivisible objects,
RSD samples an ordering of the agents uniformly at random and lets
them choose their most preferred available object in that order; if
there are more agents than objects, the last agents receive nothing,
whereas if there are more objects than agents, some objects remain
unassigned. In the random assignment literature, RSD is an extensively
studied mechanism. For example, Abdulkadiro\u{g}lu and S\"{o}nmez
\cite{abdulkadirouglu1998random} and, independently, Knuth \cite{knuth1996exact}
have shown that RSD coincides with the core from random endowments,
and Bade \cite{bade2020random} has shown that symmetrizing any deterministic
efficient mechanism satisfying strategy-proofness and non-bossiness
yields RSD.

RSD satisfies ex-post efficiency, equal treatment of equals, and strategy-proofness
(these three properties will be referred to as \emph{the axioms}).
Ex-post efficiency requires that, for every preference profile, the
mechanism only randomizes over deterministic assignments that are
Pareto efficient; equal treatment of equals requires that agents with
identical preferences receive identical lotteries; and strategy-proofness
requires that truthful reporting is a weakly dominant strategy for
every agent. These three axioms largely explain its popularity in
practice. At the same time, it is natural to consider stronger notions
of efficiency and fairness, and this led to the study of alternative
mechanisms alongside RSD. Bogomolnaia and Moulin \cite{bogomolnaia2001new}
proposed the Probabilistic Serial (PS) mechanism as an alternative
to RSD.\footnote{PS satisfies ordinal efficiency and envy-freeness, which are stronger
than ex-post efficiency and equal treatment of equals, respectively,
whereas RSD does not; on the other hand, PS does not satisfy strategy-proofness,
but only a weaker version of it.} Beyond verifying that a given mechanism satisfies desirable axioms,
an important question in mechanism design is whether a given set of
axioms characterizes that mechanism. \lyxdeleted{bzmao}{Sat Jan 10 22:51:34 2026}{For
example, Hashimoto et al. }\lyxdeleted{bzmao}{Sat Jan 10 22:51:34 2026}{\cite{hashimoto2014two}}\lyxdeleted{bzmao}{Sat Jan 10 22:51:34 2026}{,
Bogomolnaia }\lyxdeleted{bzmao}{Sat Jan 10 22:51:34 2026}{\cite{bogomolnaia2015random}}\lyxdeleted{bzmao}{Sat Jan 10 22:51:34 2026}{,
and the joint work of Bogomolnaia and Heo }\lyxdeleted{bzmao}{Sat Jan 10 22:51:34 2026}{\cite{bogomolnaia2012probabilistic}}\lyxdeleted{bzmao}{Sat Jan 10 22:51:34 2026}{
provide axiomatic characterizations of the PS mechanism.}\lyxadded{bzmao}{Sat Jan 10 22:53:18 2026}{For
example, Bogomolnaia and Heo }\lyxadded{bzmao}{Sat Jan 10 22:52:11 2026}{\cite{bogomolnaia2012probabilistic}}\lyxadded{bzmao}{Sat Jan 10 22:53:18 2026}{,
Hashimoto et al. }\lyxadded{bzmao}{Sat Jan 10 22:52:34 2026}{\cite{hashimoto2014two}}\lyxadded{bzmao}{Sat Jan 10 22:53:18 2026}{,
and Bogomolnaia }\lyxadded{bzmao}{Sat Jan 10 22:52:57 2026}{\cite{bogomolnaia2015random}}\lyxadded{bzmao}{Sat Jan 10 22:53:18 2026}{
provide axiomatic characterizations of the PS mechanism.}

In this work we focus on the question of when the axioms uniquely
characterize RSD. For the balanced case with three agents and three
objects, Bogomolnaia and Moulin \cite{bogomolnaia2001new} showed
that the axioms uniquely characterize RSD. This raised a natural question:
do the axioms single out RSD for all numbers of agents and objects?
This question remained open for a long time and attracted sustained
attention. Its difficulty also motivated work on related problems,
including axiomatic characterizations of RSD in related frameworks.
Basteck \cite{basteck2025axiomatization} has shown that, when randomized
mechanisms are viewed as assigning lotteries over deterministic assignments
to each profile, RSD can be characterized by ex-post efficiency, equal
treatment of equals, and probabilistic (Maskin) monotonicity; and
Pycia and Troyan \cite{pycia2024random} have represented randomized
mechanisms as an extensive-form game and have obtained such a characterization
by strengthening strategy-proofness to obvious strategy-proofness.
The difficulty of this question also spurred work on the limits of
the axiom system; in particular, each of the following one-axiom strengthenings
leads to an impossibility result, in the sense that no mechanism satisfies
the resulting strengthened axiom system: Bogomolnaia and Moulin \cite{bogomolnaia2001new}
strengthen ex-post efficiency to ordinal efficiency; in a different
setting with cardinal utilities, Zhou \cite{zhou1990conjecture} strengthens
ex-post efficiency to ex-ante efficiency; Nesterov \cite{nesterov2017fairness}
strengthens equal treatment of equals to envy-freeness; and Bade \cite{bade2016fairness}
strengthens strategy-proofness to group strategy-proofness.

At the center of attention was the question whether the axioms uniquely
characterize RSD in the balanced case where the number of objects
equals the number of agents. In that line of work, the case of four
agents and four objects became a touchstone: it was the first unsolved
balanced market size, and it resisted purely analytic proofs, being
resolved only via a computer-assisted analysis (Sandomirskiy \cite{fssite}
and unpublished computerized verification by others). Finally, the
long-standing question was answered in the negative by Basteck and
Ehlers \cite{basteck2025constrained}, who constructed mechanisms
different from RSD that still satisfy the axioms and showed that such
mechanisms exist for all markets with at least five agents and at
least five objects. Their construction is inspired by a result of
Erdil \cite{erdil2014strategy}, who has shown that if agents are
allowed to rank receiving nothing above some objects, then mechanisms
other than RSD exist that satisfy the axioms.

The present work completes this picture. We fix a number of agents
$n$ and a number of indivisible objects $m$ with $n,m\geq2$, and
consider the randomized mechanisms that satisfy the axioms. Our main
result is a complete classification of the pairs $\left(n,m\right)$
according to whether the axioms uniquely characterize RSD or whether
they admit additional distinct mechanisms. On the positive side, we
show that the axioms uniquely determine the assignment probabilities
for every preference profile whenever there are at most three agents
and an arbitrary number of objects, and in the balanced case with
four agents and four objects. On the negative side, we show that outside
this domain the axioms are too weak: for each remaining pair $\left(n,m\right)$
we show how to construct mechanisms that satisfy the axioms, yet differ
from RSD. In retrospect, the full classification obtained in the present
work helps explain why the balanced case $n=m=4$ proved difficult:
when $n=4$, uniqueness holds only in the balanced case $m=4$, a
phenomenon that does not arise for other values of $n$. Table \textcolor{blue}{\ref{tab:results}}
summarizes the uniqueness and non-uniqueness results.

\begin{table}[!tph]
\centering
\begin{tikzpicture}[
  x=2.2cm, y=-1.3cm, 
  line width=0.4pt,
  font=\normalsize
]

\draw (0,0) rectangle (5,5);

\draw (1,0) -- (1,5);
\draw (2,0) -- (2,1);
\draw[very thin] (2,2) -- (2,5);

\draw (3,0) -- (3,1);
\draw[very thick] (3,1) -- (3,3);
\draw[very thin] (3,3) -- (3,4);
\draw[very thick] (3,4) -- (3,5);

\draw (4,0) -- (4,1);
\draw[very thick] (4,3) -- (4,4);
\draw[very thin] (4,4) -- (4,5);

\draw (0,1) -- (5,1);
\draw (0,2) -- (1,2);
\draw[very thin] (1,2) -- (5,2);
\draw (0,4) -- (1,4);
\draw[very thick] (3,4) -- (4,4);
\draw[very thin] (4,4) -- (5,4);

\draw (0,3) -- (1,3);
\draw[very thin] (2,3) -- (3,3);
\draw[very thick] (3,3) -- (4,3);

\draw (0,0) -- (1,1);              
\node at (0.72,0.28) {$n$};        
\node at (0.28,0.72) {$m$};        

\node at (1.5,0.5) {$2$};
\node at (2.5,0.5) {$3$};
\node at (3.5,0.5) {$4$};
\node at (4.5,0.5) {$\ge 5$};

\node at (0.5,1.5) {$2$};
\node at (0.5,2.5) {$3$};
\node at (0.5,3.5) {$4$};
\node at (0.5,4.5) {$\ge 5$};

\node at (2,1.5) {\cmark{ \scriptsize(Section \ref{sec:m=2})}};
\node at (4,1.5) {\xmark{ \scriptsize(Section \ref{sec:m=2})}};
\node at (1.5,3.5) {\cmark { \scriptsize(Remark \ref{rem:nEq2_mGt2})}};
\node at (2.5,2.5) {\cmark{ \scriptsize\cite{bogomolnaia2001new}}};
\node at (2.5,4) {\cmark{ \scriptsize(Section \ref{subsec:n=3 and m>=3})}};
\node at (4.41,2.7) {\xmark{ \scriptsize(Section \ref{subsec:n>m=3,4})}};

\node[scale=0.75, transform shape, align=center] at (3.5,3.5) {%
  \cmark{}\\[0.5ex]
  \tikz[baseline]{
    \node[anchor=base] (v) {%
      {\scriptsize$\left(\begin{array}{c}
        \text{Section \LinkColorText{\ref*{subsec:n=m=4}}}\\
        \text{\& Appendix \LinkColorText{\ref*{sec:nEqmEq4_cases}}}
      \end{array}\right)$}%
    };
    \node[anchor=base, opacity=0] at ($(v.base)+(2pt,3.9pt)$) {%
      \tiny\ref{subsec:n=m=4}
    };
    \node[anchor=base, opacity=0] at ($(v.base)+(8.5pt,-3.2pt)$) {%
      \tiny\ref{sec:nEqmEq4_cases}
    };
  }
};

\node at (3.5,4.5) {\xmark{ \scriptsize(Section \ref{subsec:n=4, m>=5})}};
\node at (4.5,4.5) {\xmark{ \scriptsize\cite{basteck2025constrained}}};




\end{tikzpicture}

\caption{\label{tab:results}Do the axioms uniquely determine RSD with $n$
agents and $m$ objects?}
(\cmark{} yes, \xmark{} no)
\end{table}

For markets with at least three agents in the domain where the axioms
uniquely characterize RSD, our proofs rely on a local analysis of
the axioms, using adjacent swaps in agents' rankings to propagate
constraints across profiles and show that the assignment matrix is
uniquely determined for every preference profile. Notably, these uniqueness
arguments only require weakened versions of ex-post efficiency and
strategy-proofness. First, the weakened version of ex-post efficiency
required in our uniqueness arguments is support efficiency (in the
terminology of Brandt et al. \cite{brandt2023towards}), together
with the full assignment property (Definition \ref{def:full_assignment_property}).
Support efficiency requires that an agent can have a positive probability
of receiving a given object only if he receives it under some Pareto
efficient deterministic assignment, and the full assignment property
requires that whenever an agent (respectively, an object) is matched
in every serial dictatorship outcome, the mechanism matches that agent
(respectively, allocates that object) with probability $1$. Second,
in place of full strategy-proofness, it suffices to assume upper invariance
and lower invariance (in the terminology of Mennle and Seuken \cite{mennle2021partial}),
two conditions which, together with swap monotonicity, are equivalent
to strategy-proofness. Upper and lower invariance mean that when an
agent swaps two adjacent objects in his ranking, his probabilities
of receiving any other object, as well as the outside option, remain
unchanged.

We also provide a complete description of the case $m=2$. In this
case, we show that every mechanism satisfying the axioms can be represented
by a single function that assigns to each subset of agents the probability
with which its members receive their preferred object, when they rank
the two objects the same way, and the other agents rank them the opposite
way. This yields a simple parameterization of the entire class of
mechanisms satisfying the axioms and allows us to determine exactly
when RSD is the unique mechanism and when there is a continuum of
alternatives.

In the non-uniqueness domain we follow a different strategy, conceptually
close in spirit to the construction of Basteck and Ehlers \cite{basteck2025constrained}.
We first identify RSD as the symmetrization of a specific mechanism
that satisfies ex-post efficiency and strategy-proofness. We then
perform carefully designed adjustments to this underlying mechanism
on a restricted family of profiles, chosen so that these two axioms
remain satisfied. Symmetrizing the adjusted mechanism yields a mechanism
that satisfies all three axioms yet is distinct from RSD.

Building on the work of Basteck and Ehlers \cite{basteck2025constrained},
whose construction yields a mechanism that satisfies the axioms and,
in addition, anonymity (a strengthening of equal treatment of equals)
and neutrality, we then investigate whether further strengthening
the axiom system can restore uniqueness. In particular, we consider
adding the bounded invariance axiom and answer their open question
on whether this addition can restore uniqueness of RSD in the negative.
We then go further by additionally imposing non-bossiness and cross
monotonicity, and show that even with these extra axioms, the resulting
system still does not uniquely characterize RSD in sufficiently large
markets.

The work is organized as follows. Section \ref{sec:preliminaries}
introduces the formal model, recalls the RSD mechanism together with
the axioms of ex-post efficiency, equal treatment of equals, and strategy-proofness,
and establishes several preliminary lemmas, including characterizations
of ex-post efficiency and strategy-proofness and further properties
implied by the axioms that are used later in the proofs. In Section
\ref{sec:m=2} we analyze the case of two objects and obtain a complete
parameterization of all mechanisms satisfying the axioms in this setting.
Section \ref{sec:uniqueness} contains the positive results identifying
all market sizes for which the axioms uniquely determine RSD. Section
\ref{sec:non-uniqueness} provides the non-uniqueness results by constructing
explicit alternative mechanisms for all remaining values of $n$ and
$m$, and shows that adding the bounded invariance axiom does not
restore uniqueness of RSD.
Finally, the appendix consists of two parts. Appendix \ref{sec:appendix-adding_axioms}
introduces the non-bossiness and cross monotonicity axioms, and shows
that even adding these axioms on top of the previous ones does not
restore uniqueness of RSD in sufficiently large markets. Appendix
\ref{sec:nEqmEq4_cases} presents the exhaustive case analysis for
the market with four agents and four objects.

\section{Preliminaries} \label{sec:preliminaries}
\subsection{The model}
We begin with some preliminary definitions. Let $N$ denote the set of individuals, referred to as \textit{agents}, and $H$ the set of indivisible objects, referred to as \textit{houses}. We define their sizes as $\left|N\right|\coloneqq n$ and $\left|H\right|\coloneqq m$, and identify the set of agents with $\left[n\right]=\left\{ 1,\dots,n\right\}$.

We consider an environment in which each house can be assigned to at most one agent, and each agent can receive at most one house. Since $n$ and $m$ may differ, we define an \textit{assignment} as a matching in the complete bipartite graph with vertex set $N\cup H$, and let $\mathcal{S}$ denote the set of all such assignments. An assignment $s\in\mathcal{S}$ \textit{assigns} agent $i\in N$ to house $h\in H$ if $\left\{ i,h\right\} \in s$, and we write $s\left(i\right)\coloneqq h$. If $i$ is not assigned any house under $s$, we say that $i$ is \textit{unassigned} and write $s\left(i\right)\coloneqq\varnothing$, where $\varnothing\notin H$ denotes the \textit{null object}. We define the set of all possible objects as $O\coloneqq H\cup\left\{ \varnothing\right\}$ . From this point on, we use \textit{object} to refer to an element of $O$ (including the null object), and \textit{house} to refer specifically to an element of $H$.

We consider settings where each agent's preference depends only on the object they are assigned, rather than on the full assignment. A \textit{preference order} is a total order over the set $O$, in which every house $h\in H$ is strictly preferred to the null object $\varnothing$. Let $\mathcal{R}$ denote the set of all preference orders. For a preference order $R\in\mathcal{R}$ and two objects $o,o^{\prime}\in O$ the notation $o^{\prime}Ro$ denotes that $o^{\prime}$ is weakly preferred to $o$ under $R$, and $o^{\prime}R^{+}o$ denotes that $o^{\prime}$ is strictly preferred to $o$ under $R$. The \textit{upper contour set} of $o$ with respect to $R$ is defined as
\[C_{R}\left(o\right)\coloneqq\left\{ o^{\prime}\in O\mid o^{\prime}Ro\right\}.\]In other words, $C_{R}\left(o\right)$ is the set of objects that are at least as good as $o$ according to the preference order $R$.

Since we will consider randomized mechanisms, agents are assigned probability distributions over objects rather than the objects themselves. For a preference order $R\in\mathcal{R}$, we extend the preference relation to a partial order over the set $\Delta\left(O\right)$ of probability distributions over $O$, using \textit{first-order stochastic dominance}. Specifically, for $p,p^{\prime}\in\Delta\left(O\right)$, we say that an agent with preference order $R$ \textit{weakly prefers} $p$ to $p^{\prime}$, denoted $p\succeq_{R}p^{\prime}$, if for every object $o\in O$, \[\sum_{o^{\prime}\in C_{R}\left(o\right)}p\left(o^{\prime}\right)\geq\sum_{o^{\prime}\in C_{R}\left(o\right)}p^{\prime}\left(o^{\prime}\right).\]We say that the agent \textit{strictly prefers} $p$ to $p^{\prime}$, denoted $p\succ_{R}p^{\prime}$, if $p\succeq_{R}p^{\prime}$ and the inequality is strict for at least one $o\in O$.

A \textit{preference profile} is an element $\mathbf{P}=\left(P_{1},\dots,P_{n}\right)\in\mathcal{R}^{N}$ that assigns to each agent $i\in N$ a preference order $P_{i}\in\mathcal{R}$. For $o,o^{\prime}\in O$, if $oP_{i}o^{\prime}$, we say that agent $i$ \textit{weakly prefers} $o$ to $o^{\prime}$ under the profile $\mathbf{P}$. 

A \textit{deterministic mechanism} is a function $\mathcal{R}^{N}\rightarrow\mathcal{S}$, assigning to each preference profile a specific assignment. Similarly, an \textit{extensive-form randomized mechanism} is a function $M:\mathcal{R}^{N}\rightarrow\Delta\left(\mathcal{S}\right)$, which maps each profile to a lottery over assignments. However, in most cases, our primary concern is the outcome itself rather than the process by which it is generated. Therefore, we introduce the following definition.

\begin{defn}[Normal-form randomized mechanism]
A \textit{normal-form randomized mechanism} is a function $f:\nolinebreak\mathcal{R}^{N}\rightarrow\left[0,1\right]^{H\times N}$ that directly maps each preference profile $\mathbf{P}$ to a probability matrix. For every house-agent pair $\left(h,i\right)\in H\times N$, the entry $f\left(\mathbf{P}\right)_{h,i}$ denotes the probability that agent $i$ receives house $h$.
\end{defn}

For each $h\in H$, denote by $f\left(\mathbf{P}\right)_{h}$ the row vector of $f\left(\mathbf{P}\right)$ corresponding to $h$. Similarly, for each $i\in N$, denote by $f\left(\mathbf{P}\right)_{i}$ the column vector of $f\left(\mathbf{P}\right)$ corresponding to $i$. With a slight abuse of notation, we identify $f\left(\mathbf{P}\right)_{i}$ with the probability distribution that assigns to agent $i$ each house $h$ with probability $f\left(\mathbf{P}\right)_{h,i}$ and the null object with probability $1-\sum_{h\in H}f\left(\mathbf{P}\right)_{h,i}$.

We also define the normal form of an extensive-form randomized mechanism $M:\mathcal{R}^{N}\rightarrow\Delta\left(\mathcal{S}\right)$ to be the normal-form randomized mechanism $f_{M}:\mathcal{R}^{N}\rightarrow\left[0,1\right]^{H\times N}$ that, for every preference profile $\mathbf{P}$, yields the same assignment probabilities as those induced by $M$. Specifically, for every $\mathbf{P}\in\mathcal{R}^{N}$, $i\in N$ and $h\in H$, 

\[f_{M}\left(\mathbf{P}\right)_{h,i}\coloneqq\sum_{s\in\mathcal{S}:s\left(i\right)=h}M\left(\mathbf{P}\right)\left(s\right).\]We say that two randomized mechanisms are \textit{welfare equivalent} if they have the same normal form. Unless stated otherwise, we will use the term \textit{mechanism} to refer to a randomized mechanism in its normal form.

Using the definitions above, we now introduce the \textit{Random Serial Dictatorship (RSD)} mechanism and the main axioms used throughout the paper. To define RSD in its extensive form, we begin with the deterministic mechanisms that generate the assignments in the support of ${\rm RSD}\left(\mathbf{P}\right)$ for each $\mathbf{P}\in\mathcal{R}^{N}$.

Given a permutation $\sigma\in S_{n}$, the \textit{Serial Dictatorship} mechanism ${\rm SD}_{\sigma}$ is the deterministic mechanism in which agents are ordered according to $\sigma$  and each, in turn, selects their most preferred available house (i.e., one not chosen by any earlier agent in the sequence), and if $n>m$, the last $n-m$ agents get the null object.

\begin{defn}[RSD]
The \textit{Random Serial Dictatorship} mechanism, denoted RSD, is the randomized mechanism that samples a permutation $\sigma\in S_{n}$ uniformly at random and applies the corresponding serial dictatorship mechanism ${\rm SD}_{\sigma}$. Thus, in its extensive form, \[{\rm RSD}\left(\mathbf{P}\right)\coloneqq\frac{1}{n!}\sum_{\sigma\in S_{n}}{\rm SD}_{\sigma}\left(\mathbf{P}\right).\]
\end{defn}

An assignment $s\in\mathcal{S}$ is  \textit{(Pareto) efficient} with respect to a profile $\mathbf{P}\in\mathcal{R}^{N}$ if there is no other assignment $s^{\prime}\in\mathcal{S}$ such that some agent $i$ strictly prefers $s^{\prime}\left(i\right)$ to $s\left(i\right)$ without another agent $j$ strictly preferring $s\left(j\right)$ to $s^{\prime}\left(j\right)$. In other words, whenever there exists an agent $i$ who strictly prefers $s^{\prime}\left(i\right)$ over $s\left(i\right)$, there must also be at least one agent $j$ who strictly prefers $s\left(j\right)$ over $s^{\prime}\left(j\right)$.

A deterministic mechanism is \textit{efficient} if it produces an efficient assignment for every profile. Similarly, an extensive form mechanism is \textit{ex-post efficient} if, for every profile, its output distribution is supported entirely on efficient assignments with respect to that profile.

Given our primary focus on mechanisms in normal form, we define efficiency for them as follows:

\begin{defn}[ExPE]
A normal form mechanism $f$ satisfies \textit{Ex-Post Pareto Efficiency (ExPE)} if there exists an ex-post efficient extensive-form mechanism $M$ such that $f_{M}=f$.
\end{defn}

The following fairness axiom reflects the idea that agents who are identical in all relevant aspects should be treated identically.

\begin{defn}[ETE]
A mechanism $f:\mathcal{R}^{N}\rightarrow\left[0,1\right]^{H\times N}$ satisfies \textit{Equal Treatment of Equals (ETE)} if, for every preference profile and for every pair of agents with identical preferences in that profile, the mechanism assigns them the same distribution over objects. Formally, for every $\mathbf{P}\in\mathcal{R}^{N}$ and every pair of agents $i,i^{\prime}\in N$, if $P_{i}=P_{i^{\prime}}$, then $f\left(\mathbf{P}\right)_{i}=f\left(\mathbf{P}\right)_{i^{\prime}}$.
\end{defn}

Given a profile $\mathbf{P}\in\mathcal{R}^{N}$ and a preference order $R_{i}\in\mathcal{R}$, we denote by $\left(\mathbf{P}_{-i},R_{i}\right)\in\mathcal{R}^{N}$ the profile obtained by replacing agent $i$'s preference in $\mathbf{P}$ with $R_{i}$, while keeping all other agents' preferences unchanged. 

\begin{defn}[SP]
A mechanism $f:\mathcal{R}^{N}\rightarrow\left[0,1\right]^{H\times N}$ is \textit{strategy-proof (SP)} if for every profile $\mathbf{P}\in\mathcal{R}^{N}$, every agent $i\in N$, and every $R_{i}\in\mathcal{R}$, the following holds: \[f\left(\mathbf{P}\right)_{i}\succeq_{P_{i}}f\left(\mathbf{P}_{-i},R_{i}\right)_{i}.\]In other words, each agent weakly prefers the distribution over the objects they receive when reporting their true preferences to any distribution they could obtain by misreporting unilaterally.
\end{defn}

Having defined the main axioms, we next lay out notational conventions for agent and house renamings, then formalize associated invariance properties. We denote by $\Pi$ the set of all permutations of $N$, that is,
all bijections $N\rightarrow N$. Elements of $\Pi$ will always be
interpreted as \emph{renamings of the agents}. By contrast, $S_{n}$ refers to orderings of the agents.
Similarly, we denote by $\Gamma$ the set of all permutations of $H$,
which we interpret as \emph{renamings of the houses}.

\begin{notation}
Let $\mathbf{P}\in\mathcal{R}^{N}$, $\pi\in\Pi$ and $\tau\in\Gamma$.
We write $\left(\pi,\tau\right)\left(\mathbf{P}\right)$ for the
preference profile obtained from $\mathbf{P}$ by renaming the agents
according to $\pi$ and the houses according to $\tau$. Formally,
if $\mathbf{Q}\coloneqq\left(\pi,\tau\right)\left(\mathbf{P}\right)$,
then for each $i\in N$ and $h,h^{\prime}\in H$, 
\[
hQ_{i}h^{\prime}\iff\tau^{-1}\left(h\right)P_{\pi^{-1}\left(i\right)}\tau^{-1}\left(h^{\prime}\right).
\]
Moreover, we set $\pi\left(\mathbf{P}\right)\coloneqq\left(\pi,{\rm id}\right)\left(\mathbf{P}\right)$
and $\tau\left(\mathbf{P}\right)\coloneqq\left({\rm id},\tau\right)\left(\mathbf{P}\right)$.
\end{notation}
\begin{notation}
Let $f$ be a mechanism, $\pi\in\Pi$, and $\tau\in\Gamma$. We write
$\left(\pi,\tau\right)\left(f\right)$ for the mechanism obtained
from $f$ by renaming the agents according to $\pi$ and the houses
according to $\tau$. Formally, for every $\mathbf{P}\in\mathcal{R}^{N}$,
$h\in H$, and $i\in N$,
\[
\left(\pi,\tau\right)\left(f\right)\left(\mathbf{P}\right)_{h,i}\coloneqq f\left(\left(\pi,\tau\right)\left(\mathbf{P}\right)\right)_{\tau\left(h\right),\pi\left(i\right)}.
\]
Moreover, we set $\pi\left(f\right)\coloneqq\left(\pi,{\rm id}\right)\left(f\right)$
and $\tau\left(f\right)\coloneqq\left({\rm id},\tau\right)\left(f\right)$.
\end{notation}
We now introduce two central invariance properties of mechanisms.
\begin{defn}[Anonymity]
A mechanism $f$ is \emph{anonymous} if it is invariant under renamings of agents; that is, $\pi\left(f\right)=f$ for every $\pi\in\Pi$.
\end{defn}
\begin{defn}[Neutrality]
A mechanism $f$ is \emph{neutral} if it is invariant under renamings of houses; that is, $\tau\left(f\right)=f$ for every $\tau\in\Gamma$.
\end{defn}
\begin{rem}[Anonymity $\Rightarrow$ ETE]
If two agents submit identical preferences, then any anonymous mechanism assigns them identical probabilistic assignments. Indeed, let $f$ be anonymous and let $\mathbf{P}$ satisfy $P_i=P_j$. Let $\pi$ be the transposition of $i$ and $j$. Since $\pi\left(\mathbf{P}\right)=\mathbf{P}$ and $\pi\left(f\right)=f$, we obtain 
\[
f\left(\mathbf{P}\right)_i=\pi\left(f\right)\left(\mathbf{P}\right)_i=f\left(\pi\left(\mathbf{P}\right)\right)_{\pi\left(i\right)}=f\left(\mathbf{P}\right)_j,
\]
so agents $i$ and $j$ receive the same probabilistic assignment.
\end{rem}

\subsection{The role of efficiency in characterizing outcomes}
The following lemma, which is known in the literature (see, e.g., \cite{bogomolnaia2001new} for the balanced case where $n=m$), characterizes the set of efficient assignments
for a given preference profile. It states that every efficient assignment
arises from a serial dictatorship mechanism. Although this result
is known, we include a proof for the convenience of the reader.

\begin{lem}
Let $\mathbf{P}\in\mathcal{R}^{N}$ be a preference profile. Then
the set of all efficient assignments with respect to $\mathbf{P}$
is 
\[
\left\{ {\rm SD}_{\sigma}\left(\mathbf{P}\right)\mid\sigma\in S_{n}\right\} .
\]
\end{lem}
\begin{proof}
We first show that ${\rm SD}_{\sigma}\left(\mathbf{P}\right)$ is
an efficient assignment with respect to $\mathbf{P}$ for every $\sigma\in S_{n}$.
Suppose, for contradiction, that this is not the case. Then there
exists an assignment $s\in\mathcal{S}$ such that each agent weakly
prefers their object in $s$ to their object in ${\rm SD}_{\sigma}\left(\mathbf{P}\right)$,
and at least one agent strictly prefers their object in $s$. Denote
$A\coloneqq\left\{ i\in N\mid s\left(i\right)\neq{\rm SD}_{\sigma}\left(\mathbf{P}\right)\left(i\right)\right\} $.
Then $A$ is nonempty, and for every $i\in A$, we have $s\left(i\right)P_{i}^{+}{\rm SD}_{\sigma}\left(\mathbf{P}\right)\left(i\right)$.

Let $j\in A$ be the agent who appears first in the ordering $\sigma$
among all agents in $A$. When $j$'s turn arrives in the execution
of ${\rm SD}_{\sigma}\left(\mathbf{P}\right)$, he selects ${\rm SD}_{\sigma}\left(\mathbf{P}\right)\left(j\right)$.
Since $s\left(j\right)P_{j}^{+}{\rm SD}_{\sigma}\left(\mathbf{P}\right)\left(j\right)$,
it must be that $s\left(j\right)$ was already chosen by some earlier
agent. By the minimality of $j$ in $A$, this agent must lie in $N\setminus A$;
that is, there exists $i\in N\setminus A$ such that ${\rm SD}_{\sigma}\left(\mathbf{P}\right)\left(i\right)=s\left(j\right)$.
But then $s\left(i\right)=s\left(j\right)$, which contradicts the
feasibility of $s$, unless $s\left(j\right)=\varnothing$. However,
this contradicts the assumption that $j$ strictly prefers $s\left(j\right)$
over ${\rm SD}_{\sigma}\left(\mathbf{P}\right)\left(j\right)$. This
contradiction shows that ${\rm SD}_{\sigma}\left(\mathbf{P}\right)$
must be efficient.

To show that these are the only efficient assignments, let $s\in\mathcal{S}$
be an assignment that is efficient with respect to $\mathbf{P}$.
We will construct an ordering $\sigma\in S_{n}$ such that $s={\rm SD}_{\sigma}\left(\mathbf{P}\right)$.

We begin by showing that the efficiency of $s$ implies that at least
one agent receives his top choice under $s$. Assume, for contradiction,
that no agent receives his top choice. For each $i\in N$, let $h_{i}\in H$
denote his top choice under $\mathbf{P}$. Define a directed graph
whose vertices are the agents, and draw a directed edge from agent
$i$ to agent $j$ if and only if $s\left(j\right)=h_{i}$; that is,
if agent $j$ holds the house that agent $i$ most prefers.

Since the number of agents is finite, this directed graph is finite
and must contain either a directed cycle or a sink. In the case of
a directed cycle, the agents involved in the cycle can reassign the
houses among themselves along the cycle, so that each receives a strictly
preferred house compared to $s$, contradicting the efficiency of
$s$. In the case of a sink, there exists an agent $i\in N$ such
that $h_{i}$ is not assigned to anyone. In that case, agent $i$
could be reassigned to $h_{i}$, which he strictly prefers over $s\left(i\right)$,
without affecting the assignment of other agents, again contradicting
the efficiency of $s$. Therefore, at least one agent must receive
his top choice under $s$. 

Let $\sigma\left(1\right)$ be such an agent. Remove agent $\sigma\left(1\right)$
and the house assigned to him from the problem, and consider the induced
assignment and profile on the remaining agents and houses. Since $s$
is efficient, the restricted assignment remains efficient with respect
to the restricted profile. By the same reasoning as above, there exists
an agent among the remaining ones who receives his top choice in the
restricted problem. Let this agent be $\sigma\left(2\right)$, and
repeat the process inductively. If no houses remain, the remaining
agents can be ordered arbitrarily. In this way, we construct an ordering
$\sigma\in S_{n}$ such that $s={\rm SD}_{\sigma}\left(\mathbf{P}\right)$,
as desired.
\end{proof}
We conclude from the above lemma that any mechanism $f$ satisfying
ExPE must, for each profile $\mathbf{P}\in\mathcal{R}^{N}$, produce
an assignment matrix that lies in the convex hull of the assignment
matrices corresponding to the assignments $\left\{ {\rm SD}_{\sigma}\left(\mathbf{P}\right)\mid\sigma\in S_{n}\right\} $.
In particular, whenever a given entry is zero in all such assignment
matrices, it must be zero also in $f\left(\mathbf{P}\right)$. This
leads to the following weaker version of ExPE (which has appeared
in the literature, particularly in the context of attempts to characterize
RSD; see, e.g., \cite{brandt2023towards, fssite})
\begin{defn}
[Support efficiency]\label{def:support_efficiency}A mechanism $f$
satisfies \emph{support efficiency} if for every preference profile
$\mathbf{P}\in\mathcal{R}^{N}$, every house $h\in H$, and every
agent $i\in N$, the following holds: \\
 If ${\rm SD}_{\sigma}\left(\mathbf{P}\right)\left(i\right)\neq h$
for every ordering $\sigma\in S_{n}$, then 
\[
f\left(\mathbf{P}\right)_{h,i}=0.
\]
That is, if agent $i$ is never assigned house $h$ under the SD mechanism
with any ordering of the agents at profile $\mathbf{P}$, then the
mechanism $f$ must assign probability zero to the pair $\left(h,i\right)$
at profile $\mathbf{P}$.
\end{defn}
Alongside support efficiency, we will also use the following weak
implication of ExPE.
\begin{defn}
[Full assignment property]\textcolor{blue}{\label{def:full_assignment_property}}
A mechanism $f$ satisfies the \emph{full assignment property} if,
for every preference profile $\mathbf{P}$, the following hold:
\begin{enumerate}
\item If an agent $i$ is assigned a house in every SD mechanism at $\mathbf{P}$,
then $i$ is assigned a house with probability $1$ under $f\left(\mathbf{P}\right)$.
\item If a house $h$ is assigned to an agent in every SD mechanism at $\mathbf{P}$,
then $h$ is assigned to an agent with probability $1$ under $f\left(\mathbf{P}\right)$.
\end{enumerate}
\end{defn}
\lyxadded{bzmao}{Mon Jan  5 14:32:29 2026}{In the balanced case $m=n$,
every SD outcome assigns a house to each agent and allocates each
house. Hence, in this case, Definition \textcolor{blue}{\ref{def:full_assignment_property}}
is equivalent to requiring that, for every profile $\mathbf{P}$,
the assignment matrix $f\left(\mathbf{P}\right)$ is bi-stochastic,
that is, all row and column sums equal $1$.}
\begin{rem}
All the uniqueness proofs in Section \ref{sec:uniqueness} use only
support efficiency and the full assignment property instead of the
full ExPE axiom.
\end{rem}

\subsection{A reformulation of Strategy-Proofness}
The following lemma provides a standard local characterization of
strategy-proofness: it shows that SP is equivalent to requiring that
an agent cannot benefit from swapping two adjacent houses in their
preference. This form is well known in the literature (see, e.g.,
\cite{mennle2021partial} for the case where $n \le m$), but for the convenience of the reader,
we include a proof.
\begin{lem}
\label{lem:sp characterization}A mechanism $f$ satisfies SP if and
only if, for every preference profile $\mathbf{P}=\left(P_{1},\dots,P_{n}\right)\in\mathcal{R}^{N}$,
every agent $i\in N$, and every pair of adjacent houses $h^{\prime},h^{\prime\prime}$
in $P_{i}$, where $h^{\prime\prime}P_{i}^{+}h^{\prime}$, the following
conditions hold: 
\begin{align*}
f\left(\mathbf{P}_{-i},P_{i}^{h^{\prime},h^{\prime\prime}}\right)_{h^{\prime},i} & \geq f\left(\mathbf{P}\right)_{h^{\prime},i},\\
\forall h\in H\setminus\left\{ h^{\prime},h^{\prime\prime}\right\} :f\left(\mathbf{P}_{-i},P_{i}^{h^{\prime},h^{\prime\prime}}\right)_{h,i} & =f\left(\mathbf{P}\right)_{h,i},\\
f\left(\mathbf{P}_{-i},P_{i}^{h^{\prime},h^{\prime\prime}}\right)_{h^{\prime},i}-f\left(\mathbf{P}\right)_{h^{\prime},i} & =f\left(\mathbf{P}\right)_{h^{\prime\prime},i}-f\left(\mathbf{P}_{-i},P_{i}^{h^{\prime},h^{\prime\prime}}\right)_{h^{\prime\prime},i}.
\end{align*}
Here, $P_{i}^{h^{\prime},h^{\prime\prime}}$ denotes the preference
obtained from $P_{i}$ by swapping the adjacent houses $h^{\prime}$
and $h^{\prime\prime}$.
\end{lem}
\begin{proof}
Assume first that $f$ satisfies SP. Let $\mathbf{P}\in\mathcal{R}^{N}$,
let $i\in N$, and let $h^{\prime},h^{\prime\prime}$ be adjacent
houses in $P_{i}$, with $h^{\prime\prime}P_{i}^{+}h^{\prime}$. Denote
$\mathbf{Q}\coloneqq\left(\mathbf{P}_{-i},P_{i}^{h^{\prime},h^{\prime\prime}}\right)$,
so that $\mathbf{P}=\left(\mathbf{Q}_{-i},Q_{i}^{h^{\prime\prime},h^{\prime}}\right)$.
Then, by the definition of SP, we have: 
\[
\sum_{h\in C_{Q_{i}}\left(h^{\prime}\right)}f\left(\mathbf{Q}\right)_{h,i}\geq\sum_{h\in C_{Q_{i}}\left(h^{\prime}\right)}f\left(\mathbf{P}\right)_{h,i}.
\]
If $h^{\prime}$ is the top-ranked house in $Q_{i}$, then $C_{Q_{i}}\left(h^{\prime}\right)=\left\{ h^{\prime}\right\} $,
and the desired inequality follows immediately. Otherwise, let $h_{0}$
denote the house immediately above $h^{\prime}$ in $Q_{i}$. Then,
applying SP in the reverse direction, we obtain: 
\[
\sum_{h\in C_{P_{i}}\left(h_{0}\right)}f\left(\mathbf{P}\right)_{h,i}\geq\sum_{h\in C_{P_{i}}\left(h_{0}\right)}f\left(\mathbf{Q}\right)_{h,i}.
\]
Note that $C_{P_{i}}\left(h_{0}\right)=C_{Q_{i}}\left(h^{\prime}\right)\setminus\left\{ h^{\prime}\right\} $.
Subtracting the two inequalities above, we obtain $f\left(\mathbf{Q}\right)_{h^{\prime},i}\geq f\left(\mathbf{P}\right)_{h^{\prime},i}$
which establishes the first condition.

Now let $h\in H\setminus\left\{ h^{\prime},h^{\prime\prime}\right\} $
. By SP, we have: 
\[
\sum_{a\in C_{P_{i}}\left(h\right)}f\left(\mathbf{P}\right)_{a,i}\geq\sum_{a\in C_{P_{i}}\left(h\right)}f\left(\mathbf{Q}\right)_{a,i}\text{ and }\sum_{a\in C_{Q_{i}}\left(h\right)}f\left(\mathbf{Q}\right)_{a,i}\geq\sum_{a\in C_{Q_{i}}\left(h\right)}f\left(\mathbf{P}\right)_{a,i}.
\]
Since $C_{P_{i}}\left(h\right)=C_{Q_{i}}\left(h\right)$, all these
sums must be equal. In particular, 
\[
\sum_{a\in C_{P_{i}}\left(h\right)}f\left(\mathbf{P}\right)_{a,i}=\sum_{a\in C_{Q_{i}}\left(h\right)}f\left(\mathbf{Q}\right)_{a,i}.
\]
If $h$ is the top-ranked house, this directly implies $f\left(\mathbf{Q}\right)_{h,i}=f\left(\mathbf{P}\right)_{h,i}$.
Otherwise, assume that $h$ is not the house immediately below $h^{\prime}$
in $P_{i}$ (if such a house exists). Let $a_{h}$ denote the house
ranked directly above $h$ in $P_{i}$. Then we have: 
\[
f\left(\mathbf{Q}\right)_{h,i}=\sum_{a\in C_{Q_{i}}\left(h\right)}f\left(\mathbf{Q}\right)_{a,i}-\sum_{a\in C_{Q_{i}}\left(a_{h}\right)}f\left(\mathbf{Q}\right)_{a,i}=\sum_{a\in C_{P_{i}}\left(h\right)}f\left(\mathbf{P}\right)_{a,i}-\sum_{a\in C_{P_{i}}\left(a_{h}\right)}f\left(\mathbf{P}\right)_{a,i}=f\left(\mathbf{P}\right)_{h,i},
\]
where the equalities follow from the fact that cumulative assignment
probabilities are equal for all $h\notin\left\{ h^{\prime},h^{\prime\prime}\right\} $.
We have thus established the second condition for all relevant houses,
except the house immediately below $h^{\prime}$ in $P_{i}$, if such
a house exists.

We now move on to verify the third condition of the lemma. By SP,
we have the following two inequalities: 
\[
\sum_{a\in C_{P_{i}}\left(h^{\prime}\right)}f\left(\mathbf{P}\right)_{a,i}\geq\sum_{a\in C_{P_{i}}\left(h^{\prime}\right)}f\left(\mathbf{Q}\right)_{a,i}\text{ and }\sum_{a\in C_{Q_{i}}\left(h^{\prime\prime}\right)}f\left(\mathbf{Q}\right)_{a,i}\geq\sum_{a\in C_{Q_{i}}\left(h^{\prime\prime}\right)}f\left(\mathbf{P}\right)_{a,i}.
\]
Since $C_{P_{i}}\left(h^{\prime}\right)=C_{Q_{i}}\left(h^{\prime\prime}\right)$,
both inequalities must in fact be equalities. Moreover, 
\[
\sum_{a\in C_{P_{i}}\left(h^{\prime}\right)}f\left(\mathbf{P}\right)_{a,i}=\sum_{a\in C_{Q_{i}}\left(h^{\prime\prime}\right)}f\left(\mathbf{Q}\right)_{a,i}.
\]
If $h^{\prime\prime}$ is the top-ranked house in $P_{i}$, then rearranging
proves the third condition. Otherwise, let $a^{\prime}$ be the house
ranked immediately above $h^{\prime\prime}$ in $P_{i}$, and note
that: 
\begin{align*}
f\left(\mathbf{P}\right)_{h^{\prime\prime},i}+f\left(\mathbf{P}\right)_{h^{\prime},i} & =\sum_{a\in C_{P_{i}}\left(h^{\prime}\right)}f\left(\mathbf{P}\right)_{a,i}-\sum_{a\in C_{P_{i}}\left(a^{\prime}\right)}f\left(\mathbf{P}\right)_{a,i}\\
 & =\sum_{a\in C_{Q_{i}}\left(h^{\prime\prime}\right)}f\left(\mathbf{Q}\right)_{a,i}-\sum_{a\in C_{Q_{i}}\left(a^{\prime}\right)}f\left(\mathbf{Q}\right)_{a,i}\\
 & =f\left(\mathbf{Q}\right)_{h^{\prime},i}+f\left(\mathbf{Q}\right)_{h^{\prime\prime},i},
\end{align*}
which establishes the third condition.

Let $b$ be the house ranked immediately below $h^{\prime}$ in $P_{i}$,
if such a house exists. Then: 
\[
f\left(\mathbf{P}\right)_{b,i}=\sum_{a\in C_{P_{i}}\left(b\right)}f\left(\mathbf{P}\right)_{a,i}-\sum_{a\in C_{P_{i}}\left(h^{\prime}\right)}f\left(\mathbf{P}\right)_{a,i}=\sum_{a\in C_{Q_{i}}\left(b\right)}f\left(\mathbf{Q}\right)_{a,i}-\sum_{a\in C_{Q_{i}}\left(h^{\prime\prime}\right)}f\left(\mathbf{Q}\right)_{a,i}=f\left(\mathbf{Q}\right)_{b,i},
\]
completing the verification of the second condition for this remaining
case.

Now suppose the three conditions stated in the lemma hold for every
preference profile and every pair of adjacent houses in an agent's
ranking. We will show that $f$ satisfies SP. Let $\mathbf{P}\in\mathcal{R}^{N}$,
$i\in N$ and let $R_{i}\in\mathcal{R}$ be an alternative preference
order for agent $i$. Denote by $k$ the minimal number of adjacent
swaps required to transform $P_{i}$ into $R_{i}$. We prove by induction
on $k$ that $f\left(\mathbf{P}\right)_{i}\succeq_{P_{i}}f\left(\mathbf{P}_{-i},R_{i}\right)$.
The base case $k=0$ is immediate since $R_{i}=P_{i}$. Assume the
statement holds for all preference orders reachable from $P_{i}$
by fewer than $k$ adjacent swaps. Let $h^{\prime},h^{\prime\prime}$
be adjacent houses in $R_{i}$ such that $h^{\prime\prime}P_{i}^{+}h^{\prime}$
but $h^{\prime}R_{i}^{+}h^{\prime\prime}$, and let $R_{i}^{\prime}$
be obtained from $R_{i}$ by swapping $h^{\prime}$ and $h^{\prime\prime}$.
Since the number of swaps needed to go from $P_{i}$ to $R_{i}^{\prime}$
is less than $k$, the induction hypothesis gives $f\left(\mathbf{P}\right)_{i}\succeq_{P_{i}}f\left(\mathbf{P}_{-i},R_{i}^{\prime}\right)$.
Assume, for contradiction, that $f\left(\mathbf{P}\right)_{i}\cancel{\succeq_{P_{i}}}f\left(\mathbf{P}_{-i},R_{i}\right)$.
Then there exists some house $h_{0}\in H$ such that: 
\[
\sum_{h\in C_{P_{i}}\left(h_{0}\right)}f\left(\mathbf{P}\right)_{h,i}<\sum_{h\in C_{P_{i}}\left(h_{0}\right)}f\left(\mathbf{P}_{-i},R_{i}\right)_{h,i}.
\]
However, by the induction hypothesis, we have 
\[
\sum_{h\in C_{P_{i}}\left(h_{0}\right)}f\left(\mathbf{P}\right)_{h,i}\geq\sum_{h\in C_{P_{i}}\left(h_{0}\right)}f\left(\mathbf{P}_{-i},R_{i}^{\prime}\right)_{h,i},
\]
which implies that 
\[
\sum_{h\in C_{P_{i}}\left(h_{0}\right)}f\left(\mathbf{P}_{-i},R_{i}^{\prime}\right)_{h,i}<\sum_{h\in C_{P_{i}}\left(h_{0}\right)}f\left(\mathbf{P}_{-i},R_{i}\right)_{h,i}.
\]
By the three conditions of the lemma, this strict inequality can hold
only if $h^{\prime}\in C_{P_{i}}\left(h_{0}\right)$ and $h^{\prime\prime}\notin C_{P_{i}}\left(h_{0}\right)$,
which contradicts the fact that $h^{\prime\prime}P_{i}^{+}h^{\prime}$.
Therefore, no such $h_{0}$ exists, and $f$ satisfies SP.
\end{proof}
The next definition isolates the two invariance requirements obtained
from Lemma \ref{lem:sp characterization} by omitting the inequality;
their conjunction is therefore weaker than SP.
\begin{defn}
[Upper and lower invariance]\label{def:upper_and_lower_invariance}A
mechanism $f$ satisfies \emph{upper invariance} if, for every preference
profile $\mathbf{P}\in\mathcal{R}^{N}$, every agent $i\in N$, and
every pair of adjacent houses $h^{\prime},h^{\prime\prime}$ in $P_{i}$,
where $h^{\prime\prime}P_{i}^{+}h^{\prime}$, for every $h\in H$
such that $hP_{i}^{+}h^{\prime\prime}$ we have 
\[
f\left(\mathbf{P}_{-i},P_{i}^{h^{\prime},h^{\prime\prime}}\right)_{h,i}=f\left(\mathbf{P}\right)_{h,i}.
\]
Here, $P_{i}^{h^{\prime},h^{\prime\prime}}$ is obtained from $P_{i}$
by swapping the adjacent houses $h^{\prime}$ and $h^{\prime\prime}$,
as in Lemma \ref{lem:sp characterization}\textcolor{blue}{.}

Similarly, $f$ satisfies \emph{lower invariance} if the above equality
holds for every $h\in H$ such that $h^{\prime}P_{i}^{+}h$, and,
in addition, the probability that agent $i$ receives nothing does
not change under this swap.
\end{defn}
\begin{rem}
All the uniqueness proofs in Section \ref{sec:uniqueness} use only
upper invariance and lower invariance rather than the full SP axiom.
Moreover, in Section \ref{sec:m=2} (the case $m=2$), Proposition
\ref{claim:mEq2_characterization} remains valid when SP is replaced
by lower invariance.
\end{rem}
\begin{rem}
\label{efficiency swaps}If $f$ satisfies upper and lower invariance
and $f\left(\mathbf{P}_{-i},P_{i}^{h^{\prime},h^{\prime\prime}}\right)_{h^{\prime\prime},i}=f\left(\mathbf{P}\right)_{h^{\prime\prime},i}$,
then $f\left(\mathbf{P}_{-i},P_{i}^{h^{\prime},h^{\prime\prime}}\right)_{h^{\prime},i}=f\left(\mathbf{P}\right)_{h^{\prime},i}$
as well. Indeed, upper and lower invariance imply that agent $i$'s
probability of receiving any house other than $h^{\prime}$ or $h^{\prime\prime}$
is unchanged by the swap, and lower invariance also ensures that the
probability that agent $i$ receives nothing is unchanged. By complementing
probabilities, it follows that the total probability that agent $i$
receives either $h^{\prime}$ or $h^{\prime\prime}$ is the same under
$\mathbf{P}$ and under $\left(\mathbf{P}_{-i},P_{i}^{h^{\prime},h^{\prime\prime}}\right)$.
Combining this with the given equality for $h^{\prime\prime}$ yields
the desired equality for $h^{\prime}$. Note that when $f$ satisfies
ExPE (or even just support efficiency), we can use this fact when
$i$ cannot receive $h^{\prime\prime}$ under any efficient assignment
at $\mathbf{P}$, because in that case we would have $f\left(\mathbf{P}_{-i},P_{i}^{h^{\prime},h^{\prime\prime}}\right)_{h^{\prime\prime},i}=f\left(\mathbf{P}\right)_{h^{\prime\prime},i}=0$.
\end{rem}

\subsection{Equal Treatment for All}

A useful property that follows from the ExPE, ETE, and SP axioms concerns how a mechanism must distribute a given house among agents with similar incentives. Intuitively, if a house $h$ is assigned with positive probability to several agents, and these agents have identical binary preferences between $h$ and other houses (that is, for each other house $h^{\prime}\neq h$, they either all prefer $h^{\prime}$ to $h$ or all prefer $h$ to $h^{\prime}$), then the mechanism must treat them symmetrically with respect to $h$. That is, each of them must receive $h$ with equal probability. We refer to this implication as \emph{Equal Treatment for All (ETA)}.

\begin{prop}[ETA]
Let $f$ be a mechanism that satisfies ExPE, ETE and SP. Let $h\in H$ and $\mathbf{P}=\left(P_{1},\dots,P_{n}\right)\in\nolinebreak\mathcal{R}^{N}$. Define $I_{\mathbf{P}}\coloneqq\left\{ i\in N\mid\exists\sigma\in S_{n}:{\rm SD}_{\sigma}\left(\mathbf{P}\right)\left(i\right)=h\right\}$ . Suppose that $\mathbf{P}$ is a profile such that, under every SD mechanism, the house $h$ is assigned to some agent, and that $C_{P_{i}}\left(h\right)=C_{P_{i^{\prime}}}\left(h\right)$ for all $i,i^{\prime}\in I_{\mathbf{P}}$. Then $f\left(\mathbf{P}\right)_{h,i}=\frac{1}{\left|I_{\mathbf{P}}\right|}$ for every $i\in I_{\mathbf{P}}$.
\end{prop}
\begin{proof}
By assumption, the set $I_{\mathbf{P}}$ is non-empty. Since $f$ satisfies ExPE, it follows that $\sum_{i\in N}f\left(\mathbf{P}\right)_{h,i}=1$ and $f\left(\mathbf{P}\right)_{h,j}=0$ for all $j\in N\setminus I_{\mathbf{P}}$. Let $E_{\mathbf{P}}\subseteq I_{\mathbf{P}}$ be a largest subset of agents with identical preferences. The result is established via induction on $\left|I_{\mathbf{P}}\setminus E_{\mathbf{P}}\right|$.

The base case $E_{\mathbf{P}}=I_{\mathbf{P}}$ follows directly from the ETE property of $f$. For the inductive step, let $P\in\mathcal{R}$ denote the common preference of agents in $E_{\mathbf{P}}$. For each $i^{\prime}\in I_{\mathbf{P}}\setminus E_{\mathbf{P}}$, consider the modified profile $\mathbf{P}^{i^{\prime}}\coloneqq\left(\mathbf{P}_{-i^{\prime}},P\right)$. Assuming $I_{\mathbf{P}^{i^{\prime}}}=I_{\mathbf{P}}$, note that $E_{\mathbf{P}^{i^{\prime}}}=E_{\mathbf{P}}\cup\left\{ i^{\prime}\right\}$, so $\left|I_{\mathbf{P}^{i^{\prime}}}\setminus E_{\mathbf{P}^{i^{\prime}}}\right|<\left|I_{\mathbf{P}}\setminus E_{\mathbf{P}}\right|$. By the induction hypothesis, $f\left(\mathbf{P}^{i^{\prime}}\right)_{h,i}=\frac{1}{\left|I_{\mathbf{P}^{i^{\prime}}}\right|}=\frac{1}{\left|I_{\mathbf{P}}\right|}$ for every $i\in I_{\mathbf{P}^{i^{\prime}}}=I_{\mathbf{P}}$, and in particular for $i=i^{\prime}$. Since $P_{i^{\prime}}$ can be transformed to $P$ via adjacent swaps not involving $h$, SP implies $f\left(\mathbf{P}\right)_{h,i^{\prime}}=\frac{1}{\left|I_{\mathbf{P}}\right|}$ for each $i^{\prime}\in I_{\mathbf{P}}\setminus E_{\mathbf{P}}$. To extend the equality to all agents $i\in E_{\mathbf{P}}$, we use the two facts derived from ExPE mentioned earlier, together with the already established equality for agents in $I_{\mathbf{P}}\setminus E_{\mathbf{P}}$, and finally ETE, which together imply $f\left(\mathbf{P}\right)_{h,i}=\frac{1}{\left|I_{\mathbf{P}}\right|}$ for all $i\in E_{\mathbf{P}}$. It is left to show that $I_{\mathbf{P}^{i^{\prime}}}=I_{\mathbf{P}}$ holds for every $i^{\prime}\in I_{\mathbf{P}}\setminus E_{\mathbf{P}}$. To establish this, we prove the following lemma, which concludes the proof of the proposition.

\begin{lem}
Let $h\in H$ and let $\mathbf{P}\in\mathcal{R}^{N}$ be a profile satisfying the requirements of the proposition. Let $\mathbf{P}^{\prime}\in\mathcal{R}^{N}$ be a profile that differs from $\mathbf{P}$ only in the preferences of agents in $I_{\mathbf{P}}$ (defined as before). If $C_{P_{i}^{\prime}}\left(h\right)=C_{P_{i}}\left(h\right)$ for every agent $i\in I_{\mathbf{P}}$, then $I_{\mathbf{P}^{\prime}}=I_{\mathbf{P}}$.
\end{lem}
\begin{proof}
Let $I\coloneqq I_{\mathbf{P}}$, $J\coloneqq N\setminus I$, and $T\coloneqq C_{P_{i}}\left(h\right)\setminus\left\{ h\right\}$  for any $i\in I$ (by hypothesis, $T$ is identical for all $i\in I$). Denote $t\coloneqq\left|T\right|$. We first show that $t<\left|I\right|$.

Suppose, for contradiction, that $t\geq\left|I\right|$. Consider the assignment ${\rm SD}_{\sigma}\left(\mathbf{P}\right)$ induced by a permutation $\sigma$  that places all agents in $I$ before those in $J$. Since each agent in $I$ strictly prefers every house in $T$ to $h$ and $t\geq\left|I\right|$, each agent in $I$ selects a house from $T$. Consequently, house $h$ remains unassigned after processing $I$. By the assumption on $\mathbf{P}$, the house $h$ must be assigned to some agent in $J$. This contradicts the definition of $I$, as no agent in $J$ can receive $h$ under any SD mechanism. Thus, $t<\left|I\right|$.

We establish $I\subseteq I_{\mathbf{P}^{\prime}}$. Let $i\in I$. Consider a permutation $\sigma$  that orders all agents in $I$ in the first $\left|I\right|$ positions, with $i$ at position $t+1$. Since $C_{P_{i^{\prime}}^{\prime}}\left(h\right)=T\cup\left\{ h\right\}$  for all $i^{\prime}\in I$, under ${\rm SD}_{\sigma}\left(\mathbf{P}^{\prime}\right)$, the first $t$ agents (each in $I$) select distinct houses from $T$, exhausting $T$. Consequently, agent $i$ selects $h$. Thus, $i\in I_{\mathbf{P}^{\prime}}$.

For the converse inclusion, define $\Sigma$  as the set of permutations that assign $t$ agents from $I$ to the first $t$ positions, then all agents in $J$ to the next $\left|J\right|$ positions, and the remaining agents from $I$ to the subsequent positions. Formally,\[\Sigma\coloneqq\left\{ \sigma\in S_{n}\mid\sigma\left(J\right)=\left[t+\left|J\right|\right]\setminus\left[t\right]\right\} .\]

Note that for every $\sigma\in\Sigma$  and every $j\in J$, $\left({\rm SD}_{\sigma}\left(\mathbf{P}\right)\left(j\right)\right)P_{j}^{+}h$. This is because, under such a permutation, the first $t$ positions are occupied by agents from $I$, who select all houses in $T$, and the agents in $J$ that come before $j$ do not take $h$. Thus, $h$ is available when it is $j$'s turn, but since $j$ does not receive $h$, it must be that $j$ chooses a house that he strictly prefers to $h$.

We demonstrate that the same holds for $\mathbf{P}^{\prime}$ (i.e., $\left({\rm SD}_{\sigma}\left(\mathbf{P}^{\prime}\right)\left(j\right)\right)P_{j}^{\prime+}h$). Given that $P_{j}^{\prime}=P_{j}$, it suffices to establish ${\rm SD}_{\sigma}\left(\mathbf{P}^{\prime}\right)\left(j\right)={\rm SD}_{\sigma}\left(\mathbf{P}\right)\left(j\right)$. This equality follows because, in both profiles, the first $t$ agents from $I$ select all houses in $T$. By induction on the ordering of agents in $J$, we conclude that each agent in $J$ chooses the same house under $\mathbf{P^{\prime}}$ as under $\mathbf{P}$.

Let $\tau\in S_{n}$. We aim to demonstrate that ${\rm SD}_{\tau}\left(\mathbf{P}^{\prime}\right)\left(j\right)\neq h$ for all $j\in J$. Let $\sigma_{\tau}$ be the permutation in $\Sigma$  preserving the relative order of agents within $I$ and within $J$ from $\tau$. Specifically, for any $i,i^{\prime}\in I$, agent $i$ precedes agent $i^{\prime}$ in $\sigma_{\tau}$ if and only if the same holds in $\tau$, and analogously for agents in $J$.

For every agent $j\in J$, define $h_{j}\coloneqq{\rm SD}_{\sigma_{\tau}}\left(\mathbf{P}^{\prime}\right)\left(j\right)$. As established previously, $h_{j}P_{j}^{\prime+}h$. Moreover, since $\sigma_{\tau}\in\Sigma$, it follows that $h_{j}\notin T$ because the first $t$ agents from $I$ exhaust $T$.

Assume by way of contradiction that there exists an agent $j_{0}\in J$ such that ${\rm SD}_{\tau}\left(\mathbf{P}^{\prime}\right)\left(j_{0}\right)=h$. Let $j^{\prime}\in J$ be the first agent in the ordering $\tau$ for which, during the execution of ${\rm SD}_{\tau}\left(\mathbf{P^{\prime}}\right)$, the house $h_{j^{\prime}}$ was unavailable for selection by $j^{\prime}$, but the house $h$ was available. Such an agent exists because the condition holds for $j_{0}$.

By the definition of $j^{\prime}$, for every agent $i\in I$ preceding $j^{\prime}$ with respect to $\tau$, we have ${\rm SD}_{\tau}\left(\mathbf{P^{\prime}}\right)\left(i\right)\in T$ (because otherwise one such agent would have selected $h$, contradicting $h$'s availability for $j^{\prime}$) and thus ${\rm SD}_{\tau}\left(\mathbf{P}^{\prime}\right)\left(i\right)\neq h_{j^{\prime}}$ for every such $i$. Since $h_{j^{\prime}}$ is unavailable for selection by $j^{\prime}$ during the execution of ${\rm SD}_{\tau}\left(\mathbf{P^{\prime}}\right)$, there exists an agent $j^{\prime\prime}\in J$ preceding $j^{\prime}$ such that ${\rm SD}_{\tau}\left(\mathbf{P^{\prime}}\right)\left(j^{\prime\prime}\right)=h_{j^{\prime}}$. Note that $h$ is available for $j^{\prime\prime}$ during ${\rm SD}_{\tau}\left(\mathbf{P^{\prime}}\right)$ (as $j^{\prime\prime}$ precedes $j^{\prime}$). Because $\sigma_{\tau}$ preserves the relative order of agents within $J$ from $\tau$, $j^{\prime\prime}$ precedes $j^{\prime}$ in $\sigma_{\tau}$, implying $h_{j^{\prime\prime}}P_{j^{\prime\prime}}^{\prime+}h_{j^{\prime}}$ (since both houses were available to $j^{\prime\prime}$ during ${\rm SD}_{\sigma_{\tau}}\left(\mathbf{P^{\prime}}\right)$ and he selected $h_{j^{\prime\prime}}$). However, during ${\rm SD}_{\tau}\left(\mathbf{P^{\prime}}\right)$, $j^{\prime\prime}$ selected $h_{j^{\prime}}$, indicating $h_{j^{\prime\prime}}$ was unavailable to him - contradicting the minimality of $j^{\prime}$ in $\tau$ . Consequently, ${\rm SD}_{\tau}\left(\mathbf{P^{\prime}}\right)\left(j\right)\neq h$ for all $j\in J$ and $\tau\in S_{n}$, which implies $I_{\mathbf{P^{\prime}}}\subseteq I$, concluding the proof.
\end{proof}
\end{proof}

\subsection{Near-unanimity forces uniqueness}
To better understand the implications of the axioms ExPE, ETE, and
SP, we analyze a structured class of preference profiles in which
agents exhibit near-unanimous agreement on the relative ranking of
every pair of houses. In such highly structured settings, the combined
force of the axioms leaves essentially no room for flexibility: the
outcome is uniquely pinned down. The following lemma formalizes this
observation by proving uniqueness of the assignment matrix under these
conditions.
\begin{lem}\label{the lemma}
Assume $n\geq\max\left\{ m,4\right\} $. Let $\mathbf{P}\in\mathcal{R}^{N}$
be a preference profile such that for every pair of distinct houses
$\left\{ h_{1},h_{2}\right\} \subseteq H$, all but possibly one agent
prefer the same house between $h_{1}$ and $h_{2}$. Then, the assignment
matrix of $\mathbf{P}$ is uniquely determined across all mechanisms
satisfying ExPE, ETE, and SP.
\end{lem}
\begin{proof}
We prove the lemma by induction on the number of unordered pairs $\left\{ h_{1},h_{2}\right\} \subseteq H$
for which the agents do not unanimously prefer one house over the
other. In the base case, all agents agree on the binary preference
between every pair of houses. In other words, they have identical
rankings. Since $n\geq m$, ExPE ensures that each house is fully
allocated, and ETE requires that each house is shared equally among
the agents. Hence, the assignment matrix is uniquely determined.

For the induction step, let $\mathbf{P}\in\mathcal{R}^{N}$ be a preference
profile with $d>0$ unordered pairs of houses for which the agents
do not unanimously prefer one house over the other. Let these $d$
house pairs be $\left\{ h_{1},h_{1}^{\prime}\right\} ,\dots,\left\{ h_{d},h_{d}^{\prime}\right\} $.
For each $\ell\in\left[d\right]$, let $a_{\ell}\in N$ be the agent
whose binary preference over the pair $\left\{ h_{\ell},h_{\ell}^{\prime}\right\} $
differs from that of all other agents. Note that even if $\ell_{1}\neq\ell_{2}$,
it is possible that $a_{\ell_{1}}$ and $a_{\ell_{2}}$ refer to the
same agent.

We first claim that for every $\ell\in\left[d\right]$, there exists
a pair of adjacent houses in the ranking of agent $a_{\ell}$ such
that this agent prefers one house over the other, while all other
agents prefer the other house in that pair. To see this, consider
a pair $\left\{ h,h^{\prime}\right\} $ for which $a_{\ell}$ is the
agent whose binary preference differs from the rest. If $h$ and $h^{\prime}$
are not adjacent in the ranking of $a_{\ell}$, then there exists
a house $h^{\prime\prime}$ that appears between them in that ranking.
In that case, $a_{\ell}$ must disagree with every other agent on
the binary preference of at least one of the pairs $\left\{ h,h^{\prime\prime}\right\} $
or $\left\{ h^{\prime},h^{\prime\prime}\right\} $. Since $n\geq4$,
for at least one of these two pairs there must be
two agents whose binary preference differs from that of $a_{\ell}$.
Without loss of generality, suppose this pair is $\left\{ h,h^{\prime\prime}\right\} $.
Since at most one agent can disagree with the rest on any given pair,
it follows that all other agents share the same preference over $\left\{ h,h^{\prime\prime}\right\} $,
and that this preference differs from that of $a_{\ell}$. Repeating
this process eventually leads to a pair of adjacent houses in the
ranking of $a_{\ell}$, for which his binary preference over that
pair differs from that of all the other agents.

Let $\left\{ h,h^{\prime}\right\} $ be such a pair, and let $\mathbf{P}_{\ell}$
be the profile obtained by swapping the positions of $h$ and $h^{\prime}$
in the ranking of agent $a_{\ell}$. In this modified profile, there
are only $d-1$ house pairs for which there is no consensus among
the agents regarding their binary ranking. By the induction hypothesis,
the assignment matrix for $\mathbf{P}_{\ell}$ is uniquely determined
by the axioms.

By SP, the probabilities that $a_{\ell}$ receives any house $h^{\prime\prime}\notin\left\{ h,h^{\prime}\right\} $
must remain the same in both $\mathbf{P}$ and $\mathbf{P}_{\ell}$.
Since $a_{\ell}$ is the only agent whose preference over the pair
$\left\{ h,h^{\prime}\right\} $ differs from the others, efficiency
implies that under the profile $\mathbf{P}$, he cannot receive the
house he ranks lower between the two. It then follows from SP that
his probability of receiving the house he ranks higher is also determined.
Therefore, the probabilities assigned to each agent in the set $\left\{ a_{1},\dots,a_{d}\right\} $
are uniquely determined by the axioms in the profile $\mathbf{P}$.

Now, let $N^{\prime}\coloneqq N\setminus\left\{ a_{1},\dots,a_{d}\right\} $.
By definition of the agents $a_{1},\dots,a_{d}$, all agents in $N^{\prime}$
share the same ranking in $\mathbf{P}$. This set may be empty, in
which case the proof is complete. Otherwise, since $n\geq m$, the
axiom ExPE guarantees that each house is fully allocated, and by ETE,
all agents in $N^{\prime}$ must receive the same probability for
each house. Therefore, the assignment probabilities for all agents
in $N^{\prime}$ are uniquely determined in the profile $\mathbf{P}$,
which completes the proof.
\end{proof}
\section{Characterization for \texorpdfstring{$m=2$}{m=2}} \label{sec:m=2}
In the following section, we characterize the mechanisms satisfying
ExPE, ETE, and SP for $n\geq2$ and $m=2$.

Denote the set of houses by $H\coloneqq\left\{ a,b\right\} $. For
each preference profile $\mathbf{Q}\in\mathcal{R}^{N}$ and house
$h\in H$, let $N_{\mathbf{Q},h}\subseteq N$ be the subset of agents
who rank $h$ first. Define $n_{\mathbf{Q},h}\coloneqq\left|N_{\mathbf{Q},h}\right|$
to be the number of such agents.

We begin by establishing key facts that will allow us to interpret
a mechanism satisfying the axioms as a function. By ExPE, each house
is allocated with probability $1$.

Furthermore, using induction on $n_{\mathbf{Q},a}$, we establish
that each agent is assigned a house with probability $\frac{2}{n}$.

First, as each of the two houses is fully allocated, the total probability
of being assigned a house across all agents is $2$. Second, when
all agents have identical preferences, ETE ensures that they must
all receive the same probability of being assigned, which must therefore
be $\frac{2}{n}$. This establishes the base case $n_{\mathbf{Q},a}=0$,
and the argument applies equally to the case $n_{\mathbf{Q},a}=n$.
Now consider the case where $0<n_{\mathbf{Q},a}<n$. Let $i\in N_{\mathbf{Q},a}$,
and let $\mathbf{Q}^{\prime}$ be the profile in which agent $i$
ranks $b$ above $a$, while the preferences of all other agents remain
as in $\mathbf{Q}$. Note that $n_{\mathbf{Q}^{\prime},a}<n_{\mathbf{Q},a}$,
so by the induction hypothesis, every agent in $\mathbf{Q}^{\prime}$
is assigned with probability $\frac{2}{n}$. Then, by SP, agent $i$
must also be assigned with probability $\frac{2}{n}$ in the original
profile $\mathbf{Q}$. Since there are only two types of preferences
in $\mathbf{Q}$, and the total sum of the agents' assignment probabilities
is $2$, ETE and the complementarity of probabilities together imply
that all agents in $\mathbf{Q}$ are assigned with probability $\frac{2}{n}$.

Consequently, by ETE and the complementarity of probabilities, it
suffices to determine the probability that an agent in $N_{\mathbf{Q},a}$
receives house $a$. In other words, any mechanism satisfying the
axioms is uniquely determined once this probability is specified.

Hence, we consider functions of the form $\varphi:2^{N}\setminus\left\{ \varnothing\right\} \rightarrow\left[0,\frac{2}{n}\right]$,
and define the corresponding mechanism $f\left(\varphi\right)$ to
be the mechanism that satisfies ETE and, for every $\mathbf{Q}\in\mathcal{R}^{N}$,
assigns each agent in $N_{\mathbf{Q},a}$ to house $a$ with probability
$\varphi\left(N_{\mathbf{Q},a}\right)$, allocates each house with
probability $1$, and assigns each agent with probability $\frac{2}{n}$.
The empty set is excluded from the domain of $\varphi$, since assigning
agents in $\varnothing$ to $a$ with any probability is meaningless.
Moreover, if $\varphi_{1}\neq\varphi_{2}$, then $f\left(\varphi_{1}\right)\neq f\left(\varphi_{2}\right)$,
ensuring that each such function induces a different mechanism.
\begin{prop}
\label{claim:mEq2_characterization}Let $n\geq2$ and $m=2$. The
correspondence described above defines a bijection between the set
of mechanisms satisfying ExPE, ETE and SP, and the set of functions:
\[
\mathcal{F}\coloneqq\left\{ \varphi:2^{N}\setminus\left\{ \varnothing\right\} \rightarrow\left[\frac{1}{n},\frac{2}{n}\right]\mid\begin{array}{c}
\forall R\in2^{N}\setminus\left\{ \varnothing\right\} \mathrel{:}\varphi\left(R\right)\leq\frac{1}{\left|R\right|}\\
\forall i\in N:\varphi\left(\left\{ i\right\} \right)=\frac{2}{n},\,\varphi\left(N\setminus\left\{ i\right\} \right)=\frac{1}{n-1}
\end{array}\right\} 
\]
\end{prop}
\begin{proof}
We first show that any mechanism $f$ satisfying the axioms must correspond
to a function in $\mathcal{F}$. First, observe that when $n_{\mathbf{Q},a}\in\left\{ 0,1,n-1,n\right\} $,
the probabilities in the profile are fully determined by ExPE and
ETE. Specifically, this corresponds to 
\[
\varphi\left(\left\{ i\right\} \right)=\frac{2}{n},\ \varphi\left(N\setminus\left\{ i\right\} \right)=\frac{1}{n-1},\ {\rm and}\ \varphi\left(N\right)=\frac{1}{n}.
\]
The last condition is implicitly satisfied in the definition of $\mathcal{F}$
since the range of the functions is constrained to $\left[\frac{1}{n},\frac{2}{n}\right]$,
so since $\varphi\left(R\right)\leq\frac{1}{\left|R\right|}$ we obtain
\[
\frac{1}{n}\leq\varphi\left(N\right)\leq\frac{1}{\left|N\right|}=\frac{1}{n}.
\]
Thus, the constraint holds automatically.

Now, since $\varphi\left(R\right)$ represents the probability that
an agent $i\in R=N_{\mathbf{Q},a}$ receives house $a$, we obtain
the constraint $\left|R\right|\varphi\left(R\right)\leq1$. Rearranging,
this yields $\varphi\left(R\right)\leq\frac{1}{\left|R\right|}$.

It remains to show that $\varphi\left(R\right)\geq\frac{1}{n}$ for
any $R\in2^{N}\setminus\left\{ \varnothing\right\} $. Let $R\in2^{N}\setminus\left\{ \varnothing\right\} $
and denote by $\mathbf{Q}\in\mathcal{R}^{N}$ the corresponding profile
such that $R=N_{\mathbf{Q},a}$. When the profile is $\mathbf{Q}$,
an agent in $R$ receives $a$ with probability $\varphi\left(R\right)$
and $b$ with probability $\frac{2}{n}-\varphi\left(R\right)$. Thus,
the probability that some agent in $R$ receives $a$ is $\left|R\right|\varphi\left(R\right)$
while the probability that some agent in $R$ receives $b$ is $\left|R\right|\left(\frac{2}{n}-\varphi\left(R\right)\right)$.
Since the mechanism satisfies ExPE, it can be expressed as a convex
combination of deterministic serial dictatorship mechanisms. In these
mechanisms, whenever an agent in $R=N_{\mathbf{Q},a}$ receives $b$,
another agent in $R$ must receive $a$. This implies that the event
of some agent in $R$ receiving $b$ (in the mechanism $f$ when the
profile is $\mathbf{Q}$) is contained in the event of some agent
in $R$ receiving $a$, which yields 
\begin{align*}
\left|R\right|\left(\frac{2}{n}-\varphi\left(R\right)\right) & \leq\left|R\right|\varphi\left(R\right)\\
\implies\varphi\left(R\right) & \geq\frac{1}{n}.
\end{align*}
Thus, we conclude that any mechanism satisfying the axioms corresponds
to a function in $\mathcal{F}$.

Conversely, let $\varphi\in\mathcal{F}$. We will show that the corresponding
mechanism $f\left(\varphi\right)$ satisfies the axioms. Since the
correspondence was defined to preserve equal treatment among agents
with identical preferences, it follows that $f\left(\varphi\right)$
satisfies ETE.

Now, we will show that $f\left(\varphi\right)$ satisfies SP in two
steps: first, we will show that if $\varphi\in\mathcal{F}$ satisfies
a certain constraint, then $f\left(\varphi\right)$ satisfies SP.
Next, we will show that every $\varphi\in\mathcal{F}$ satisfies this
constraint.

As a first step, for every $R\subsetneq N$, define 
\[
g_{\varphi}\left(R\right)\coloneqq\frac{1-\left|R\right|\varphi\left(R\right)}{n-\left|R\right|},
\]
with the convention that $\left|R\right|\varphi\left(R\right)\coloneqq0$
when $R=\varnothing$. Note that since $\varphi\left(R\right)$ represents
the probability that an agent who prefers $a$ gets $a$, $g_{\varphi}\left(R\right)$
represents the probability that an agent who prefers $b$ gets $a$.
We will show that if $\varphi$ satisfies the constraint $\varphi\left(R\right)\geq g_{\varphi}\left(R\setminus\left\{ i\right\} \right)$
for every nonempty $R\subseteq N$ and for every $i\in R$, then $f\left(\varphi\right)$
satisfies SP.

Indeed, let $i\in N$. We will show that $i$ cannot gain from manipulation.
Since there are only two houses, and every agent receives some house
with a fixed probability, it suffices to show that $i$ gets $a$
with (weakly) higher probability when he ranks $a$ first. Let $\mathbf{Q}\in\mathcal{R}^{N}$.
If $i\in N_{\mathbf{Q},a}$, then we show that the probability of
$i$ getting $a$ weakly decreases if he changes his preference. In
the profile $\mathbf{Q}$, the probability that $i$ gets $a$ is
$\varphi\left(N_{\mathbf{Q},a}\right)$, and if he changes his preference,
his probability of getting $a$ becomes $g_{\varphi}\left(N_{\mathbf{Q},a}\setminus\left\{ i\right\} \right)$,
which is weakly lower by the hypothesis, as required. Similarly, if
$i\notin N_{\mathbf{Q},a}$ then the probability of $i$ getting $a$
is $g_{\varphi}\left(N_{\mathbf{Q},a}\right)$, and if he changes
his preference, it becomes $\varphi\left(N_{\mathbf{Q},a}\cup\left\{ i\right\} \right)$,
which is weakly higher, as required. Thus, under the aforementioned
constraint on $\varphi$, $f\left(\varphi\right)$ satisfies SP, as
desired.

For the second step, we will show that every $\varphi\in\mathcal{F}$
satisfies this constraint. Indeed, let $R\subseteq N$ be nonempty
and let $i\in R$. Since $\varphi\left(R\right)\geq\frac{1}{n}$ for
any nonempty $R\subseteq N$, it suffices to show that $g_{\varphi}\left(R\setminus\left\{ i\right\} \right)\leq\frac{1}{n}$.
If $R=\left\{ i\right\} $ then $g_{\varphi}\left(R\setminus\left\{ i\right\} \right)=g_{\varphi}\left(\varnothing\right)=\frac{1}{n}$.
Otherwise we have: 
\begin{align*}
g_{\varphi}\left(R\setminus\left\{ i\right\} \right) & =\frac{1-\left|R\setminus\left\{ i\right\} \right|\varphi\left(R\setminus\left\{ i\right\} \right)}{n-\left|R\setminus\left\{ i\right\} \right|}\leq\frac{1-\left|R\setminus\left\{ i\right\} \right|\frac{1}{n}}{n-\left|R\setminus\left\{ i\right\} \right|}=\frac{1}{n}
\end{align*}
Thus, the constraint is satisfied. In conclusion, for every $\varphi\in\mathcal{F}$,
$f\left(\varphi\right)$ satisfies SP.

Regarding ExPE, we will show that for every $\varphi\in\mathcal{F}$,
$f\left(\varphi\right)$ coincides with $f_{M}$ for some extensive-form
mechanism $M$ that can be expressed as a convex combination of the
efficient mechanisms $\left\{ {\rm SD}_{\sigma}\right\} _{\sigma\in S_{n}}$.
First, note that since there are only two houses, the only aspect
of the order of the agents that matters is the identity of the first
two agents to arrive. Considering this, for $h_{1},h_{2}\in H$, we
define: 
\[
\Sigma_{h_{1},h_{2}}^{\mathbf{Q}}\coloneqq\left\{ \sigma\in S_{n}\mid\forall i\in\left[2\right]:\sigma\left(i\right)\in N_{\mathbf{Q},h_{i}}\right\} .
\]
Using this notation, we define $M$ as follows: 
\[
M\left(\mathbf{Q}\right)\coloneqq\sum_{\sigma\in S_{n}}w_{\sigma,\varphi}^{\mathbf{Q}}{\rm SD}_{\sigma}\left(\mathbf{Q}\right),
\]
where: 
\begin{align*}
w_{\sigma,\varphi}^{\mathbf{Q}}\coloneqq\begin{cases}
\frac{\frac{2}{n}-\varphi\left(N_{\mathbf{Q},a}\right)}{\left(n_{\mathbf{Q},a}-1\right)\left(n-2\right)!} & \sigma\in\Sigma_{a,a}^{\mathbf{Q}}\\
\frac{g_{\varphi}\left(N_{\mathbf{Q},a}\right)}{\left(n_{\mathbf{Q},b}-1\right)\left(n-2\right)!} & \sigma\in\Sigma_{b,b}^{\mathbf{Q}}\\
\frac{2}{n_{\mathbf{Q},b}\left(n-2\right)!}\left(\varphi\left(N_{\mathbf{Q},a}\right)-\frac{1}{n}\right) & \sigma\in\Sigma_{a,b}^{\mathbf{Q}}\\
0 & {\rm otherwise}
\end{cases}
\end{align*}

First, we will show that the weights are well-defined and non-negative.
If $\Sigma_{a,a}^{\mathbf{Q}}\neq\varnothing$, then $n_{\mathbf{Q},a}\geq2$,
so we have $\frac{2}{n}-\varphi\left(N_{\mathbf{Q},a}\right)\geq0$
and $n_{\mathbf{Q},a}-1\geq1$. Similarly, if $\Sigma_{b,b}^{\mathbf{Q}}\neq\varnothing$
we have that $n_{\mathbf{Q},b}-1\geq1$ and $g_{\varphi}\left(N_{\mathbf{Q},a}\right)=\frac{1-n_{\mathbf{Q},a}\varphi\left(N_{\mathbf{Q},a}\right)}{n_{\mathbf{Q},b}}$,
which is non-negative. If $\Sigma_{a,b}^{\mathbf{Q}}\neq\varnothing$
then $N_{\mathbf{Q},a}\notin\left\{ \varnothing,N\right\} $ , so
$n_{\mathbf{Q},b}\geq1$ and $\varphi\left(N_{\mathbf{Q},a}\right)\geq\frac{1}{n}$.
Next, we will show that the sum of the weights equals $1$. We note
that $\left|\Sigma_{a,a}^{\mathbf{Q}}\right|=n_{\mathbf{Q},a}\left(n_{\mathbf{Q},a}-1\right)\left(n-2\right)!$,
$\left|\Sigma_{b,b}^{\mathbf{Q}}\right|=n_{\mathbf{Q},b}\left(n_{\mathbf{Q},b}-1\right)\left(n-2\right)!$
and $\left|\Sigma_{a,b}^{\mathbf{Q}}\right|=n_{\mathbf{Q},a}n_{\mathbf{Q},b}\left(n-2\right)!$.
Therefore, for every $\mathbf{Q}\in\mathcal{R}^{N}$, the sum of the
weights is given by: 
\begin{align*}
\sum_{\sigma\in S_{n}}w_{\sigma,\varphi}^{\mathbf{Q}} & =\sum_{\sigma\in\Sigma_{a,a}^{\mathbf{Q}}}w_{\sigma,\varphi}^{\mathbf{Q}}+\sum_{\sigma\in\Sigma_{b,b}^{\mathbf{Q}}}w_{\sigma,\varphi}^{\mathbf{Q}}+\sum_{\sigma\in\Sigma_{a,b}^{\mathbf{Q}}}w_{\sigma,\varphi}^{\mathbf{Q}}\\
 & =n_{Q,a}\left(n_{Q,a}-1\right)\left(n-2\right)!\frac{\frac{2}{n}-\varphi\left(N_{Q,a}\right)}{\left(n_{Q,a}-1\right)\left(n-2\right)!}\\
 & \qquad+n_{Q,b}\left(n_{Q,b}-1\right)\left(n-2\right)!\frac{g_{\varphi}\left(N_{Q,a}\right)}{\left(n_{Q,b}-1\right)\left(n-2\right)!}\\
 & \qquad+n_{Q,a}n_{Q,b}\left(n-2\right)!\cdot\frac{2}{n_{Q,b}\left(n-2\right)!}\left(\varphi\left(N_{Q,a}\right)-\frac{1}{n}\right)\\
 & =\frac{2n_{\mathbf{Q},a}}{n}-n_{\mathbf{Q},a}\varphi\left(N_{\mathbf{Q},a}\right)+n_{\mathbf{Q},b}\cdot\frac{1-n_{\mathbf{Q},a}\varphi\left(N_{\mathbf{Q},a}\right)}{n-n_{\mathbf{Q},a}}+2n_{\mathbf{Q},a}\varphi\left(N_{\mathbf{Q},a}\right)-\frac{2n_{\mathbf{Q},a}}{n}\\
 & =1
\end{align*}

Now, since $M$ is a convex combination of $\left\{ {\rm SD}_{\sigma}\right\} _{\sigma\in S_{n}}$,
$f_{M}$ satisfies ExPE and SP (as a convex combination of mechanisms
satisfying ExPE and SP).

Furthermore, ETE holds because of the symmetry in how the weights
are defined. Therefore, by using the fact that each agent is assigned
a house with probability $\frac{2}{n}$, along with ETE and complementing
to $1$, in order to prove that $f\left(\varphi\right)=f_{M}$, it
suffices to prove that $f_{M}$ assigns the same probability as $f\left(\varphi\right)$
for an agent in $N_{\mathbf{Q},a}$ to receive house $a$, i.e., the
probability $\varphi\left(N_{\mathbf{Q},a}\right)$.

Indeed, for any $\mathbf{Q}\in\mathcal{R}^{N}$, an agent $i\in N_{\mathbf{Q},a}$
gets $a$ in the following cases: when the order is in $\Sigma_{a,a}^{\mathbf{Q}}$
and agent $i$ is first, with $n_{\mathbf{Q},a}-1$ options for the
second agent from $N_{\mathbf{Q},a}\setminus\left\{ i\right\} $ and
$\left(n-2\right)!$ options for the arrangement of the other agents.
Additionally, agent $i$ also gets the house $a$ when the order is
in $\Sigma_{a,b}^{\mathbf{Q}}$ and agent $i$ is selected from $N_{\mathbf{Q},a}$,
with $n_{\mathbf{Q},b}$ possible choices for the second agent from
$N_{\mathbf{Q},b}$, and $\left(n-2\right)!$ options for the arrangement
of the other agents. Thus, the probability that the agent gets $a$
is: 
\begin{align*}
\left(n_{\mathbf{Q},a}-1\right)\left(n-2\right)!\cdot\frac{\frac{2}{n}-\varphi\left(N_{\mathbf{Q},a}\right)}{\left(n_{\mathbf{Q},a}-1\right)\left(n-2\right)!}+n_{\mathbf{Q},b}\left(n-2\right)!\cdot\frac{2}{n_{\mathbf{Q},b}\left(n-2\right)!}\left(\varphi\left(N_{\mathbf{Q},a}\right)-\frac{1}{n}\right)=\varphi\left(N_{\mathbf{Q},a}\right)
\end{align*}
Therefore, the equality $f\left(\varphi\right)=f_{M}$ holds, and
$f\left(\varphi\right)$ satisfies ExPE as well.
\end{proof}
\begin{rem}
\label{rem:m=2} For $n=2,3$ there is a unique function $\varphi\in\mathcal{F}$,
hence a unique mechanism satisfying the axioms. However, when $n\geq4$,
for every subset $R\subseteq N$ of agents of size $2\leq\left|R\right|\leq n-2$,
we obtain an interval of positive length for the options for $\varphi\left(R\right)$.
Therefore, the dimension of $\mathcal{F}$ is given by $\sum_{i=2}^{n-2}\binom{n}{i}=2^{n}-2n-2$.
Moreover, if we were to strengthen the axiom of ETE to anonymity,
the dimension would decrease to $\sum_{i=2}^{n-2}1=n-3$, as there
would be only one degree of freedom for each subset size. If we also
required neutrality, the dimension would drop further to $\left\lceil \frac{n-3}{2}\right\rceil $,
since the value of $\varphi$ on sets larger than $\frac{n}{2}$ would
be determined by its value on their complements.
\end{rem}
\begin{rem}
\lyxadded{bzmao}{Mon Jan  5 18:01:36 2026}{In the proof of Proposition
\textcolor{blue}{\ref{claim:mEq2_characterization}}, the first part
relies only on lower invariance, whereas the second part establishes
SP. Consequently, Proposition \textcolor{blue}{\ref{claim:mEq2_characterization}}
remains valid if SP is replaced by lower invariance.}
\end{rem}
\begin{rem}
\label{rem:mEq2_nGe4} The function corresponding to the RSD mechanism
is 
\[
\varphi_{{\rm RSD}}\left(R\right)=\frac{1}{n}+\frac{n-\left|R\right|}{\left(n-1\right)n}.
\]
The first term, $\frac{1}{n}$, accounts for the permutations in which
the agent appears first and thus selects their top choice. The second
term accounts for the permutations where the agent appears second,
following an agent who prefers the house $b$, thereby still allowing
the agent to choose house $a$.

Moreover, we observe that when $n\geq4$, RSD is dominated by the
mechanism corresponding to the function 
\[
\varphi\left(R\right)=\min\left\{ \frac{2}{n},\frac{1}{\left|R\right|}\right\} ,
\]
which assigns to every agent their preferred house with the maximum
probability allowed under the axioms.
\end{rem}
\begin{rem}
\label{rem:mEq2_domination_generalization}\lyxadded{bzmao}{Sun Jan  4 16:47:01 2026}{The
domination phenomenon observed above for the case $m=2$ with $n\geq4$
extends to markets with $m\geq2$ and $n\geq m+2$ via an extension-and-symmetrization
argument. Fix such a pair $\left(n,m\right)$, and set $\alpha\coloneqq m-2$.
Starting from a mechanism $f$ on the $\left(n-\alpha,2\right)$ market
that dominates $f_{{\rm RSD}}$ (for example, the mechanism corresponding
to the function $\varphi$ established in the remark above), we extend
it to the $\left(n,m\right)$ market by adding $\alpha$ auxiliary
agents and $\alpha$ new houses. The auxiliary agents are ordered
first, and each selects their most preferred remaining house. After
these $\alpha$ steps, if the set of remaining houses coincides exactly
with the original set of $2$ houses (equivalently, the auxiliary
agents have selected exactly the $\alpha$ added houses), we apply
$f$ to the $n-\alpha$ original agents with their preferences restricted
to these $2$ houses; otherwise, we apply $f_{{\rm RSD}}$ to the
$n-\alpha$ original agents with preferences restricted to the remaining
houses. Finally, we symmetrize the resulting rule by uniformly averaging
over all renamings of the agents, which restores anonymity (and hence
ETE) while preserving the other axioms. This construction yields a
mechanism satisfying the axioms that dominates $f_{{\rm RSD}}$ in
the $\left(n,m\right)$ market. We omit the details.}
\end{rem}

\section{Cases where there is uniqueness} \label{sec:uniqueness}
\lyxdeleted{bzmao}{Sat Jan 10 23:03:12 2026}{To better understand
the implications of the axioms ExPE, ETE, and SP, this section identifies
the values of \mbox{$m$} and \mbox{$n$} for which these axioms jointly}\lyxadded{bzmao}{Sat Jan 10 23:03:48 2026}{This
section identifies the values of $m$ and $n$ for which the axioms
ExPE, ETE, and SP jointly} determine a unique normal-form mechanism.
That is, for such values of $m$ and $n$, the axioms are sufficiently
strong to uniquely determine the assignment matrix for every possible
preference profile.
\begin{rem}
Throughout this section, ExPE is used only via the weaker conditions
of support efficiency (Definition \ref{def:support_efficiency}) and
the full assignment property (Definition \ref{def:full_assignment_property}),
and SP is used only via upper and lower invariance (Definition \ref{def:upper_and_lower_invariance}).
The reader may verify that the arguments rely only on these weaker
conditions.
\end{rem}
\begin{rem}
\label{rem:nEq2_mGt2}When $n=2$ and $m>2$, the axioms uniquely
determine the mechanism. Indeed, ExPE implies that each agent receives
one of his top two choices with probability $1$, so it suffices to
determine the probability that each agent receives his top choice.
If the two agents share the same top choice, then ETA implies that
each receives their common top choice with probability $\frac{1}{2}$.
Otherwise, the unique efficient deterministic assignment assigns each
agent his top choice, and hence both agents receive their respective
top choices with probability $1$.
\end{rem}

\subsection{\texorpdfstring{$n=3$ and $m>3$}{n=3 and m>3}} \label{subsec:n=3 and m>=3}
We focus here on the case
$n=3$ and $m>3$. In the balanced case $n=m=3$, the result was already
established by Bogomolnaia and Moulin \lyxadded{bzmao}{Fri Jan  9 22:44:37 2026}{\cite{bogomolnaia2001new}}\lyxadded{bzmao}{Fri Jan  9 22:53:29 2026}{;
nevertheless, the proof below applies verbatim to $m=3$ and thus
also yields an alternative proof of their result.}
\begin{prop}
For $n=3$ and $m\geq3$, the axioms ExPE, ETE, and SP uniquely characterize
$f_{{\rm RSD}}$.
\end{prop}
\begin{proof}
Let $f$ be a mechanism satisfying the axioms. We will prove that
$f$ necessarily coincides with $f_{{\rm RSD}}$. To this end, we
focus on profiles where $f$ differs from $f_{{\rm RSD}}$.

First, we claim that if there exists a profile where $f$ differs
from $f_{{\rm RSD}}$, then there also exists such a profile in which
all agents rank the same house as their first choice. Second, we show
that $f$ cannot differ from $f_{{\rm RSD}}$ at such a profile.

For the first part, let $\mathbf{P}$ be a profile such that $f\left(\mathbf{P}\right)\neq f_{{\rm RSD}}\left(\mathbf{P}\right)$.
If all three agents rank the same house first, we may proceed directly
to the second part. Otherwise, note that by efficiency, if the three
agents rank three different houses first, then the probabilities are
uniquely determined at this profile. Thus, it must be the case that
two agents rank the same house first, and the third agent ranks a
different house first. Without loss of generality, suppose that agents
$1$ and $2$ rank house $a$ first, while agent $3$ ranks house
$b$ first. Since $f\left(\mathbf{P}\right)\neq f_{{\rm RSD}}\left(\mathbf{P}\right)$,
there must exist some agent for whom the two mechanisms assign different
distributions.

Suppose that $f\left(\mathbf{P}\right)_{3}\neq f_{{\rm RSD}}\left(\mathbf{P}\right)_{3}$.
Let $x$ be a house for which the two mechanisms differ, that is,
$f\left(\mathbf{P}\right)_{x,3}\neq f_{{\rm RSD}}\left(\mathbf{P}\right)_{x,3}$.
Note that $x\neq a$, because by efficiency, agent $3$ is not assigned
house $a$. Now, consider a sequence of profiles starting from $\mathbf{P}$,
where at each step we perform a swap in agent $3$'s preferences between
house $a$ and the house ranked immediately above it, until we reach
a profile $\mathbf{Q}$ in which agent $3$ ranks $a$ first. We claim
that $f\left(\mathbf{Q}\right)_{x,3}\neq f_{{\rm RSD}}\left(\mathbf{Q}\right)_{x,3}$
as well.

To see this, observe that the axioms determine the probability that
agent $3$ is assigned house $a$ at each profile along the sequence:
in the final profile $\mathbf{Q}$, agent $3$ is assigned $a$ with
probability $\frac{1}{3}$ by ETA and efficiency (since efficiency
ensures that $a$ is assigned with probability $1$). In all earlier
profiles, efficiency guarantees that agent $3$ is not assigned $a$.
Thus, the additional probability that agent $3$ gains for house $a$
after each swap is fully determined by the axioms, and by SP it corresponds
exactly to the reduction in the probability of the house with which
$a$ was swapped. In particular, for the swap involving $a$ and $x$
(if such a swap occurred), the change in $x$'s probability is determined,
while by SP, all other swaps do not affect the probability of $x$.
Therefore, we conclude that $f\left(\mathbf{Q}\right)_{x,3}-f_{{\rm RSD}}\left(\mathbf{Q}\right)_{x,3}=f\left(\mathbf{P}\right)_{x,3}-f_{{\rm RSD}}\left(\mathbf{P}\right)_{x,3}$
and hence $f\left(\mathbf{Q}\right)_{x,3}\neq f_{{\rm RSD}}\left(\mathbf{Q}\right)_{x,3}$,
as claimed.

Otherwise, $f\left(\mathbf{P}\right)_{3}=f_{{\rm RSD}}\left(\mathbf{P}\right)_{3}$
and the difference between the mechanisms lies in the distribution
of another agent. Without loss of generality, suppose it is agent
$1$. By ETA, the probability that agent $1$ is assigned house $a$
is $\frac{1}{2}$. Furthermore, by efficiency, agent $1$ cannot be
assigned houses that are not among their top 3 preferences. Since
agent $1$ is assigned some house with probability $1$ (because there
are more houses than agents), there must be at least two houses with
different probabilities, and these houses must be the second and third
preferences in agent $1$'s ranking.

Note that the probability of the third-ranked house is not determined
by the axioms, so it must be efficient for agent $1$ to receive this
house. In other words, there must be an ordering of the agents where
agent $1$ takes their third choice. However, agent $1$ can only
take their third choice if they are last in the ordering. In this
case, agents $2$ and $3$ must be assigned houses $a$ and $b$,
respectively. For agent $1$ to receive their third choice, agent
$1$'s top two preferences must be $a$ and $b$. Hence, agent $1$
must rank $b$ second.

At this point, by the earlier arguments, we have $f\left(\mathbf{P}\right)_{b,1}\neq f_{{\rm RSD}}\left(\mathbf{P}\right)_{b,1}$,
and since agent $3$ ranks $b$ first, house $b$ must be assigned
to some agent with probability $1$. Thus, there must be some other
agent where the allocation probability of $b$ differs between the
two mechanisms. Since $f\left(\mathbf{P}\right)_{3}=f_{{\rm RSD}}\left(\mathbf{P}\right)_{3}$,
this agent must be agent $2$, and the difference in probabilities
must cancel the difference observed for agent $1$.

By applying the same analysis to agent $2$, we conclude that agent
$2$ must rank $b$ second. Now, since agents $1$ and $2$ rank $a$
first and $b$ second, and agent $3$ ranks $b$ first, we have $f_{{\rm RSD}}\left(\mathbf{P}\right)_{b,i}=\frac{1}{6}$
for $i=1,2$. Moreover, since we have shown that the differences for
the two agents must cancel, it follows that $f\left(\mathbf{P}\right)_{b,2}\neq f_{{\rm RSD}}\left(\mathbf{P}\right)_{b,2}$
and $f\left(\mathbf{P}\right)_{b,2}\neq f\left(\mathbf{P}\right)_{b,1}$.
Hence, by ETE, their preferences must differ.

However, since both agents rank $a$ first and $b$ second, we can
make their preferences similar by applying a sequence of swaps to
agent $2$'s preference order that do not involve $b$. Let $\mathbf{P}^{\prime}$
denote the profile obtained after these swaps. Since these swaps do
not involve $b$, by SP, the difference in probabilities remains unchanged,
i.e., $f\left(\mathbf{P}\right)_{b,2}-f_{{\rm RSD}}\left(\mathbf{P}\right)_{b,2}=f\left(\mathbf{P}^{\prime}\right)_{b,2}-f_{{\rm RSD}}\left(\mathbf{P}^{\prime}\right)_{b,2}$.
Moreover, since in $\mathbf{P}^{\prime}$, agents $1$ and $2$ have
identical preferences, their differences must also be the same, i.e.,
$f\left(\mathbf{P}^{\prime}\right)_{b,2}-f_{{\rm RSD}}\left(\mathbf{P}^{\prime}\right)_{b,2}=f\left(\mathbf{P}^{\prime}\right)_{b,1}-f_{{\rm RSD}}\left(\mathbf{P}^{\prime}\right)_{b,1}$.
Since the total probability of $b$ is $1$ in $\mathbf{P}^{\prime}$,
in order to cancel the differences, we must have $f\left(\mathbf{P}^{\prime}\right)_{b,3}\neq f_{{\rm RSD}}\left(\mathbf{P}^{\prime}\right)_{b,3}$.
Thus, $f\left(\mathbf{P}^{\prime}\right)_{3}\neq f_{{\rm RSD}}\left(\mathbf{P}^{\prime}\right)_{3}$,
and we can apply the same arguments of the first case (where $f\left(\mathbf{P}\right)_{3}\neq f_{{\rm RSD}}\left(\mathbf{P}\right)_{3}$)
to obtain a profile $\mathbf{Q}^{\prime}$ where all agents rank the
same house first and $f\left(\mathbf{Q}^{\prime}\right)\neq f_{{\rm RSD}}\left(\mathbf{Q}^{\prime}\right)$,
as claimed.

For the second part, let $\mathbf{R}$ be a profile in which all agents
rank the same house first, and suppose $f\left(\mathbf{R}\right)\neq f_{{\rm RSD}}\left(\mathbf{R}\right)$.
Without loss of generality, assume that this top-ranked house is $a$,
and that the allocation differs at agent $1$. As in the previous
part, the difference must occur in the probabilities assigned to $1$'s
second and third choices, and thus, their third-ranked house must
be efficient for them. Let $b$ denote agent $1$'s second choice.
Then, for agent $1$'s third choice to be efficient, $b$ must be
taken before agent $1$ chooses, which requires at least one other
agent to rank $b$ second. Without loss of generality, suppose this
is agent $2$.

Now, regardless of whether agent $3$ ranks $b$ second or not, the
allocation probabilities for house $b$ must be fully determined by
ETA (if agent $3$ does not rank $b$ second, then by efficiency,
he cannot receive it). But this contradicts the assumption that agent
$1$'s allocation differs from $f_{{\rm RSD}}$ at house $b$. Therefore,
the assumption that such a profile $\mathbf{R}$ exists must be false.
It follows, by the arguments above, that $f=f_{{\rm RSD}}$, as desired.
\end{proof}

\subsection{\texorpdfstring{$n=m=4$}{n=m=4}} \label{subsec:n=m=4}
\lyxadded{bzmao}{Fri Jan  9 23:33:22 2026}{We focus here on the balanced
case $n=m=4$. The result is not new: it has been confirmed via computer-aided
proofs (see, e.g., Sandomirskiy }\lyxadded{bzmao}{Fri Jan  9 23:26:25 2026}{\cite{fssite}}\lyxadded{bzmao}{Fri Jan  9 23:33:22 2026}{).
The proof below is purely analytic, but it requires a detailed case
analysis.}
\begin{prop}
For $n=m=4$, the axioms ExPE, ETE, and SP uniquely characterize $f_{{\rm RSD}}$.
\end{prop}
\begin{proof}
We prove the proposition by showing that for each profile $\mathbf{P}\in\mathcal{R}^{N}$,
the assignment matrix is uniquely determined by the axioms ExPE, ETE,
and SP. The proof proceeds by induction on 
\[
d=d_{\mathbf{P}}\coloneqq\left|\left\{ \left(\left\{ i,j\right\} ,\left\{ h_{1},h_{2}\right\} \right)\in\binom{N}{2}\times\binom{H}{2}\mid\begin{array}{c}
\text{\ensuremath{i} and \ensuremath{j} disagree on the binary}\\
\text{preference between \ensuremath{h_{1}} and \ensuremath{h_{2}}}
\end{array}\right\} \right|.
\]
Intuitively, $d_{\mathbf{P}}$ quantifies the extent of disagreement
present in the profile $\mathbf{P}$.

In the base case $d=0$, all agents have identical rankings over the
houses. In this case, ExPE and ETE directly determine the assignment
matrix.

Now suppose $\mathbf{P}\in\mathcal{R}^{N}$ is a profile with $d=d_{\mathbf{P}}>0$.
We first observe that the assignment of agent $i$ in $\mathbf{P}$
is determined under the induction hypothesis if the following condition
$\left(*\right)_{\mathbf{P},i}$ holds: there exists at least one
pair $\left\{ h_{1},h_{2}\right\} $ of adjacent houses in $P_{i}$
such that less than two other agents agree with agent $i$ on the
binary preference between those houses; and furthermore, if there
is exactly one such pair, then at least one of the houses in that
pair cannot be assigned to agent $i$ in any efficient assignment
with respect to $\mathbf{P}$.

To show this, suppose we find such a pair $\left\{ h_{1},h_{2}\right\} $
of adjacent houses in agent $i$'s ranking. Consider the profile $\mathbf{P}^{\prime}$
obtained by swapping the positions of $h_{1}$ and $h_{2}$ in $i$'s
ranking. This operation reduces the number of disagreements in the
profile, so $d_{\mathbf{P}^{\prime}}<d$. By the induction hypothesis,
the assignment matrix for $\mathbf{P}^{\prime}$ is uniquely determined
by the axioms. By SP, the probability that agent $i$ receives any
house other than $h_{1}$ or $h_{2}$ must remain unchanged between
$\mathbf{P}$ and $\mathbf{P}^{\prime}$, and is therefore determined
in $\mathbf{P}$ as well. Since $n=m$, agent $i$ must be assigned
to some house. Therefore, once the probabilities that agent $i$ receives
the two houses outside the pair $\left\{ h_{1},h_{2}\right\} $ are
known, it suffices to determine either the probability that $i$ receives
$h_{1}$ or the probability that he receives $h_{2}$.

If there exists another such pair of adjacent houses in agent $i$'s
ranking, then without loss of generality, we may assume that this
second pair does not involve $h_{1}$. Applying the same reasoning
as before, we can determine the probability that agent $i$ receives
$h_{1}$, which completes the determination of his assignment. If
no such additional pair exists, then by assumption, one of the houses
in the original pair cannot be assigned to agent $i$ under any efficient
assignment. In that case, the probability that agent $i$ receives
that house is zero, and thus his assignment is fully determined. We
have therefore shown that under condition $\left(*\right)_{\mathbf{P},i}$,
the assignment of agent $i$ in profile $\mathbf{P}$ is uniquely
determined by the axioms.

Since $n=m$, every house must be assigned with probability $1$.
Therefore, to determine the entire assignment matrix of $\mathbf{P}$,
it suffices to determine the assignments of $n-1=3$ agents. More
generally, it suffices to determine the assignments of all agents
except a set of agents that have identical rankings, by ETE. By the
argument above, the assignment of any agent satisfying the condition
is fully determined.

Hence, it remains to consider only those profiles $\mathbf{P}$ in
which there exist at least two agents $i\in N$ with different rankings
who do not satisfy $\left(*\right)_{\mathbf{P},i}$. For clarity,
we introduce the following definitions.
\begin{defn}
[Supported agent]Fix a profile $\mathbf{P}$. An agent $i\in N$
is \emph{supported} (with respect to $\mathbf{P}$) if for every pair
$\left\{ h_{1},h_{2}\right\} $ of adjacent houses in $P_{i}$, there
are at least two agents other than $i$ who agree with agent $i$
on the binary preference between $\left\{ h_{1},h_{2}\right\} $,
possibly with a single exception: there may be one pair of adjacent
houses $\left\{ h_{1},h_{2}\right\} $ for which fewer than two other
agents agree with $i$, provided that both $h_{1}$ and $h_{2}$ are
assigned to agent $i$ in some efficient assignment with respect to
$\mathbf{P}$.
\end{defn}
\begin{rem}
Note that, if such an exceptional pair exists, then there must be
exactly one other agent who agrees with agent $i$ on the binary preference
between $h_{1}$ and $h_{2}$. Otherwise, by efficiency, agent $i$
could not be assigned the house he ranks lower in that pair, contradicting
the assumption that both houses must be assigned to him in some efficient
assignment.
\end{rem}
\begin{defn}
[Supported profile]A profile $\mathbf{P}$ is \emph{supported} if
there exist at least two supported agents with respect to $\mathbf{P}$
whose rankings are different.
\end{defn}
Thus, in the terminology above, the class of profiles that remain
under consideration are precisely the supported profiles. We complete
the proof by exhaustively verifying that in all such profiles, the
assignment matrix is uniquely determined by the axioms. The exhaustive
case analysis appears in Appendix \ref{sec:nEqmEq4_cases}.
\end{proof}
%

\section{Cases where there are other mechanisms} \label{sec:non-uniqueness}

To better understand the limitations of the axioms ExPE, ETE, and
SP, this section identifies the values of $m$ and $n$ for which
these axioms do \emph{not} suffice to characterize $f_{{\rm RSD}}$.
That is, for such values of $m$ and $n$, the axioms are too weak
to uniquely determine the assignment matrix for every possible preference
profile. The case $m=2$ was already resolved earlier, and the case
$n,m\geq5$ has been shown by Basteck and Ehlers \cite{basteck2025constrained}. The mechanisms
we construct in this section rely on ideas similar to those used in
their work, and they cover the remaining settings, namely $n>m\in\left\{3,4\right\}$
and $n=4,\,m\geq5$. 

\subsection{\texorpdfstring{$n>m\in\left\{3,4\right\}$}{n>m=3,4}} \label{subsec:n>m=3,4}
\begin{prop}
For $n>m\geq3$, the axioms ExPE, ETE, and SP do not suffice to characterize
the mechanism $f_{{\rm RSD}}$.
\end{prop}
\begin{proof}
To establish the \lyxdeleted{bzmao}{Mon Jan  5 19:12:42 2026}{lemma}\lyxadded{bzmao}{Mon Jan  5 19:12:44 2026}{proposition},
we construct a mechanism that satisfies the axioms but differs from
$f_{{\rm RSD}}$. Let $H\coloneqq\left\{ h_{1},\dots,h_{m}\right\} $.
We begin with some preliminary notations before specifying the mechanism.

We denote by $f^{1}$ the mechanism corresponding to choosing uniformly
an ordering of the agents in which agent $1$ comes first. Formally,
$f^{1}$ is the normal form of the extensive-form mechanism 
\[
M^{1}\coloneqq\frac{1}{\left(n-1\right)!}\sum_{\substack{\sigma\in S_{n}\\
\sigma\left(1\right)=1
}
}{\rm SD}_{\sigma}.
\]
Note that RSD can be obtained by symmetrizing $M^{1}$ over all renamings
of the agents, and therefore 
\[
f_{{\rm RSD}}=\frac{1}{n!}\sum_{\pi\in\Pi}\pi\left(f^{1}\right).
\]

For a preference profile $\mathbf{P}\in\mathcal{R}^{N}$ and an agent
$i\in N$, let ${\rm top}\left(P_{i}\right)$ denote the house most
preferred by agent $i$ in $\mathbf{P}$. Let $x=x_{\mathbf{P}}$
denote the second-best house of agent $1$ in $\mathbf{P}$, and set
$a\coloneqq h_{1}$. We then define 
\[
\mathsf{P}\coloneqq\left\{ \mathbf{P}\in\mathcal{R}^{N}\mid\begin{array}{c}
\forall i\in\left[m\right]:{\rm top}\left(P_{i}\right)=h_{i},\\
aP_{n}x
\end{array}\right\} .
\]
Note that $x\neq a$ for every $\mathbf{P}\in\mathcal{\mathsf{P}}$.

Let $\varepsilon>0$ be sufficiently small (e.g., $\varepsilon<\frac{1}{\left(n-1\right)!}$
suffices), and define the vector $v=\left(v_{h}\right)_{h\in H}$
by 
\[
v_{h}\coloneqq\begin{cases}
\varepsilon & \text{if }h=x,\\
-\varepsilon & \text{if }h=a,\\
0 & \text{otherwise}.
\end{cases}
\]
Now define the mechanism $f^{1,n}$ by 
\[
f^{1,n}\left(\mathbf{P}\right)_{i}\coloneqq\begin{cases}
f^{1}\left(\mathbf{P}\right)_{i}+v & \text{if }i=1\text{ and }\mathbf{P}\in\mathsf{P},\\
f^{1}\left(\mathbf{P}\right)_{i}-v & \text{if }i=n\text{ and }\mathbf{P}\in\mathsf{P},\\
f^{1}\left(\mathbf{P}\right)_{i} & \text{otherwise}.
\end{cases}
\]
The mechanism is well defined: the row sums and column sums remain
unchanged by construction, and the entries remain nonnegative. To
see the latter, note that in $f^{1}$, for profiles in $\mathsf{P}$,
agent $1$ receives $a$ with probability one, while agent $n$ receives
$x$ in at least one ordering where agent $1$ comes first. Specifically,
if $x=h_{i}$, then agent $n$ receives $x$ in the ordering where
the agents in $\left[m\right]\setminus\left\{ i\right\} $ appear
first in their natural order, and agent $n$ comes immediately after
them. Hence, in $f^{1}$ agent $n$ obtains $x$ with probability
at least $\frac{1}{\left(n-1\right)!}$. Thus, the adjustment by $v$
preserves feasibility while slightly shifting the allocation between
agents $1$ and $n$.

With these notations, define 
\[
f\coloneqq\frac{1}{n!}\sum_{\pi\in\Pi}\pi\left(f^{1,n}\right).
\]
We claim that $f$ is the desired mechanism, that is, it satisfies
the axioms ExPE, ETE, and SP, yet it differs from $f_{{\rm RSD}}$.
First, by construction, since it is defined by symmetrizing over all
agent renamings, it clearly satisfies ETE (and even anonymity). For
ExPE and SP, it suffices to verify that $f^{1,n}$ satisfies them,
because these properties are invariant under renamings and preserved
under convex combinations.

For ExPE, it is clear that $f^{1,n}$ satisfies the axiom whenever
$\mathbf{P}\notin\mathsf{P}$, since in that case $f^{1,n}\left(\mathbf{P}\right)=f^{1}\left(\mathbf{P}\right)$
and $f^{1}$ satisfies ExPE by construction. We therefore focus on
the case $\mathbf{P}\in\mathsf{P}$. Let $x\coloneqq h_{i}$ with
$i\neq1$. Consider the two assignments $s,s^{\prime}:N\rightarrow O$
defined by 
\[
\begin{array}{cc}
s\left(j\right)\coloneqq\begin{cases}
h_{j} & \text{if \ensuremath{j\in\left[m\right]\setminus\left\{ i\right\} }},\\
x & \text{if \ensuremath{j=n}},\\
\varnothing & \text{otherwise},
\end{cases} & s^{\prime}\left(j\right)\coloneqq\begin{cases}
x & \text{if \ensuremath{j=1}},\\
a & \text{if \ensuremath{j=n}},\\
s\left(j\right) & \text{otherwise}.
\end{cases}\end{array}
\]
Here $s^{\prime}$ differs from $s$ only by swapping the houses $a$
and $x$ between agents $1$ and $n$. Transferring $\varepsilon$
weight from $s$ to $s^{\prime}$ describes exactly the adjustment
needed to pass from $f^{1}\left(\mathbf{P}\right)$ to $f^{1,n}\left(\mathbf{P}\right)$.

Since $f^{1}$ satisfies ExPE, to show that $f^{1,n}$ satisfies ExPE
it remains to verify that this transfer is between efficient assignments,
and that the transferred weight $\varepsilon$ is less than the weight
placed on $s$ by some ex-post efficient extensive-form mechanism
whose normal form is $f^{1}$. Because $\varepsilon$ is chosen sufficiently
small and $M^{1}$ is such an extensive-form mechanism, it suffices
to show that $s$ can be obtained from an ordering of the agents where
agent $1$ comes first, and that $s^{\prime}$ can be obtained from
some ordering.

First, the assignment $s$ can be obtained by an ordering in which
the agents in $\left[m\right]\setminus\left\{ i\right\} $ appear
first in their natural order, and agent $n$ comes immediately after
them. This is indeed an ordering where agent $1$ comes first.

Second, let $I\subseteq\left[m\right]$ be the set of indices of the
houses that agent $n$ strictly prefers to $a$. Note that $1\notin I$
since $h_{1}=a$, and $i\notin I$ since $\mathbf{P}\in\mathsf{P}$
implies that agent $n$ prefers $a$ over $x=h_{i}$. Then $s^{\prime}$
can be obtained by an ordering where the agents in $I$ come first
(in an arbitrary order), followed by agent $n$, then agent $1$,
and then the remaining agents in $\left[m\right]\setminus\left\{ i\right\} $,
all before any of the other agents. Therefore, $f^{1,n}$ satisfies
ExPE.

The next step is to verify that $f^{1,n}$ also satisfies SP. Since
$f^{1}$ satisfies SP, agents $2,\dots,n-1$ cannot gain from any
unilateral misreport. We therefore need only check agents $1$ and
$n$. By the local characterization of SP, it suffices to consider
misreports in the form of a swap between two adjacent houses. Let
$\mathbf{Q}$ be the profile before the swap and $\mathbf{R}$ the
profile after it.

Since $f^{1}$ already satisfies SP, the only potentially problematic
cases are those where 
\[
f^{1,n}\left(\mathbf{Q}\right)-f^{1}\left(\mathbf{Q}\right)\neq f^{1,n}\left(\mathbf{R}\right)-f^{1}\left(\mathbf{R}\right),
\]
so that the modification in $f^{1,n}$ could change incentives, and
it suffices to verify that, taking $\mathbf{Q}$ as the profile of
true preferences, the difference under $\mathbf{Q}$ is at least as
favorable for the swapping agent as the difference under $\mathbf{R}$.
We will refer to this difference as the adjustment, meaning the $\varepsilon$
weight transfer that specifies how $f^{1,n}$ deviates from $f^{1}$
at a given profile.

Now, because $f^{1}\left(\mathbf{P}\right)=f^{1,n}\left(\mathbf{P}\right)$
for every $\mathbf{P}\notin\mathsf{P}$, we must have at least one
of $\mathbf{Q},\mathbf{R}$ in $\mathsf{P}$. Moreover, the roles
of $\mathbf{Q}$ and $\mathbf{R}$ can be interchanged: under the
local characterization of SP, the condition is expressed as a collection
of equalities together with a single inequality, and that inequality
simply reverses when the two houses are swapped. Hence, without loss
of generality, we may assume $\mathbf{Q}\in\mathsf{P}$. With these
reductions, it remains to check only those adjacent swaps by agents
$1$ and $n$ that actually affect the adjustment.

For agent $1$, since $\mathbf{Q}\in\mathsf{P}$, his first and second
choices are $a$ and $x$, respectively. Any swap not involving one
of these houses leaves the adjustment unchanged. Let $y$ denote his
third-preferred house in $\mathbf{Q}$. It remains to consider two
cases:

Swap between $a$ and $x$: here we directly compare $f^{1,n}\left(\mathbf{Q}\right)$
and $f^{1,n}\left(\mathbf{R}\right)$. In $\mathbf{Q}$, agent $1$
receives $a$ with probability $1-\varepsilon$ and $x$ with probability
$\varepsilon$. In $\mathbf{R}$, he receives $x$ with probability
$1$, which is strictly worse for him. Note also that the weight transfer
occurs exactly between the two houses involved in the swap.

Swap between $x$ and $y$: here we compare the adjustments mentioned
earlier. In $\mathbf{Q}$, the adjustment from $f^{1}$ to $f^{1,n}$
transfers $\varepsilon$ weight from $a$ to $x$. In $\mathbf{R}$,
the adjustment transfers $\varepsilon$ weight from $a$ to $y$.
Since agent $1$ prefers $x$ over $y$, this is less favorable for
him. Moreover, the difference between the adjustments involves only
the houses in the swap - it consists of an $\varepsilon$ weight transfer
from $x$ to $y$.

Thus, agent $1$ cannot gain from unilateral manipulation. For agent
$n$, the only swap that can matter is between $a$ and $x$. In this
case we move from a profile $\mathbf{Q}\in\mathsf{P}$ where the adjustment
transfers $\varepsilon$ weight from $x$ to $a$ (which he prefers),
to a profile $\mathbf{R}\notin\mathsf{P}$, where no transfer occurs.
This is again less favorable, and the difference concerns only $a$
and $x$, the houses in the swap. Therefore, agent $n$ cannot gain
from a unilateral manipulation either. Hence, $f$ satisfies ExPE,
ETE, and SP.

Finally, it remains to observe that $f\neq f_{{\rm RSD}}$. Fix a
profile $\mathbf{Q}\in\mathsf{P}$ such that $x_{\mathbf{Q}}=h_{2}$
and no agent other than agent $1$ ranks $a$ as his top choice. Note
that such a profile exists, by our assumption that $m\ge3$. By the
definition of $f^{1,n}$ on profiles in $\mathsf{P}$, we have 
\[
f^{1,n}\left(\mathbf{Q}\right)_{h_{2},1}>f^{1}\left(\mathbf{Q}\right)_{h_{2},1}.
\]
We now claim that for every $\pi\in\Pi$, 
\[
\pi\left(f^{1,n}\right)\left(\mathbf{Q}\right)_{h_{2},1}\geq\pi\left(f^{1}\right)\left(\mathbf{Q}\right)_{h_{2},1}.
\]
Suppose otherwise, in which case for some $\pi$ the adjustment from
$\pi\left(f^{1}\right)$ to $\pi\left(f^{1,n}\right)$ would reduce
the probability that agent $1$ receives $h_{2}$. Since such a reduction
can only occur when $\pi\left(\mathbf{Q}\right)\in\mathsf{P}$, this
enforces agent $\pi^{-1}\left(1\right)$ to rank $a$ as his first
choice in $\mathbf{Q}$, so by the construction of $\mathbf{Q}$ we
must have $\pi^{-1}\left(1\right)=1$. However, the adjustment never
reduces the probability that agent $1$ receives his second-best house,
a contradiction.

Hence, $\pi\left(f^{1,n}\right)\left(\mathbf{Q}\right)_{h_{2},1}\geq\pi\left(f^{1}\right)\left(\mathbf{Q}\right)_{h_{2},1}$
for all $\pi\in\Pi$, and strict inequality holds for $\pi={\rm id}$.
Since $f$ and $f_{{\rm RSD}}$ are the averages of $\left\{ \pi\left(f^{1,n}\right)\right\} _{\pi\in\Pi}$
and $\left\{ \pi\left(f^{1}\right)\right\} _{\pi\in\Pi}$, respectively,
it follows that 
\[
f\left(\mathbf{Q}\right)_{h_{2},1}>f_{{\rm RSD}}\left(\mathbf{Q}\right)_{h_{2},1}.
\]
Therefore, $f\neq f_{{\rm RSD}}$, so the axioms do not suffice to
characterize $f_{{\rm RSD}}$ in such settings.
\end{proof}

\subsection{\texorpdfstring{$n=4,\,m\ge 5$}{n=4, m>=5}} \label{subsec:n=4, m>=5}
In the setting $n=4,\,m\geq5$, we show that $f_{{\rm RSD}}$ is not
unique by presenting another mechanism satisfying the axioms. The
idea is borrowed from \cite{basteck2025constrained}: we take the
mechanism introduced there for five agents and restrict it to four
agents by removing agent $5$ and its allocation. We now formalize
this construction in the following proposition.
\begin{prop}
\label{lem:nEq4_mGe5}For $n=4,\,m\geq5$, the axioms ExPE, ETE, and
SP do not suffice to characterize the mechanism $f_{{\rm RSD}}$.
\end{prop}
\begin{proof}
Let $a,b,c,d,e$ denote five distinct houses. We extend the notation
from the previous proof: for a preference profile $\mathbf{P}$, an
agent $i\in N$, and $j\in\left[m\right]$, let ${\rm top}_{j}\left(P_{i}\right)$
denote the $j$-th most-preferred house of agent $i$ under $\mathbf{P}$.
Using this notation, define 
\[
\mathsf{P}\coloneqq\left\{ \mathbf{P}\in\mathcal{R}^{N}\mid\begin{array}{c}
\forall h\in H\setminus\left\{ a,b,e\right\} :aP_{1}h\text{ and }bP_{1}h,\\
\forall i\in N\setminus\left\{ 1\right\} :{\rm top}_{1}\left(P_{i}\right)=e\text{ and }{\rm top}_{3}\left(P_{i}\right)=c,\\
{\rm top}_{2}\left(P_{2}\right)=a,\,{\rm top}_{2}\left(P_{3}\right)=b,\,{\rm top}_{2}\left(P_{4}\right)=d
\end{array}\right\} .
\]
That is, $\mathsf{P}$ is the set of profiles in which agent $1$'s
top two houses in $H\setminus\left\{ e\right\} $ are $a$ and $b$
(in some order), while agents $2,3,4$ all rank $e$ first and $c$
third, with their second-ranked houses fixed as $a,b,d$, respectively.
Denote by $x=x_{\mathbf{P}}$ the most-preferred house of agent $1$
in $H\setminus\left\{ a,b,e\right\} $ under profile $\mathbf{P}$,
and by $y=y_{\mathbf{P}}$ his less-preferred house among $\left\{ a,b\right\} $.

Let $\varepsilon>0$ be sufficiently small (for instance, $\varepsilon<\frac{1}{4!}$
suffices), and define the vector $v=\left(v_{h}\right)_{h\in H}$
by 
\[
v_{h}\coloneqq\begin{cases}
\varepsilon & \text{if \ensuremath{h=y}},\\
-\varepsilon & \text{if \ensuremath{h=x}},\\
0 & \text{otherwise}.
\end{cases}
\]
Now define a mechanism $g$ by 
\[
g\left(\mathbf{P}\right)_{i}\coloneqq\begin{cases}
f_{{\rm RSD}}\left(\mathbf{P}\right)_{i}+v & \text{if \ensuremath{i=1} and \ensuremath{\mathbf{P}\in\mathsf{P}}},\\
f_{{\rm RSD}}\left(\mathbf{P}\right)_{i} & \text{otherwise}.
\end{cases}
\]
In words, $g$ coincides with $f_{{\rm RSD}}$ except on profiles
in $\mathsf{P}$, where agent $1$'s assignment is adjusted by transferring
$\varepsilon$ probability from house $x$ to house $y$. The mechanism
is well defined, since under the ordering $4231$ (agents listed from
first to last), agent $1$ receives house $x$. We now obtain the
desired mechanism $f$ by symmetrizing $g$ over all renamings of
the agents: 
\[
f\coloneqq\frac{1}{4!}\sum_{\pi\in\Pi}\pi\left(g\right),
\]
and we will show that $f$ satisfies the axioms and differs from $f_{{\rm RSD}}$.
Since $f$ is obtained by symmetrizing $g$ over all renamings, it
satisfies ETE (and in fact anonymity). For the remaining axioms, as
in the previous proof, it suffices to verify them for $g$.

For ExPE, note first that $f_{{\rm RSD}}$ satisfies the axiom, so
it suffices to consider profiles $\mathbf{P}\in\mathsf{P}$. Since
every efficient assignment with respect to $\mathbf{P}$ has probability
at least $\varepsilon$ in ${\rm RSD}\left(\mathbf{P}\right)$, it
is enough to construct an adjustment that shifts probability between
efficient assignments so that agent $1$'s allocation transfers $\varepsilon$
probability from $x$ to $y$, while the allocations of the other
agents remain unchanged. We distinguish two symmetric cases, according
to whether $aP_{1}b$ or $bP_{1}a$. We present the case $aP_{1}b$;
the other follows by the same reasoning with the roles of agents $2$
and $3$ interchanged. In the case $aP_{1}b$, we have $y=b$. Consider
the following four assignments, where each agent is matched with the
house listed beneath him: 
\[
\begin{pmatrix}1 & 2 & 3 & 4\\
x & a & b & e
\end{pmatrix}\begin{pmatrix}1 & 2 & 3 & 4\\
a & c & e & d
\end{pmatrix}\rightarrow\begin{pmatrix}1 & 2 & 3 & 4\\
b & a & e & d
\end{pmatrix}\begin{pmatrix}1 & 2 & 3 & 4\\
a & c & b & e
\end{pmatrix}.
\]
Reducing $\varepsilon$ probability from each assignment on the left
and adding $\varepsilon$ probability to each assignment on the right
yields exactly the desired adjustment: agent $1$'s allocation shifts
$\varepsilon$ probability from $x$ to $b$, while the allocations
of all the other agents remain unchanged. It remains to check that
these four assignments are efficient. Indeed, the two assignments
on the left are generated by the orderings $4231$ and $3124$, and
the two on the right by $3214$ and $4123$. Hence, all are efficient
assignments, and thus $g$ satisfies ExPE.

For SP, as in the previous proof, it suffices to consider the adjacent
swaps of agent $1$ that can affect the adjustment from $f_{{\rm RSD}}$
to $g$. Let $\mathbf{Q}$ denote the profile of true preferences
and $\mathbf{R}$ the profile after the swap. Without loss of generality
(by the same argument as in the previous proof), we may assume that
$\mathbf{Q}\in\mathsf{P}$. Moreover, by symmetry, we may assume that
$aQ_{1}b$, and denote by $c^{\prime},d^{\prime}$ the two most-preferred
houses of agent $1$ in $H\setminus\left\{ a,b,e\right\} $, ordered
so that $c^{\prime}Q_{1}d^{\prime}$. The only adjacent swaps of agent
$1$ that can influence the adjustment are those between the pairs
$\left(a,b\right),\left(b,c^{\prime}\right),\left(c^{\prime},d^{\prime}\right)$.
By the definition of $g$, in each case the difference between the
adjustments under $\mathbf{Q}$ and $\mathbf{R}$ consists solely
of a shift in probability between the two houses in the swapped pair.
When the swapped pair is $\left(b,c^{\prime}\right)$, there is no
adjustment under $\mathbf{R}$ , which is worse for agent $1$ than
the adjustment under $\mathbf{Q}$; hence he cannot gain from this
deviation. For the swaps $\left(a,b\right)$ and $\left(c^{\prime},d^{\prime}\right)$,
the adjustment under $\mathbf{R}$ is actually more favorable to agent
$1$, which presents the potential difficulty.

When the swapped pair is $\left(c^{\prime},d^{\prime}\right)$, the
change in the adjustment from $\mathbf{Q}$ to $\mathbf{R}$ shifts
$\varepsilon$ probability from $d^{\prime}$ to $c^{\prime}$. Hence,\[\resizebox{\linewidth}{!}{$\begin{aligned}\left(\sum_{h\in C_{Q_{1}}\left(h^{\prime}\right)}g\left(\mathbf{R}\right)_{1,h}-\sum_{h\in C_{Q_{1}}\left(h^{\prime}\right)}f_{{\rm RSD}}\left(\mathbf{R}\right)_{1,h}\right)-\left(\sum_{h\in C_{Q_{1}}\left(h^{\prime}\right)}g\left(\mathbf{Q}\right)_{1,h}-\sum_{h\in C_{Q_{1}}\left(h^{\prime}\right)}f_{{\rm RSD}}\left(\mathbf{Q}\right)_{1,h}\right)=\begin{cases}
\varepsilon & \text{if \ensuremath{h^{\prime}=c^{\prime}}},\\
0 & \text{otherwise}.
\end{cases}\end{aligned}$}\] Since $f_{{\rm RSD}}$ satisfies SP, for every $h^{\prime}\neq c^{\prime}$
we already have 
\[
\sum_{h\in C_{Q_{1}}\left(h^{\prime}\right)}g\left(\mathbf{R}\right)_{1,h}\leq\sum_{h\in C_{Q_{1}}\left(h^{\prime}\right)}g\left(\mathbf{Q}\right)_{1,h}.
\]
To establish the same inequality for $h^{\prime}=c^{\prime}$, it
suffices to show that 
\[
\sum_{h\in C_{Q_{1}}\left(c^{\prime}\right)}f_{{\rm RSD}}\left(\mathbf{R}\right)_{1,h}\leq\left(\sum_{h\in C_{Q_{1}}\left(c^{\prime}\right)}f_{{\rm RSD}}\left(\mathbf{Q}\right)_{1,h}\right)-\varepsilon.
\]
By the definition of RSD and the choice of $\varepsilon$, this reduces
to showing that 
\[
\left|\left\{ \sigma\in S_{n}\mid{\rm SD}_{\sigma}\left(\mathbf{R}\right)\left(1\right)\in C_{Q_{1}}\left(c^{\prime}\right)\right\} \right|<\left|\left\{ \sigma\in S_{n}\mid{\rm SD}_{\sigma}\left(\mathbf{Q}\right)\left(1\right)\in C_{Q_{1}}\left(c^{\prime}\right)\right\} \right|.
\]
That is, the number of orderings in which agent $1$ receives a house
in $C_{Q_{1}}\left(c^{\prime}\right)$ strictly decreases when moving
from $\mathbf{Q}$ to $\mathbf{R}$. Indeed, the set corresponding
to $\mathbf{R}$ is contained in that for $\mathbf{Q}$, and they
are not equal: for example, the ordering $4231$ belongs to the set
for $\mathbf{Q}$ but not to the one for $\mathbf{R}$. Hence, $g\left(\mathbf{R}\right)_{1}\preceq_{Q_{1}}g\left(\mathbf{Q}\right)_{1}$.
For the swapped pair $\left(a,b\right)$, the same argument applies:
$a$ takes the role of $c^{\prime}$, $b$ the role of $d^{\prime}$,
and the relevant ordering is any ordering with agent $1$ placed second.
Thus, $g$ satisfies SP.

It remains to show that $f\neq f_{{\rm RSD}}$. Since RSD satisfies
anonymity, we have $\pi\left(f_{{\rm RSD}}\right)=f_{{\rm RSD}}$
for every $\pi\in\Pi$. Fix a profile $\mathbf{P}\in\mathsf{P}$.
By the definition of $g$, $g\left(\mathbf{P}\right)\neq f_{{\rm RSD}}\left(\mathbf{P}\right)$.
Moreover, for every ${\rm id}\neq\pi\in\Pi$, we have $\pi\left(\mathbf{P}\right)\notin\mathsf{P}$:
Indeed, agent $1$ is the only agent whose top two houses in $H\setminus\left\{ e\right\} $
are $\left\{ a,b\right\} $, and agents $2,3,4$ each have a distinct
second choice. Consequently, $\pi\left(g\right)\left(\mathbf{P}\right)=\pi\left(f_{{\rm RSD}}\right)\left(\mathbf{P}\right)$.
Since $f$ is the average of $\left\{ \pi\left(g\right)\right\} _{\pi\in\Pi}$
and $f_{{\rm RSD}}$ is the average of $\left\{ \pi\left(f_{{\rm RSD}}\right)\right\} _{\pi\in\Pi}$,
it follows that $f\left(\mathbf{P}\right)\neq f_{{\rm RSD}}\left(\mathbf{P}\right)$
and therefore $f\neq f_{{\rm RSD}}$. Thus, the axioms do not suffice
to characterize $f_{{\rm RSD}}$ in these settings.
\end{proof}
\begin{rem}
[Dominance over $f_{{\rm RSD}}$]\label{rem:nEq4_mGe5_dominance}

The mechanism constructed in the proof of Proposition \textcolor{blue}{\ref{lem:nEq4_mGe5}}
does not merely differ from $f_{{\rm RSD}}$; it dominates it. Concretely,
for every profile, it weakly improves each agent's allocation relative
to $f_{{\rm RSD}}$; moreover, at some profile, it strictly improves
at least one agent's allocation.

This follows directly from the construction in the proof: relative
to $f_{{\rm RSD}}$, the only change is that at certain profiles,
it shifts an $\varepsilon$-amount of probability mass in one agent's
allocation from house $x$ to house $y$, which this agent ranks strictly
above $x$, while leaving all other agents' allocations unchanged.
The subsequent symmetrization step (averaging over renamings) preserves
these profile-by-profile weak improvements for every agent and therefore
yields strict improvement at any profile where at least one of the
averaged terms performs the $\varepsilon$-shift.
\end{rem}
\begin{rem}
\label{rem:nEq4_mGe5_generalization}The dominance over $f_{{\rm RSD}}$
established in Remark \textcolor{blue}{\ref{rem:nEq4_mGe5_dominance}}
is not confined to the specific market sizes treated in Proposition
\textcolor{blue}{\ref{lem:nEq4_mGe5}}. Using the same extension-and-symmetrization
idea described in Remark \ref{rem:mEq2_domination_generalization},
one can lift the construction to any $\left(n,m\right)$ market with
$4\leq n\leq m-1$. Briefly, one enlarges the market by adding an
equal number of auxiliary agents and new houses, and then lets the
auxiliary agents select first. If, after the auxiliary agents select,
the set of remaining houses coincides with the original set of houses
(that is, none of the original houses is selected by an auxiliary
agent), then one applies the dominating mechanism constructed in Proposition
\textcolor{blue}{\ref{lem:nEq4_mGe5}}; otherwise, one applies $f_{{\rm RSD}}$.
Finally, one uniformly averages over all renamings of the agents.
The resulting mechanism satisfies the axioms and dominates $f_{{\rm RSD}}$
in the $\left(n,m\right)$ market. We omit the details.
\end{rem}
\begin{rem}
    The non-uniqueness results naturally raise a welfare question: in domains where the axioms do not uniquely characterize $f_{\rm RSD}$, can $f_{\rm RSD}$ be dominated by another mechanism satisfying the same axioms? Remarks \ref{rem:mEq2_domination_generalization} and \ref{rem:nEq4_mGe5_generalization} provide an affirmative answer to this question for $n\ge m+2$ (where $m\ge 2$) and for $4\le n\le m-1$, respectively. In the remaining regions of non-uniqueness, namely $n = m+1 \ge 4$ and $n = m \ge 5$, it is an open question whether dominance is possible, or whether $f_{\rm RSD}$ is undominated despite non-uniqueness.
\end{rem}

\subsection{Non-uniqueness under strengthened and additional axioms} \label{subsec:adding_BI}
Since we have established that $f_{{\rm RSD}}$ is not always the
only mechanism satisfying ExPE, ETE, and SP, a natural question arises:
can adding or strengthening axioms guarantee uniqueness for every
number of agents $n$ and houses $m$? Basteck and Ehlers \cite{basteck2025constrained}
have raised this type of question with respect to adding the axiom
\emph{Bounded Invariance (BI)}:
\begin{defn}
A mechanism $f$ satisfies \emph{Bounded Invariance (BI)} if for every
preference profile $\mathbf{P}\in\mathcal{R}^{N}$, agent $i\in N$,
preference order $P_{i}^{\prime}\in\mathcal{R}$, and house $h\in H$,
the following holds: 
\[
\left(\forall h^{\prime}\in C_{P_{i}}\left(h\right):C_{P_{i}^{\prime}}\left(h^{\prime}\right)=C_{P_{i}}\left(h^{\prime}\right)\right)\implies f\left(\mathbf{P}\right)_{h}=f\left(\mathbf{P}_{-i},P_{i}^{\prime}\right)_{h}.
\]
In other words, if an agent changes its ranking without changing the
part above and including $h$, then the assignment probability of
$h$ remains unchanged for every agent.
\end{defn}

However, the answer to the uniqueness question under this extra axiom
is negative, as the following result shows, even when strengthening
ETE to anonymity and adding neutrality.
\begin{prop}
\label{prop:BI_non-uniqueness} When $n\geq m-1\geq5$, $f_{{\rm RSD}}$
is not the unique mechanism satisfying anonymity, neutrality, ExPE,
SP, and BI.
\end{prop}

\begin{proof}
We construct a suitable mechanism $f$. Denote $H\coloneqq\left\{ a,b,h_{1},\dots,h_{m-2}\right\} $
and $j\coloneqq m-1\in N$; from this point on, $m-1$ refers to the
integer, and ``agent $j$'' refers to the agent with index $m-1$.

Define 
\[
\mathsf{P}\coloneqq\left\{ \mathbf{P}\in\mathcal{R}^{N}\mid\forall i\in\left[m-2\right],\,\forall h\in H\setminus\left\{ a,b,h_{i}\right\} :aP_{i}bP_{i}h_{i}P_{i}h\right\} .
\]
Thus, for each $i\in\left[m-2\right]$, agent $i$'s top three choices
are $a$, $b$, and $h_{i}$, in that order. For every $\mathbf{P}\in\mathcal{R}^{N}$,
let $x\coloneqq x_{\mathbf{P}}$ and $y\coloneqq y_{\mathbf{P}}$
denote agent $j$'s second and third choices among $H\setminus\left\{ a,b\right\} $,
respectively. Fix $\varepsilon>0$ sufficiently small (e.g., $\varepsilon\leq\frac{1}{n!}$),
and define $v\coloneqq\left(v_{h}\right)_{h\in H}$ by 
\[
v_{h}\coloneqq\begin{cases}
\varepsilon & \text{if }h=x,\\
-\varepsilon & \text{if }h=y,\\
0 & \text{otherwise}.
\end{cases}
\]
With $v$ as defined, we specify the mechanism $g$ by altering $f_{{\rm RSD}}$
only on $\mathsf{P}$: shift $\varepsilon$ probability from $y$
to $x$ in agent $j$'s assignment and, if $n\geq m$, make the opposite
shift for agent $m$; namely, 
\[
g\left(\mathbf{P}\right)_{i}=\begin{cases}
f_{{\rm RSD}}\left(\mathbf{P}\right)_{i}+v & \text{if }\mathbf{P}\in\mathsf{P}\text{ and }i=j,\\
f_{{\rm RSD}}\left(\mathbf{P}\right)_{i}-v & \text{if }\mathbf{P}\in\mathsf{P}\text{ and }i=m\text{ (when \ensuremath{n\geq m}),}\\
f_{{\rm RSD}}\left(\mathbf{P}\right)_{i} & \text{otherwise}.
\end{cases}
\]
From now on, we assume without loss of generality that $h_{i}P_{j}h_{i+1}$
for every $i\in\left[m-3\right]$; the remaining cases follow by permuting
the roles of agents $1,\dots,m-2$ consistently with the permutation
of agent $j$'s preferences among $h_{1},\dots,h_{m-2}$. Thus, $x_{\mathbf{P}}=h_{2}$
and $y_{\mathbf{P}}=h_{3}$.

To show that $g$ is well defined, it suffices to verify that, under
$f_{{\rm RSD}}$, each of the following occurs with probability at
least $\varepsilon$: agent $j$ receives $h_{3}$; when $n=m-1$,
the house $h_{2}$ remains unassigned; and when $n\ge m$, agent $m$
receives $h_{2}$. Since each ordering in RSD has probability $\frac{1}{n!}$
and $\varepsilon\le\frac{1}{n!}$, it is enough to exhibit one ordering
for each event. For agent $j$: in an ordering starting with $4321j$,
the first four agents take $a,b,h_{2},h_{1}$, so when agent $j$
arrives he receives $h_{3}$. For the remaining events, consider an
ordering that begins with $21j$, followed by the agents in $\left\{ 3,\dots,m-2\right\} $
in some order; here agents $2$ and $1$ take $a$ and $b$, agent
$j$ then takes $h_{1}$, and the agents in $\left\{ 3,\dots,m-2\right\} $
take their respective $h_{i}$'s, leaving $h_{2}$ as the only unassigned
house. If $n=m-1$, this already lists all agents, so the ordering
is complete and $h_{2}$ remains unassigned; if $n\ge m$, let agent
$m$ come next and he receives $h_{2}$. Therefore, $g$ is well defined.

Now define the mechanism $f$ as the symmetrization of $g$ over all
renamings of agents and houses: 
\[
f=\frac{1}{n!m!}\sum_{\left(\pi,\tau\right)\in\Pi\times\Gamma}\left(\pi,\tau\right)\left(g\right).
\]
Note also that, since $f_{{\rm RSD}}$ is anonymous and neutral, it
equals its own symmetrization: 
\[
f_{{\rm RSD}}=\frac{1}{n!m!}\sum_{\left(\pi,\tau\right)\in\Pi\times\Gamma}\left(\pi,\tau\right)\left(f_{{\rm RSD}}\right).
\]
We claim that $f$ is the desired mechanism. By construction, $f$
is anonymous and neutral. Moreover, ExPE, SP, and BI are invariant
under renamings and are preserved under convex combinations; hence
it suffices to verify these properties for $g$. We will therefore
show that $g$ satisfies ExPE, SP, and BI, and we will also establish
that $f\neq f_{{\rm RSD}}$.

We now verify that $g$ satisfies ExPE. We note that $f_{{\rm RSD}}$
satisfies ExPE; hence, it suffices to consider $\mathbf{P}\in\mathsf{P}$
and to show that the adjustment from $f_{{\rm RSD}}$ to $g$ can
be realized as a probability shift between efficient assignments,
reducing at most $\varepsilon$ probability from any single assignment.
Consider the four assignments $s_{1},s_{2},s_{1}^{\prime},s_{2}^{\prime}$
given by 
\begin{align*}
s_{1}\left(i\right)=\begin{cases}
a & \text{if }i=4,\\
b & \text{if }i=3,\\
h_{3} & \text{if }i=j,\\
h_{i} & \text{if }i\in\left[m-2\right]\setminus\left\{ 3,4\right\} ,\\
h_{4} & \text{if }i=m,\\
\varnothing & \text{otherwise},
\end{cases} & s_{2}\left(i\right)=\begin{cases}
a & \text{if }i=2,\\
b & \text{if }i=1,\\
h_{1} & \text{if }i=j,\\
h_{i} & \text{if }i\in\left[m-2\right]\setminus\left\{ 1,2\right\} ,\\
h_{2} & \text{if }i=m,\\
\varnothing & \text{otherwise},
\end{cases}\\
s_{1}^{\prime}\left(i\right)=\begin{cases}
a & \text{if }i=4,\\
b & \text{if }i=1,\\
h_{1} & \text{if }i=j,\\
h_{i} & \text{if }i\in\left[m-2\right]\setminus\left\{ 1,4\right\} ,\\
h_{4} & \text{if }i=m,\\
\varnothing & \text{otherwise},
\end{cases} & s_{2}^{\prime}\left(i\right)=\begin{cases}
a & \text{if }i=2,\\
b & \text{if }i=3,\\
h_{2} & \text{if }i=j,\\
h_{i} & \text{if }i\in\left[m-2\right]\setminus\left\{ 2,3\right\} ,\\
h_{3} & \text{if }i=m,\\
\varnothing & \text{otherwise}.
\end{cases}
\end{align*}
Reducing $\varepsilon$ probability from $s_{1}$ and $s_{2}$ and
adding $\varepsilon$ probability to $s_{1}^{\prime}$ and $s_{2}^{\prime}$
implements exactly the required adjustment: agent $j$ shifts $\varepsilon$
probability from $y=h_{3}$ to $x=h_{2}$; when $n\geq m$, agent
$m$ shifts $\varepsilon$ probability from $x=h_{2}$ to $y=h_{3}$;
and all other agents' marginals are unchanged. Each of $s_{1},s_{2},s_{1}^{\prime},s_{2}^{\prime}$
is the outcome of some ordering of the agents; hence each is an efficient
assignment. For example, they can be obtained from orderings whose
initial segments are $4321j$, $21j34$, $41j23$, and $2314j$, respectively;
in each case the ordering then continues with the remaining agents
in $\left[m-2\right]$ in some order, then (if $n\geq m$) agent $m$,
and finally the rest. Hence the adjustment is a probability shift
between efficient assignments, and $g$ satisfies ExPE.

For SP, since $f_{{\rm RSD}}$ satisfies it, every agent in $N\setminus\left\{ j,m\right\} $
cannot gain by manipulation, because their assignments coincide under
$g$ and $f_{{\rm RSD}}$ at every profile. Moreover, agent $m$ (if
he exists) cannot gain either, since his report cannot affect the
adjustment. It remains to show that agent $j$ cannot gain.

Let $\mathbf{P}\in\mathcal{R}^{N}$. As before, it suffices to consider
manipulations by agent $j$ that swap adjacent houses and thereby
change the adjustment. If $\mathbf{P}\notin\mathsf{P}$, then $\left(\mathbf{P}_{-j},R_{j}\right)\notin\mathsf{P}$
for all $R_{j}\in\mathcal{R}$, so the adjustment never changes. For
$\mathbf{P}\in\mathsf{P}$, the only adjacent swaps by agent $j$
that can affect the adjustment are between the pairs $\left(h_{k},h_{k+1}\right)$
for $k=1,2,3$. By the definition of $g$, the difference between
the adjustments to agent $j$'s assignment before and after any such
swap consists solely of a probability shift between the two houses
in the swapped pair.
\begin{itemize}
\item If the swapped pair is $\left(h_{2},h_{3}\right)$, the adjustment
flips from transferring $\varepsilon$ probability from $h_{3}$ to
$h_{2}$ to the opposite transfer. This is worse for agent $j$, so
he cannot gain from this swap.
\item If the swapped pair is $\left(h_{3},h_{4}\right)$, the adjustment
changes from transferring $\varepsilon$ probability from $h_{3}$
to $h_{2}$ to transferring $\varepsilon$ from $h_{4}$ to $h_{2}$.
Equivalently, the net change shifts $\varepsilon$ probability from
$h_{4}$ to $h_{3}$. By the same reasoning as in Proposition \ref{lem:nEq4_mGe5},
it suffices to exhibit some ordering in which agent $j$ receives
a house in $C_{P_{j}}\left(h_{3}\right)$ before the swap but not
after it. Any ordering whose initial segment is $4321j$ has this
property: under $\mathbf{P}$, agent $j$ receives $h_{3}$, whereas
after the swap he receives $h_{4}$.
\item If the swapped pair is $\left(h_{1},h_{2}\right)$, the same argument
applies, now taking some ordering whose initial segment is $21j$.
\end{itemize}
Therefore, $g$ satisfies SP.

For BI, let $\mathbf{P}\in\mathcal{R}^{N}$, $i\in N$, $P_{i}^{\prime}\in\mathcal{R}$,
and $h\in H$. Any change from $P_{i}$ to $P_{i}^{\prime}$ can be
written as a sequence of adjacent swaps. Since BI requires equality
at each step for every house that appears above the swap at that step,
it suffices to consider the case where $P_{i}^{\prime}$ differs from
$P_{i}$ by a single adjacent swap. Because $f_{{\rm RSD}}$ satisfies
BI, if the adjustments for $\mathbf{P}$ and $\mathbf{P}^{\prime}\coloneqq\left(\mathbf{P}_{-i},P_{i}^{\prime}\right)$
are identical with respect to $h$, then 
\[
g\left(\mathbf{P}^{\prime}\right)_{h}-g\left(\mathbf{P}\right)_{h}=f_{{\rm RSD}}\left(\mathbf{P}^{\prime}\right)_{h}-f_{{\rm RSD}}\left(\mathbf{P}\right)_{h}=0.
\]
Thus it remains to check whether there exists a swap below a house
$h$ that changes the adjustment at $h$.

If $i\in N\setminus\left[m-1\right]$, then no swap by agent $i$
changes the adjustment. If $i\in\left[m-2\right]$, the only houses
that can lie above a swap of agent $i$ that changes the adjustment
are $a$ and $b$; however, the adjustment never alters any agent's
probability of receiving $a$ or $b$. Otherwise, $i=j$. In that
case, as noted earlier, the only adjacent swaps of agent $j$ that
can affect the adjustment are between the pairs $\left(h_{k},h_{k+1}\right)$
for $k=1,2,3$. Each such swap shifts probability only between the
two swapped houses for agent $j$ (and, if $n\geq m$, for agent $m$
as well). Hence no probability associated with any house above the
swap is affected. Therefore, $g$ satisfies BI. Hence, $f$ satisfies
all the axioms stated in the proposition.

It remains to show that $f\neq f_{{\rm RSD}}$. Fix a profile $\mathbf{Q}\in\mathsf{P}$
such that: agent $j$'s top choices are $h_{1}Q_{j}h_{2}$; for every
$i\in N\setminus\left\{ j\right\} $, agent $i$'s top two choices
are $aQ_{i}b$; and for every $i\in N\setminus\left\{ 2,j\right\} $,
the house $h_{2}$ is not among $i$'s top four choices (these conditions
are feasible since $m\geq6$). By the construction of $g$, we have
$g\left(\mathbf{Q}\right)_{h_{2},j}>f_{{\rm RSD}}\left(\mathbf{Q}\right)_{h_{2},j}$.

We claim that for every $\left(\pi,\tau\right)\in\Pi\times\Gamma$,
\[
\left(\pi,\tau\right)\left(g\right)\left(\mathbf{Q}\right)_{h_{2},j}\geq\left(\pi,\tau\right)\left(f_{{\rm RSD}}\right)\left(\mathbf{Q}\right)_{h_{2},j}.
\]
Assume, toward a contradiction, that this fails for some $\left(\pi,\tau\right)$.
Let $\mathbf{Q}^{\prime}\coloneqq\left(\pi,\tau\right)\left(\mathbf{Q}\right)$.
Then 
\[
g\left(\mathbf{Q}^{\prime}\right)_{\tau\left(h_{2}\right),\pi\left(j\right)}<f_{{\rm RSD}}\left(\mathbf{Q}^{\prime}\right)_{\tau\left(h_{2}\right),\pi\left(j\right)}.
\]
Note that $\tau\left(h_{2}\right)$ is the second most-preferred house
of $\pi\left(j\right)$ in $\mathbf{Q}^{\prime}$. The only agent
whose probability of obtaining his second most-preferred house can
decrease under the adjustment from $f_{{\rm RSD}}$ to $g$ is agent
$m$, hence $\pi\left(j\right)=m$.

Such a decrease can occur only if agent $j$ ranks $\tau\left(h_{2}\right)$
second among $H\setminus\left\{ a,b\right\} $ in $\mathbf{Q}^{\prime}$.
Consequently, agent $j$ must rank $\tau\left(h_{2}\right)$ among
his top four houses in $\mathbf{Q}^{\prime}$, so agent $\pi^{-1}\left(j\right)$
must rank $h_{2}$ among his top four houses in $\mathbf{Q}$. By
the construction of $\mathbf{Q}$, this forces $\pi^{-1}\left(j\right)\in\left\{ 2,j\right\} $;
since $\pi\left(j\right)\neq j$, we have $\pi^{-1}\left(j\right)=2$.

In $\mathbf{Q}$, the top three choices of agent $2$ are $aQ_{2}bQ_{2}h_{2}$;
therefore, in $\mathbf{Q}^{\prime}$ the top three choices of $\pi\left(2\right)=j$
are $\tau\left(a\right)Q_{j}^{\prime}\tau\left(b\right)Q_{j}^{\prime}\tau\left(h_{2}\right)$.
Because $j$ ranks $\tau\left(h_{2}\right)$ second among $H\setminus\left\{ a,b\right\} $
in $\mathbf{Q}^{\prime}$, we must have $\left\{ \tau\left(a\right),\tau\left(b\right)\right\} \neq\left\{ a,b\right\} $.

Since $\mathbf{Q}^{\prime}\in\mathsf{P}$ (otherwise no adjustment
would occur), agents $1,\dots,m-2$ must rank $a$ and $b$ as their
top two houses in $\mathbf{Q}^{\prime}$. Hence agents $\pi^{-1}\left(1\right),\dots,\pi^{-1}\left(m-2\right)$
rank $\tau^{-1}\left(a\right)$ and $\tau^{-1}\left(b\right)$ as
their top two houses in $\mathbf{Q}$. We obtain a contradiction:
in $\mathbf{Q}$ every agent $i\neq j$ has $\left\{ a,b\right\} $
as his top two houses, and since $\left\{ \tau^{-1}\left(a\right),\tau^{-1}\left(b\right)\right\} \neq\left\{ a,b\right\} $,
only $j$ could have $\left\{ \tau^{-1}\left(a\right),\tau^{-1}\left(b\right)\right\} $
as his top two. Yet $\pi^{-1}\left(1\right),\dots,\pi^{-1}\left(m-2\right)$
are $m-2\geq4$ distinct agents, so they cannot all be $j$. Therefore
the claimed inequality holds for every $\left(\pi,\tau\right)\in\Pi\times\Gamma$.

Finally, since $f$ is the average of $\left\{ \left(\pi,\tau\right)\left(g\right)\right\} _{\left(\pi,\tau\right)\in\Pi\times\Gamma}$
and $f_{{\rm RSD}}$ is the average of $\left\{ \left(\pi,\tau\right)\left(f_{{\rm RSD}}\right)\right\} _{\left(\pi,\tau\right)\in\Pi\times\Gamma}$,
we obtain 
\[
f\left(\mathbf{Q}\right)_{h_{2},j}>f_{{\rm RSD}}\left(\mathbf{Q}\right)_{h_{2},j},
\]
and hence $f\neq f_{{\rm RSD}}$, as desired.
\end{proof}

\appendix
\section{Adding more axioms: persistence of non-uniqueness} \label{sec:appendix-adding_axioms}
Section \lyxdeleted{bzmao}{Mon Jan  5 23:46:05 2026}{5.3}\lyxadded{bzmao}{Mon Jan  5 23:46:08 2026}{}\lyxadded{bzmao}{Mon Jan  5 23:46:16 2026}{\ref{subsec:adding_BI}}
established that anonymity, neutrality, ExPE, SP and BI do not uniquely
characterize $f_{{\rm RSD}}$. To pursue uniqueness, we introduce
two additional axioms, both satisfied by $f_{{\rm RSD}}$, and examine
whether the extended set of axioms suffices for uniqueness. However,
this attempt fails: the mechanism from \lyxdeleted{bzmao}{Mon Jan  5 19:54:20 2026}{Lemma
5.3}\lyxadded{bzmao}{Mon Jan  5 19:54:20 2026}{Proposition }\lyxadded{bzmao}{Mon Jan  5 19:54:20 2026}{\ref{prop:BI_non-uniqueness}}
also satisfies one of these new axioms, and under a slightly stricter
size condition, $n\geq m-1\geq7$ (compared to $n\geq m-1\geq5$ in
\lyxdeleted{bzmao}{Mon Jan  5 19:54:26 2026}{Lemma 5.3}\lyxadded{bzmao}{Mon Jan  5 19:54:26 2026}{Proposition
}\lyxadded{bzmao}{Mon Jan  5 19:54:26 2026}{\ref{prop:BI_non-uniqueness}}),
we construct another mechanism, similar in spirit to that of \lyxdeleted{bzmao}{Mon Jan  5 19:54:31 2026}{Lemma
5.3}\lyxadded{bzmao}{Mon Jan  5 19:54:31 2026}{Proposition }\lyxadded{bzmao}{Mon Jan  5 19:54:31 2026}{\ref{prop:BI_non-uniqueness}},
which satisfies all the axioms (those from the \lyxdeleted{bzmao}{Mon Jan  5 19:54:43 2026}{lemma}\lyxadded{bzmao}{Mon Jan  5 19:54:45 2026}{proposition}
together with the two new ones) and still differs from $f_{{\rm RSD}}$.
We now define one of these additional axioms.
\begin{defn}
[NB] A mechanism $f$ satisfies \emph{Non-Bossiness (NB)} if, for
every preference profile $\mathbf{P}\in\mathcal{R}^{N}$, agent $i\in N$,
and preference order $P_{i}^{\prime}\in\mathcal{R}$, 
\[
f\left(\mathbf{P}\right)_{i}=f\left(\mathbf{P}_{-i},P_{i}^{\prime}\right)_{i}\implies f\left(\mathbf{P}\right)=f\left(\mathbf{P}_{-i},P_{i}^{\prime}\right).
\]
\end{defn}

In words, an agent cannot alter any other agent's assignment without
simultaneously changing his own. The mechanism $f_{{\rm RSD}}$ satisfies
NB (see \cite{bade2016fairness}). Before showing that the mechanism
from \lyxdeleted{bzmao}{Mon Jan  5 19:54:57 2026}{Lemma 5.3}\lyxadded{bzmao}{Mon Jan  5 19:54:57 2026}{Proposition
}\lyxadded{bzmao}{Mon Jan  5 19:54:57 2026}{\ref{prop:BI_non-uniqueness}}
satisfies this axiom, we first establish an auxiliary claim that will
be useful in the proof. The claim shows that for mechanisms satisfying
SP, it is sufficient to verify NB only for preference orders that
differ by an adjacent swap.
\begin{claim}
Let $f$ be a mechanism that satisfies SP. Suppose that for every
preference profile $\mathbf{P}\in\mathcal{R}^{N}$, agent $i\in N$,
and preference order $P_{i}^{\prime\prime}\in\mathcal{R}$ obtained
from $P_{i}$ by an adjacent swap, 
\[
f\left(\mathbf{P}\right)_{i}=f\left(\mathbf{P}_{-i},P_{i}^{\prime\prime}\right)_{i}\implies f\left(\mathbf{P}\right)=f\left(\mathbf{P}_{-i},P_{i}^{\prime\prime}\right).
\]
Then $f$ satisfies NB.
\end{claim}

\begin{proof}
Fix $i\in N$. For every $P_{i},P_{i}^{\prime}\in\mathcal{R}$, denote
by $\left(R_{i}^{\ell}\right)_{\ell=1}^{k}$ the sequence of preference
orders satisfying $R_{i}^{1}=P_{i}$, $R_{i}^{k}=P_{i}^{\prime}$,
and for every $1\leq\ell<k$, $R_{i}^{\ell+1}$ is obtained from $R_{i}^{\ell}$
by an adjacent swap that elevates the house $h_{\ell}$, defined as
the house satisfying $C_{R_{i}^{\ell}}\left(h_{\ell}\right)\neq C_{P_{i}^{\prime}}\left(h_{\ell}\right)$
and, for every $h\in C_{P_{i}^{\prime}}\left(h_{\ell}\right)\setminus\left\{ h_{\ell}\right\} $,
$C_{R_{i}^{\ell}}\left(h\right)=C_{P_{i}^{\prime}}\left(h\right)$.

We show that for every $\mathbf{P}\in\mathcal{R}^{N}$ and $P_{i}^{\prime}\in\mathcal{R}$
such that $f\left(\mathbf{P}\right)_{i}=f\left(\mathbf{P}_{-i},P_{i}^{\prime}\right)_{i}$,
it also holds that $f\left(\mathbf{P}\right)=f\left(\mathbf{P}_{-i},P_{i}^{\prime}\right)$,
by induction on $k$ (the length of the sequence from $P_{i}$ to
$P_{i}^{\prime}$). The base case $k=1$ is immediate, since $P_{i}^{\prime}=P_{i}$.
For the inductive step, assume the claim holds for all sequences of
length smaller than $k$. Denote $P_{i}^{\prime\prime}\coloneqq R_{i}^{2}$.
Note that $h_{1}$ is a house that does not move downward in any of
the adjacent swaps corresponding to the sequence from $P_{i}$ to
$P_{i}^{\prime}$. Since $f$ satisfies SP, we have 
\[
f\left(\mathbf{P}\right)_{h_{1},i}\leq f\left(\mathbf{P}_{-i},P_{i}^{\prime\prime}\right)_{h_{1},i}\leq f\left(\mathbf{P}_{-i},P_{i}^{\prime}\right)_{h_{1},i}.
\]
Because $f\left(\mathbf{P}\right)_{i}=f\left(\mathbf{P}_{-i},P_{i}^{\prime}\right)_{i}$,
these inequalities must in fact be equalities. Moreover, since $f$
satisfies SP, this implies $f\left(\mathbf{P}\right)_{i}=f\left(\mathbf{P}_{-i},P_{i}^{\prime\prime}\right)_{i}$.
By the assumption of the claim, we then have $f\left(\mathbf{P}\right)=f\left(\mathbf{P}_{-i},P_{i}^{\prime\prime}\right)$,
and since the sequence from $P_{i}^{\prime\prime}$ to $P_{i}^{\prime}$
has length less than $k$, the induction hypothesis gives $f\left(\mathbf{P}_{-i},P_{i}^{\prime\prime}\right)=f\left(\mathbf{P}_{-i},P_{i}^{\prime}\right)$.
Combining these equalities yields the desired result.
\end{proof}
We now show that the mechanism from \lyxdeleted{bzmao}{Mon Jan  5 19:55:05 2026}{Lemma
5.3}\lyxadded{bzmao}{Mon Jan  5 19:55:05 2026}{Proposition }\lyxadded{bzmao}{Mon Jan  5 19:55:05 2026}{\ref{prop:BI_non-uniqueness}}
also satisfies this axiom.
\begin{claim}
The mechanism $f$ from \lyxdeleted{bzmao}{Mon Jan  5 19:55:07 2026}{Lemma
5.3}\lyxadded{bzmao}{Mon Jan  5 19:55:07 2026}{Proposition }\lyxadded{bzmao}{Mon Jan  5 19:55:07 2026}{\ref{prop:BI_non-uniqueness}}
satisfies NB (for small enough $\varepsilon$).
\end{claim}

\begin{proof}
Let $\mathbf{P}\in\mathcal{R}^{N}$ and $i\in N$. Since $f$ satisfies
SP, the previous claim implies that it is sufficient to show that
\[
f\left(\mathbf{P}\right)_{i}=f\left(\mathbf{P}_{-i},P_{i}^{\prime}\right)_{i}\implies f\left(\mathbf{P}\right)=f\left(\mathbf{P}_{-i},P_{i}^{\prime}\right)
\]
for every preference order $P_{i}^{\prime}$ that can be obtained
from $P_{i}$ by an adjacent swap. Let $P_{i}^{\prime}\in\mathcal{R}$
be such an order and denote $\mathbf{P}^{\prime}\coloneqq\left(\mathbf{P}_{-i},P_{i}^{\prime}\right)$.
We will show that if $f\left(\mathbf{P}\right)_{i}=f\left(\mathbf{P}^{\prime}\right)_{i}$,
then $f\left(\mathbf{P}\right)=f\left(\mathbf{P}^{\prime}\right)$.
Recall that $f_{{\rm RSD}}$ satisfies NB \cite{bade2016fairness}.
Define $d\coloneqq g-f_{{\rm RSD}}$ (that is, $d\left(\mathbf{P}\right)=g\left(\mathbf{P}\right)-f_{{\rm RSD}}\left(\mathbf{P}\right)$
for every $\mathbf{P}\in\mathcal{R}^{N}$) and $D\coloneqq f-f_{{\rm RSD}}$.
Since $f$ is the symmetrization of $g$ (and $f_{{\rm RSD}}$ is
the symmetrization of itself), $D$ is the symmetrization of $d$,
i.e. 
\[
D=\frac{1}{n!m!}\sum_{\left(\pi,\tau\right)\in\Pi\times\Gamma}\left(\pi,\tau\right)\left(d\right).
\]
Here, the notation $\left(\pi,\tau\right)\left(d\right)$ is used
in the same sense as for mechanisms, even though $d$ is not itself
a mechanism; specifically, $\left(\pi,\tau\right)\left(d\right)$
denotes the difference $\left(\pi,\tau\right)\left(g\right)-\left(\pi,\tau\right)\left(f_{{\rm RSD}}\right)$.
Moreover, by the definition of $g$, for every $\mathbf{P}\in\mathcal{R}^{N}$,
$i\in N$, and $h\in H$, $\left|d\left(\mathbf{P}\right)_{h,i}\right|\leq\varepsilon$,
and we assume here that $\varepsilon<\frac{1}{2n!}$.

We first show that if $f\left(\mathbf{P}\right)_{i}=f\left(\mathbf{P}^{\prime}\right)_{i}$,
then $f_{{\rm RSD}}\left(\mathbf{P}\right)_{i}=f_{{\rm RSD}}\left(\mathbf{P}^{\prime}\right)_{i}$.
Assume, for contradiction, that this is not the case. Then, there
exists $h\in H$ such that $f_{{\rm RSD}}\left(\mathbf{P}\right)_{h,i}\neq f_{{\rm RSD}}\left(\mathbf{P}^{\prime}\right)_{h,i}$.
Since $f_{{\rm RSD}}$ assigns each deterministic SD mechanism a probability
of $\frac{1}{n!}$, it follows that 
\[
\left|f_{{\rm RSD}}\left(\mathbf{P}\right)_{h,i}-f_{{\rm RSD}}\left(\mathbf{P}^{\prime}\right)_{h,i}\right|\geq\frac{1}{n!}.
\]
By the definition of $D$, this implies 
\[
\left|D\left(\mathbf{P}\right)_{h,i}-D\left(\mathbf{P}^{\prime}\right)_{h,i}\right|\geq\frac{1}{n!}.
\]
On the other hand, 
\begin{align*}
\left|D\left(\mathbf{P}\right)_{h,i}-D\left(\mathbf{P}^{\prime}\right)_{h,i}\right| & \leq\frac{1}{n!m!}\sum_{\left(\pi,\tau\right)\in\Pi\times\Gamma}\left|d\left(\left(\pi,\tau\right)\left(\mathbf{P}\right)\right)_{\tau\left(h\right),\pi\left(i\right)}\right|+\left|d\left(\left(\pi,\tau\right)\left(\mathbf{P}^{\prime}\right)\right)_{\tau\left(h\right),\pi\left(i\right)}\right|\\
 & \leq\frac{1}{n!m!}\sum_{\left(\pi,\tau\right)\in\Pi\times\Gamma}2\varepsilon\\
 & <\frac{1}{n!}
\end{align*}
which contradicts the previous inequality. Therefore, $f_{{\rm RSD}}\left(\mathbf{P}\right)_{i}=f_{{\rm RSD}}\left(\mathbf{P}^{\prime}\right)_{i}$.

Since $f_{{\rm RSD}}$ satisfies NB, we also have $f_{{\rm RSD}}\left(\mathbf{P}\right)=f_{{\rm RSD}}\left(\mathbf{P}^{\prime}\right)$.
Thus, to establish $f\left(\mathbf{P}\right)=f\left(\mathbf{P}^{\prime}\right)$,
it suffices to show that $D\left(\mathbf{P}\right)=D\left(\mathbf{P}^{\prime}\right)$.
Assume, for contradiction, that this is not the case. Because $f\left(\mathbf{P}\right)_{i}=f\left(\mathbf{P}^{\prime}\right)_{i}$
and $f_{{\rm RSD}}\left(\mathbf{P}\right)_{i}=f_{{\rm RSD}}\left(\mathbf{P}^{\prime}\right)_{i}$,
we obtain $D\left(\mathbf{P}\right)_{i}=D\left(\mathbf{P}^{\prime}\right)_{i}$.
Hence, there exists an agent $k\in N\setminus\left\{ i\right\} $
and a house $h\in H$ such that $D\left(\mathbf{P}\right)_{h,k}\neq D\left(\mathbf{P}^{\prime}\right)_{h,k}$.
Since $D$ is the symmetrization of $d$, there must exist some $\left(\pi,\tau\right)\in\Pi\times\Gamma$
such that $\left(\pi,\tau\right)\left(d\right)\left(\mathbf{P}\right)_{h,k}\neq\left(\pi,\tau\right)\left(d\right)\left(\mathbf{P}^{\prime}\right)_{h,k}$.

Without loss of generality, we may assume $\left(\pi,\tau\right)=\left({\rm id},{\rm id}\right)$;
otherwise, we can replace $\mathbf{P},\mathbf{P}^{\prime},i,k,h$
with $\left(\pi,\tau\right)\left(\mathbf{P}\right),\left(\pi,\tau\right)\left(\mathbf{P}^{\prime}\right),\pi\left(i\right),\pi\left(k\right),\tau\left(h\right)$
and apply the same argument. Thus, $d\left(\mathbf{P}\right)_{h,k}\neq d\left(\mathbf{P}^{\prime}\right)_{h,k}$.
We will now show that this contradicts the earlier conclusion that
$f_{{\rm RSD}}\left(\mathbf{P}\right)_{i}=f_{{\rm RSD}}\left(\mathbf{P}^{\prime}\right)_{i}$.

Since agents in $N\setminus\left[m-1\right]$ cannot influence the
adjustment by changing their preferences, we may assume $i\in\left[m-1\right]$.
Moreover, when agent $i$ performs an adjacent swap that elevates
some $h^{\prime}\in H$ in his ranking, the set of orderings where
agent $i$ receives $h^{\prime}$ in the post-swap profile contains
the corresponding set before the swap. Hence, to establish $f_{{\rm RSD}}\left(\mathbf{P}\right)_{h^{\prime},i}\neq f_{{\rm RSD}}\left(\mathbf{P}^{\prime}\right)_{h^{\prime},i}$,
it suffices to exhibit an ordering in which agent $i$ receives $h^{\prime}$
only after the swap.

If $i\in\left[m-2\right]$, the swap in agent $i$'s preference order
can affect the adjustment only when exactly one of $\mathbf{P}$ or
$\mathbf{P}^{\prime}$ belongs to $\mathsf{P}$. Without loss of generality,
suppose $\mathbf{P}\in\mathsf{P}$. Then, the top three houses in
$P_{i}$ are $aP_{i}bP_{i}h_{i}$, and the swap involves one of the
following adjacent pairs: $\left(a,b\right)$, $\left(b,h_{i}\right)$,
or $\left(h_{i},h\right)$ for some $h\in H\setminus\left\{ a,b,h_{i}\right\} $.
Let $i^{\prime},i^{\prime\prime}\in\left[m-2\right]\setminus\left\{ i\right\} $
be two distinct agents. For each of the adjacent pairs mentioned,
there exists an ordering where agent $i$ receives the elevated house
only after the swap: if the pair is $\left(a,b\right)$, let $i$
go first; if $\left(b,h_{i}\right)$, start with $i^{\prime}i$; and
if $\left(h_{i},h\right)$ for some $h\in H\setminus\left\{ a,b,h_{i}\right\} $,
start with $i^{\prime\prime}i^{\prime}i$.

Otherwise, $i=j$. In this case, the swap in agent $j$'s preference
order can affect the adjustment only when both $\mathbf{P}$ and $\mathbf{P}^{\prime}$
belong to $\mathsf{P}$. Without loss of generality, assume $h_{\ell}P_{j}h_{\ell+1}$
for every $\ell\in\left[m-3\right]$; then the swap must involve one
of the adjacent pairs $\left(h_{1},h_{2}\right)$, $\left(h_{2},h_{3}\right)$,
or $\left(h_{3},h_{4}\right)$. For each of these pairs, there exists
an ordering where agent $j$ receives the elevated house only after
the swap: for $\left(h_{1},h_{2}\right)$, consider an ordering starting
with $21j$; for $\left(h_{2},h_{3}\right)$, an ordering starting
with $321j$; and for $\left(h_{3},h_{4}\right)$, an ordering starting
with $4321j$.

In all cases, we therefore obtain $f_{{\rm RSD}}\left(\mathbf{P}\right)_{i}\neq f_{{\rm RSD}}\left(\mathbf{P}^{\prime}\right)_{i}$.
This contradiction implies that our assumption was false, and hence
$D\left(\mathbf{P}\right)=D\left(\mathbf{P}^{\prime}\right)$. Therefore,
$f\left(\mathbf{P}\right)=f\left(\mathbf{P}^{\prime}\right)$, and
$f$ satisfies NB, as required.
\end{proof}
We now introduce the second new axiom. 
\begin{defn}
[CM] A mechanism $f$ satisfies \emph{Cross Monotonicity (CM)} if,
for every preference profile $\mathbf{P}\in\mathcal{R}^{N}$, agent
$i\in N$, pair of adjacent houses $h^{\prime}P_{i}^{+}h$, and agent
$j\in N\setminus\left\{ i\right\} $, 
\[
f\left(\mathbf{P}_{-i},P_{i}^{h}\right)_{h,j}\leq f\left(\mathbf{P}\right)_{h,j},
\]
where $P_{i}^{h}$ is the preference order obtained from $P_{i}$
by swapping $h$ with $h^{\prime}$.
\end{defn}

In words, if agent $i$ swaps house $h$ with the house immediately
above it in his preference order, then every other agent's probability
of receiving $h$ weakly decreases. To the best of our knowledge,
the CM axiom has not appeared in the literature on assignment mechanisms,
though a similar notion has been studied in the context of cost sharing
\cite{immorlica2008limitations}. We next observe that $f_{{\rm RSD}}$
satisfies the CM axiom, as stated in the following claim.
\begin{claim}
$f_{{\rm RSD}}$ satisfies CM.
\end{claim}

\begin{proof}
Since the distribution of ${\rm RSD}$ over the SD mechanisms does
not depend on the preference profile, it suffices to show that for
every two distinct agents $i,i^{\prime}\in N$, profile $\mathbf{P}\in\mathcal{R}^{N}$,
and pair of adjacent houses $hP_{i}^{+}h^{\prime}$, the set of orderings
in which $i^{\prime}$ receives $h$ under $\mathbf{P}$ is contained
in the corresponding set under the profile obtained from $\mathbf{P}$
after agent $i$ swaps $h$ with $h^{\prime}$.

Indeed, if $i^{\prime}$ appears before $i$, a change in $i$'s preferences
does not affect $i^{\prime}$'s outcome. If $i^{\prime}$ receives
$h$ when arriving after $i$ under the profile $\mathbf{P}$, then
$h$ must have been available to $i$, but $i$ chose some other house
$h^{\prime\prime}$ such that $h^{\prime\prime}P_{i}^{+}h$. Agent
$i$ would make the same choice after swapping $h$ with $h^{\prime}$,
and therefore $i^{\prime}$ continues to receive $h$ in all such
orderings, as required.
\end{proof}
Having defined the new axioms, we now show that there exists a mechanism
other than $f_{{\rm RSD}}$ satisfying all of them.
\begin{prop}
When $n\geq m-1\geq7$, $f_{{\rm RSD}}$ is not the unique mechanism
satisfying anonymity, neutrality, ExPE, SP, BI, NB, and CM.
\end{prop}

\begin{proof}
We construct a suitable mechanism $f$. As in \lyxdeleted{bzmao}{Mon Jan  5 19:52:44 2026}{Lemma
5.3}\lyxadded{bzmao}{Mon Jan  5 19:52:47 2026}{Proposition }\lyxadded{bzmao}{Mon Jan  5 19:53:07 2026}{\ref{prop:BI_non-uniqueness}},
let $H\coloneqq\left\{ a,b,h_{1},\dots,h_{m-2}\right\} $, and let
$j\in N$ denote the agent with index $m-1$. Define $\mathsf{P}$
exactly as before, and let $x=x_{\mathbf{P}}$ and $y=y_{\mathbf{P}}$
denote agent $j$'s third and fourth choices among $H\setminus\left\{ a,b\right\} $,
respectively.

Fix $\varepsilon>0$ sufficiently small (as in the previous claim,
we assume $\varepsilon<\frac{1}{2n!}$), and define the vector $v=\left(v_{h}\right)_{h\in H}$
and the mechanism $g$ as in \lyxdeleted{bzmao}{Mon Jan  5 19:53:39 2026}{Lemma
5.3}\lyxadded{bzmao}{Mon Jan  5 19:53:39 2026}{Proposition }\lyxadded{bzmao}{Mon Jan  5 19:53:39 2026}{\ref{prop:BI_non-uniqueness}}.
We can again assume, without loss of generality, that $h_{i}P_{j}h_{i+1}$
for every $i\in\left[m-3\right]$, so that $x_{\mathbf{P}}=h_{3}$
and $y_{\mathbf{P}}=h_{4}$.

To verify that $g$ is well defined, we apply the same reasoning as
in \lyxdeleted{bzmao}{Mon Jan  5 19:53:48 2026}{Lemma 5.3}\lyxadded{bzmao}{Mon Jan  5 19:53:48 2026}{Proposition
}\lyxadded{bzmao}{Mon Jan  5 19:53:48 2026}{\ref{prop:BI_non-uniqueness}}.
In this case, agent $j$ receives $h_{4}$ in an ordering starting
with $54321j$. In an ordering beginning with $321j$, followed by
the agents in $\left\{ 4,\dots,m-2\right\} $ in some order, $h_{3}$
remains the only unassigned house afterward. Hence, $h_{3}$ remains
unassigned when $n=m-1$; and when $n\geq m$, we can place agent
$m$ immediately after these agents, assigning him $h_{3}$. Therefore,
$g$ is well defined.

Define $f$ as the symmetrization of $g$ over all renamings of agents
and houses. We show that $f$ is the desired mechanism. By construction,
$f$ satisfies anonymity and neutrality. For ExPE, SP, BI, and CM,
it suffices to verify that $g$ satisfies them. The arguments for
the first three are similar to those in \lyxdeleted{bzmao}{Mon Jan  5 19:53:58 2026}{Lemma
5.3}\lyxadded{bzmao}{Mon Jan  5 19:53:58 2026}{Proposition }\lyxadded{bzmao}{Mon Jan  5 19:53:58 2026}{\ref{prop:BI_non-uniqueness}},
and we highlight below the necessary modifications that allow the
reasoning to be adapted to the present setting. For BI, we use the
same reasoning as before, noting that the three adjacent pairs of
houses in $j$'s preference order that influence the adjustment are
$\left(h_{k},h_{k+1}\right)$ for $k=2,3,4$. For SP, the pairs of
adjacent houses that modify the adjustment in the wrong direction
(i.e., in a way favorable to agent $j$) are $\left(h_{2},h_{3}\right)$
and $\left(h_{4},h_{5}\right)$. As before, it suffices to exhibit
an ordering where agent $j$ receives a house in $C_{P_{j}}\left(h_{2}\right)$
before the swap of $h_{2}$ and $h_{3}$ but not afterward, and an
ordering where he receives a house in $C_{P_{j}}\left(h_{4}\right)$
before the swap of $h_{4}$ and $h_{5}$ but not afterward. Two such
orderings are those whose initial segments are $321j$ and $54321j$.

For ExPE, the adjustment can be described by a transfer of probabilities
that decreases $\varepsilon$ from each of $s_{1}$ and $s_{2}$ and
adds $\varepsilon$ to each of $s_{1}^{\prime}$ and $s_{2}^{\prime}$,
where the explicit form of these assignments is given below. 
\begin{align*}
s_{1}\left(i\right)=\begin{cases}
a & i=5,\\
b & i=4,\\
h_{4} & i=j,\\
h_{i} & i\in\left[m-2\right]\setminus\left\{ 4,5\right\} ,\\
h_{5} & i=m,\\
\varnothing & \text{otherwise},
\end{cases} & s_{2}\left(i\right)=\begin{cases}
a & i=3,\\
b & i=1,\\
h_{1} & i=j,\\
h_{i} & i\in\left[m-2\right]\setminus\left\{ 1,3\right\} ,\\
h_{3} & i=m,\\
\varnothing & \text{otherwise},
\end{cases}\\
s_{1}^{\prime}\left(i\right)=\begin{cases}
a & i=5,\\
b & i=1,\\
h_{1} & i=j,\\
h_{i} & i\in\left[m-2\right]\setminus\left\{ 1,5\right\} ,\\
h_{5} & i=m,\\
\varnothing & \text{otherwise},
\end{cases} & s_{2}^{\prime}\left(i\right)=\begin{cases}
a & i=3,\\
b & i=4,\\
h_{3} & i=j,\\
h_{i} & i\in\left[m-2\right]\setminus\left\{ 3,4\right\} ,\\
h_{4} & i=m,\\
\varnothing & \text{otherwise}.
\end{cases}
\end{align*}
These assignments are efficient, as they correspond to orderings beginning
with $54321j$, $31j$, $51j$, $3412j$, followed by the remaining
agents in $\left[m-2\right]$ in some order, and finally (if $n\geq m$)
agent $m$, who receives the remaining house.

Since $f$ satisfies SP and has a structure similar to that of the
mechanism in \lyxdeleted{bzmao}{Mon Jan  5 19:54:04 2026}{Lemma 5.3}\lyxadded{bzmao}{Mon Jan  5 19:54:04 2026}{Proposition
}\lyxadded{bzmao}{Mon Jan  5 19:54:04 2026}{\ref{prop:BI_non-uniqueness}},
the same reasoning as in the previous claim shows that $f$ also satisfies
NB.

We next show that $g$ satisfies CM. Because $f_{{\rm RSD}}$ satisfies
CM, using arguments analogous to those employed to establish SP and
BI, it suffices to focus on adjacent swaps that influence the adjustment.
For each such swap, let $i$ denote the agent performing the swap,
$h$ the house that moves up in $i$'s preference order, and $i^{\prime}$
another agent whose adjustment changes in the wrong direction, that
is, in a way that increases the probability that $i^{\prime}$ receives
$h$ when $i$ raises $h$ in his preference order. For each such
agent $i^{\prime}$, it is enough to exhibit an ordering that overrides
this change, namely one in which $i^{\prime}$ receives $h$ before
the swap but not after $i$ raises it in his preference order. This
will suffice, because the set of orderings in which $i^{\prime}$
receives $h$ when $i$ ranks it higher is contained in the corresponding
set of orderings when $i$ ranks it lower, by the same reasoning as
in the previous claim.

Observe that only agents in $\left[m-1\right]$ can make an adjacent
swap that affects the adjustment, and this occurs only when either
the profile before the swap or the one after it belongs to $\mathsf{P}$.
By considering the opposite inequality for the house that moves down
in the swap, we may assume without loss of generality that the profile
before the swap belongs to $\mathsf{P}$; we denote this profile by
$\mathbf{P}$. Moreover, the only agents affected by the adjustment
are $j$ and $m$ (if the latter exists), and since the adjustment
for these two agents always changes in opposite directions, it cannot
be in the undesired direction for both simultaneously. Hence, for
every swap we consider, there is at most one agent (either $j$ or
$m$) for whom the change is in the undesired direction, and for that
agent we will need to exhibit an ordering that overrides this change.
We distinguish between cases according to the agent performing the
swap.
\begin{itemize}
\item Agent $j$: In this case, we examine only the adjustment of agent
$m$. As noted earlier, the only swaps that agent $j$ can make that
affect the adjustment are between the adjacent pairs $\left(h_{k},h_{k+1}\right)$
for $k=2,3,4$.
\begin{itemize}
\item When $k=3$ (that is, $j$ changes his preferences to $P_{j}^{h_{4}}$),
the change in the adjustment for agent $m$ is in the correct direction
for both houses: a negative change for $h_{4}$ and a positive one
for $h_{3}$.
\item When $k\in\left\{ 2,4\right\} $, the change in the adjustment for
agent $m$ is in the undesired direction for both houses. Consider
an ordering in which agents $k$ and $k+1$ appear first, followed
by the remaining agents in $\left[m-2\right]$, then agent $j$, and
finally agent $m$. In this ordering, agent $m$ receives the house
that agent $j$ prefers less among $h_{k}$ and $h_{k+1}$. This construction
simultaneously demonstrates that there exists an ordering where agent
$m$ receives $h_{k+1}$ before the swap but not afterward, and an
ordering where he does not receive $h_{k}$ before the swap but does
receive it afterward, as required.
\end{itemize}
\item Agents in $\left[m-2\right]$: a swap performed by these agents affects
the adjustment only if it involves one of their top three houses.
The only houses for which the adjustment changes are $h_{3}$ and
$h_{4}$; thus, at least one of them must participate in the swap.
We analyze the cases for each house separately.
\begin{itemize}
\item When $h_{3}$ is involved in the swap, we distinguish cases according
to the identity of the swapping agent.
\begin{itemize}
\item When the swapping agent is $i=3$: house $h_{3}$ can be swapped either
upward (with $b$) or downward (with some $h_{k}$ for $k\in\left[m-2\right]\setminus\left\{ 3\right\} $).
\begin{itemize}
\item If $h_{3}$ is swapped with $b$, the change in the adjustment is
in the wrong direction for agent $m$. We therefore need an ordering
where agent $m$ receives $h_{3}$ before the swap but not afterward.
An ordering beginning with $13j$, followed by the other agents in
$\left[m-2\right]$ in some order, and then agent $m$, satisfies
this condition.
\item If $h_{3}$ is swapped downward with some $h_{k}$ for $k\in\left[m-2\right]\setminus\left\{ 3\right\} $,
the change in the adjustment is not in the correct direction for agent
$j$. Hence, we require an ordering where agent $j$ does not receive
$h_{3}$ before the swap but does receive it afterward. If $k\in\left\{ 1,2\right\} $,
consider an ordering beginning with $k43k^{\prime}j$, where $k^{\prime}=3-k$;
if $k\notin\left\{ 1,2\right\} $, consider an ordering starting with
$45123j$.
\end{itemize}
\item When the swapping agent is $i\in\left[m-2\right]\setminus\left\{ 3\right\} $:
here, we consider only the swap where agent $i$ exchanges $h_{3}$
with the house immediately above it, and only when that house is $h_{i}$,
since in any other swap of agent $i$ involving $h_{3}$, the swap
does not change the adjustment. In this case, the change is not in
the correct direction for agent $m$, so we must find an ordering
where agent $m$ receives $h_{3}$ before the swap but not afterward.
When $i\neq1$, an ordering beginning with $13ij$, followed by the
other agents in $\left[m-2\right]$ in some order, and finally agent
$m$, satisfies this requirement. When $i=1$, a similar ordering
with initial segment $231j$ serves the same purpose.
\end{itemize}
\item When $h_{4}$ is involved in the swap, we proceed in a manner analogous
to the case of $h_{3}$.
\begin{itemize}
\item When the swapping agent is $i=4$:
\begin{itemize}
\item If $h_{4}$ is swapped upward with $b$, the adjustment for agent
$j$ is not in the correct direction. In this case, an ordering beginning
with $54321j$ is one in which agent $j$ receives $h_{4}$ before
the swap but not afterward.
\item If $h_{4}$ is swapped downward with some $h_{k}$ for $k\in\left[m-2\right]\setminus\left\{ 4\right\} $,
the adjustment for agent $m$ is not in the correct direction. When
$k\neq1$, an ordering beginning with $1kj$, followed by the other
agents in $\left[m-2\right]$, and then agent $m$, is one in which
agent $m$ receives $h_{4}$ only after the swap. When $k=1$, we
consider a similar ordering beginning with $214j$.
\end{itemize}
\item When the swapping agent is $i\in\left[m-2\right]\setminus\left\{ 4\right\} $:
we consider only the case in which $h_{4}$ is swapped upward with
$h_{i}$. In this case, the change is not in the correct direction
for agent $j$. Let $i^{\prime},i^{\prime\prime}\in\left[m-2\right]\setminus\left\{ 1,2,3,i\right\} $
be two distinct agents (such agents exist since $m\geq8$). An ordering
beginning with $i^{\prime}i^{\prime\prime}$, followed by the agents
in $\left\{ 1,2,3,i\right\} $ in some order (note that $i$ may belong
to $\left\{ 1,2,3\right\} $), and then agent $j$, is one in which
agent $j$ receives $h_{4}$ only before the swap.
\end{itemize}
\end{itemize}
\end{itemize}
Thus, $g$ satisfies CM, and consequently $f$ satisfies all the axioms
stated in the \lyxdeleted{bzmao}{Mon Jan  5 19:04:30 2026}{lemma}\lyxadded{bzmao}{Mon Jan  5 19:04:32 2026}{proposition}.

Finally, to show that $f\neq f_{{\rm RSD}}$, we apply an argument
similar to that used in \lyxdeleted{bzmao}{Mon Jan  5 19:05:03 2026}{Lemma
5.3}\lyxadded{bzmao}{Mon Jan  5 19:05:07 2026}{Proposition }\lyxadded{bzmao}{Mon Jan  5 19:06:20 2026}{\ref{prop:BI_non-uniqueness}}.
Consider a profile $\mathbf{Q}\in\mathsf{P}$ such that agent $j$'s
top choices are $h_{1}Q_{j}h_{2}Q_{j}h_{3}$; for every $i\in N\setminus\left\{ j\right\} $,
the top choices are $aQ_{i}b$; and for every $i\in N\setminus\left\{ 3,j\right\} $,
the house $h_{3}$ is not among agent $i$'s top five choices. Applying
similar arguments to this profile establishes the desired distinction
between $f$ and $f_{{\rm RSD}}$.
\end{proof}

\section{Exhaustive case analysis for \texorpdfstring{$n=m=4$}{n=m=4}} \label{sec:nEqmEq4_cases}
Here, we exhaustively verify that for every profile, the assignment
matrix induced by any mechanism satisfying ExPE, ETE, and SP (hereafter,
\emph{the axioms}) is uniquely determined. Let $H\coloneqq\left\{ a,b,c,d\right\} $.
We use the shorthand $xyzw$ to denote $xP_{i}^{+}yP_{i}^{+}zP_{i}^{+}w$
when the ranking $P_{i}$ is clear from the context. A preference
profile is represented by a $4\times4$ table, where the $i$-th column
specifies agent $i$'s ranking. We use $x,y,z$ to indicate multiple
possible rankings or profiles, where $x,y,z$ may be replaced by some
houses such that the resulting sequence constitutes a valid ranking.
For example, when we write $abxy$, this may refer to either $abcd$
or $abdc$. Throughout the verification of each profile, the symbols
$x,y,z$ are fixed.

For profiles that are not supported, we have already seen in Section \ref{subsec:n=m=4} how their assignment matrices are determined by the induction
hypothesis. For supported profiles, we may also apply the induction
hypothesis. However, note that the lexicographic order in which we
examine the profiles does not coincide with the order of induction.
Therefore, whenever we invoke the induction hypothesis, we ensure
that it applies only to profiles with a smaller disagreement parameter.
\begin{rem}
\label{unsupported profile equal param}When determining the assignment
matrix of a supported profile, we may also use unsupported agent rankings
from profiles with the same disagreement parameter, since their corresponding
column entries were determined solely based on profiles with smaller
disagreement parameters.
\end{rem}

To aid readability, we provide brief explanations after each profile
table describing how its entries are determined. Since many of these
arguments are similar or repetitive, we will give full explanations
only for the first few instances of each argument type. Whenever an
argument reappears later, we will omit the immediate explanation;
if it is not immediate but still repetitive, we will summarize it
in a remark and refer to that remark beside the relevant table.

For clarity, we will color each entry of the assignment matrix according
to the reason for which it is determined:
\begin{itemize}
\item \textcolor{red}{Red} - determined by efficiency (the agent's probability
of receiving that house is zero).
\item \textcolor{blue}{Blue} - determined by ETA.
\item \textcolor{green}{Green} - determined by SP using another already-determined
profile (by induction or otherwise).
\item \textcolor{violet}{Purple} - determined by complementing that agent's
probabilities to one.
\item \textcolor{orange}{Orange} - determined by complementing that house's
probabilities to one.
\item \textcolor{brown}{Brown} - determined by reasoning stated in a remark
and referred to beside the relevant profile table.
\end{itemize}
Entries that will be determined in subsequent steps are \uline{underlined},
and the corresponding entries to be determined in later profiles are
\textbf{bolded}.

We denote by $\left(i,h\right)$ the probability that a mechanism
satisfying the axioms assigns house $h$ to agent $i$ in a given
profile. Finally, when we say that \emph{a profile is determined},
we mean that the assignment matrix induced by any mechanism satisfying
the axioms is uniquely determined for that profile.

\subsection{Profiles with two agents sharing the same ranking}

We begin with the preference profiles in which two agents share the
same ranking. Since the axioms are invariant under renamings of both
agents and houses, we may, without loss of generality, assume that
agents $1$ and $2$ share the same ranking, $abcd$, and that the
ranking of agent $4$ does not precede that of agent $3$ in the lexicographic
order.
\begin{rem}
\label{3 ranks a first, 4 not} Whenever agent $3$ also ranks $a$
first while agent $4$ does not, agent $4$ is not supported, because
of the adjacent pair consisting of $a$ and the house immediately
above it. Hence, in such profiles the probabilities of agent $4$
are determined.
\end{rem}

\begin{rem}
\label{abdc - a,b determined}\label{acbd - a,d determined}\label{bacd - c,d determined}Whenever
agent $3$'s ranking is obtained from the ranking $abcd$ by an adjacent
swap, the probabilities associated with the two houses that are not
involved in that swap are determined. For example, when agent $3$
ranks $abdc$, the probabilities associated with houses $a$ and $b$
are determined. This is because the only pair of houses for which
the agents might not exhibit near-unanimous agreement is $\left\{ c,d\right\} $.
Thus, agents $1$, $2$, and $3$ can swap $\left\{ c,d\right\} $
in their rankings to reach a profile determined by Lemma \textcolor{blue}{\ref{the lemma}}
(hereafter, \emph{the lemma}). By SP, such a swap does not affect
their probabilities of receiving $a$ or $b$. The probabilities of
these houses for agent $4$ are then determined by house complementarity.
The same argument applies when agent $3$ ranks $acbd$ or $bacd$:
in the former case, the probabilities associated with $a$ and $d$
are determined, and in the latter, the probabilities associated with
$c$ and $d$ are determined.
\end{rem}

\begin{casenv}
\item When agent $3$ also ranks $abcd$, the profile is determined by the
lemma, since there are $n-1=3$ agents with the same ranking. In particular,
the agents exhibit near-unanimous agreement on the relative ranking
of every pair of houses, and the profile is therefore determined.
\item When agent $3$ ranks $abdc$, note that the lemma determines every
profile $\mathbf{P}$ satisfying $cP_{4}d$, because in such profiles
there is no pair of houses on which the agents fail to exhibit near-unanimous
agreement, and those profiles are therefore determined. We now go
through the remaining possible rankings for $P_{4}$.\vspace{1mm}

\begin{casenv}[leftmargin=0pt]\item {\centering{\itshape%
\begin{tabular}{|c|c|c|c|}
\hline 
\textcolor{brown}{a} & \textcolor{brown}{a} & \textcolor{brown}{a} & \textcolor{brown}{a}\tabularnewline
\hline 
\textcolor{brown}{b} & \textcolor{brown}{b} & \textcolor{brown}{b} & \textcolor{brown}{b}\tabularnewline
\hline 
c & c & d & d\tabularnewline
\hline 
\uline{d} & d & c & c\tabularnewline
\hline 
\end{tabular}}\rlap{(\ref{abdc - a,b determined})}\par}\vspace{1mm}

Since there are two pairs of agents with identical rankings, by ETE
and complementing to one, it suffices to determine $\left(1,d\right)$.
By SP, it suffices to determine $\left(1,d\right)$ in the following
profile where agent $1$ changes his ranking from $abcd$ to $cbad$.\vspace{1mm}\footnote{The ranking of the agent can be modified through a sequence of adjacent
swaps that does not involve $d$. By SP, the agent's probability of
receiving $d$ remains unchanged after each such swap.}

\begin{minipage}[t]{0.18\linewidth}{\itshape%
\begin{tabular}{|c|c|c|c|}
\hline 
c & \textcolor{blue}{a} & \textcolor{blue}{a} & \textcolor{blue}{a}\tabularnewline
\hline 
\textcolor{red}{b} & \textcolor{blue}{b} & \textcolor{blue}{b} & \textcolor{blue}{b}\tabularnewline
\hline 
\textcolor{red}{a} & c & \textcolor{violet}{d} & \textcolor{violet}{d}\tabularnewline
\hline 
\textbf{d} & \uline{d} & \textcolor{red}{c} & \textcolor{red}{c}\tabularnewline
\hline 
\end{tabular}}\end{minipage}\hfill\justifiedminipage{Here, $\left(1,b\right)$,
$\left(1,a\right)$, $\left(3,c\right)$, and $\left(4,c\right)$
are determined by efficiency,\footnotemark{} and consequently the
remaining probabilities for houses $a$ and $b$ are determined by
ETA. Next, $\left(3,d\right)$ and $\left(4,d\right)$ are determined
by the agent's complement. By house complementarity, to determine
$\left(1,d\right)$ it}\vspace{1mm} suffices to determine $\left(2,d\right)$;
by SP, this in turn reduces to determining $\left(2,d\right)$ in
the following profile.\vspace{1mm}\footnotetext{$\left(1,b\right)$
and $\left(1,a\right)$ are determined because agent $1$ is the only
agent preferring $c$ over $b$ and $a$. An assignment in which he
receives $b$ or $a$ while another agent receives $c$ would be inefficient,
since swapping houses would strictly benefit both. $\left(3,c\right)$
and $\left(4,c\right)$ are determined because agents $3$ and $4$
rank $c$ last, implying that they can receive it only if they are
last in the ordering. Since agent $1$ ranks $c$ first, no agent
following him in the ordering can receive it (if the house is available
at his turn, he will take it).}

\begin{minipage}[t]{0.17\linewidth}{\itshape%
\begin{tabular}{|c|c|c|c|}
\hline 
c & c & \textcolor{blue}{a} & \textcolor{blue}{a}\tabularnewline
\hline 
b & b & \uline{b} & b\tabularnewline
\hline 
\textcolor{red}{a} & \textcolor{red}{a} & d & d\tabularnewline
\hline 
d & \makebox[0pt][c]{\textbf{d}} & \textcolor{red}{c} & \textcolor{red}{c}\tabularnewline
\hline 
\end{tabular}}\end{minipage}\hfill\minipar{$\left(1,a\right)$, $\left(2,a\right)$,
$\left(3,c\right)$, and $\left(4,c\right)$ are determined by efficiency,\footnotemark{}
and consequently the remaining probabilities for $a$ are determined
by ETA. Since there are two pairs of agents with identical rankings,
it suffices to determine $\left(3,b\right)$, and it suffices to do
so in the following profile.}\vspace{1mm}\footnotetext{Each agent
$i\in\left\{ 1,2\right\} $ ranks $a$ second to last, implying that
he receives it only when he is last or second to last, and in particular
either agent $3$ or agent $4$ comes before agent $i$, while both
agents $3$ and $4$ rank $a$ first, which implies that any agent
that comes after one of them in the ordering would not receive $a$.}

\begin{minipage}[t]{0.17\linewidth}{\itshape%
\begin{tabular}{|c|c|c|c|}
\hline 
c & c & \textcolor{blue}{a} & \textcolor{blue}{a}\tabularnewline
\hline 
b & b & \makebox[0pt][c]{\textbf{\textcolor{violet}{b}}} & b\tabularnewline
\hline 
\textcolor{red}{a} & \textcolor{red}{a} & \textcolor{red}{c} & d\tabularnewline
\hline 
d & d & \textcolor{green}{d} & c\tabularnewline
\hline 
\end{tabular}}\end{minipage}\hfill\justifiedminipage{$\left(1,a\right)$, $\left(2,a\right)$,
and $\left(3,c\right)$ are determined by efficiency, and consequently
the remaining probabilities for $a$ are determined by ETA. Note that
$\left(3,d\right)$ is determined by SP and by considering the profile
where agent $3$ changes his ranking to $cbad$ (i.e., matches the
rankings of agents $1$ and}\vspace{1mm} $2$). That profile is
determined by the lemma because there are three agents with the same
ranking. Then, $\left(3,b\right)$ is determined by agent's complement,
as required.\vspace{1mm}

\item {\centering{\itshape%
\begin{tabular}{|c|c|c|c|}
\hline 
\textcolor{brown}{a} & \textcolor{brown}{a} & \textcolor{brown}{a} & \textcolor{green}{a}\tabularnewline
\hline 
\textcolor{brown}{b} & \textcolor{brown}{b} & \textcolor{brown}{b} & \textcolor{green}{d}\tabularnewline
\hline 
c & c & d & \textcolor{green}{x}\tabularnewline
\hline 
\uline{d} & d & c & \textcolor{green}{y}\tabularnewline
\hline 
\end{tabular}}\rlap{(\ref{abdc - a,b determined})}\par}\vspace{1mm}First, note
that agent $4$ is not supported because we can see that all other
agents have the opposite preference with respect to the adjacent pair
containing $b$ and the house immediately above it in agent $4$'s
ranking (either $c$ or $d$). Since agents $1$ and $2$ share the
same ranking, it suffices to determine $\left(1,d\right)$ in the
following profile.\vspace{1mm}

\begin{minipage}[t]{0.17\linewidth}{\itshape%
\begin{tabular}{|c|c|c|c|}
\hline 
c & \textcolor{blue}{a} & \textcolor{blue}{a} & \textcolor{blue}{a}\tabularnewline
\hline 
\textcolor{red}{b} & \textcolor{blue}{b} & \textcolor{blue}{b} & \textcolor{violet}{d}\tabularnewline
\hline 
\textcolor{red}{a} & c & \textcolor{violet}{d} & \textcolor{red}{x}\tabularnewline
\hline 
\makebox[0pt][c]{\textbf{d}} & \uline{d} & \textcolor{red}{c} & \textcolor{red}{y}\tabularnewline
\hline 
\end{tabular}}\end{minipage}\hfill\minipar{$\left(1,b\right)$, $\left(1,a\right)$,
$\left(3,c\right)$, $\left(4,b\right)$, and $\left(4,c\right)$
are determined by efficiency,\footnotemark{} and consequently ETA
determines $a$ and $b$ (the remaining probabilities for those houses),
and $\left(3,d\right)$ and $\left(4,d\right)$ are determined by
agent's complement. It then suffices to determine $\left(2,d\right)$
in the following profile.}{\vspace{1mm}}\footnotetext{Note that
agent $4$ cannot receive $b$ or $c$ because for that to happen,
some other agent should receive $d$ beforehand, and the only agent
other than $4$ which can receive $d$ when he is not the last in
the ordering is agent $3$ (because he does not rank $d$ last). For
that to happen, agents $1$ and $2$ should come before agent $3$
and take the houses that he prefers over $d$, that is, $a$ and $b$,
but because of their preferences, they will receive $a$ and $c$
instead. Furthermore, agent $1$ cannot receive $b$ because for that
to happen, another agent must receive $c$ beforehand, but the only
other agent who might receive $c$ is agent $2$, and he prefers $b$,
so he would not take $c$ if $b$ were available when it is his turn.}

\begin{minipage}[t]{0.17\linewidth}{\itshape%
\begin{tabular}{|c|c|c|c|}
\hline 
c & c & \textcolor{blue}{a} & \textcolor{blue}{a}\tabularnewline
\hline 
b & b & \uline{b} & \textcolor{violet}{d}\tabularnewline
\hline 
\textcolor{red}{a} & \textcolor{red}{a} & d & \textcolor{red}{x}\tabularnewline
\hline 
d & \makebox[0pt][c]{\textbf{d}} & \textcolor{red}{c} & \textcolor{red}{y}\tabularnewline
\hline 
\end{tabular}}\end{minipage}\hfill\minipar{$\left(1,a\right)$, $\left(2,a\right)$,
$\left(3,c\right)$, $\left(4,b\right)$, and $\left(4,c\right)$
are determined by efficiency, and consequently $a$ is determined
by ETA, and $\left(4,d\right)$ is determined by agent's complement.
Since agents $1$ and $2$ share the same ranking, it suffices to
determine $\left(3,b\right)$ in the following profile.}\vspace{1mm}

\begin{minipage}[t]{0.17\linewidth}{\itshape%
\begin{tabular}{|c|c|c|c|}
\hline 
c & c & \textcolor{blue}{a} & \textcolor{blue}{a}\tabularnewline
\hline 
b & b & \makebox[0pt][c]{\textbf{\textcolor{violet}{b}}\textcolor{violet}{}} & d\tabularnewline
\hline 
\textcolor{red}{a} & \textcolor{red}{a} & \textcolor{red}{c} & x\tabularnewline
\hline 
d & d & \textcolor{green}{d} & y\tabularnewline
\hline 
\end{tabular}}\end{minipage}\hfill\minipar{$\left(1,a\right)$, $\left(2,a\right)$,
and $\left(3,c\right)$ are determined by efficiency, and consequently
$a$ is determined by ETA. $\left(3,d\right)$ is determined by the
lemma (we can switch to the profile where agent $3$ matches the ranking
of agents $1$ and $2$), and then $\left(3,b\right)$ is determined
by agent's complement, as required.}\vspace{1mm}

\item {\centering{\itshape%
\begin{tabular}{|c|c|c|c|}
\hline 
\textcolor{brown}{a} & \textcolor{brown}{a} & \textcolor{brown}{a} & \textcolor{brown}{b}\tabularnewline
\hline 
\textcolor{brown}{b} & \textcolor{brown}{b} & \textcolor{brown}{b} & \textcolor{brown}{x}\tabularnewline
\hline 
c & c & d & \textcolor{brown}{y}\tabularnewline
\hline 
d & d & \uline{c} & \textcolor{brown}{z}\tabularnewline
\hline 
\end{tabular}}\rlap{(\ref{3 ranks a first, 4 not}, \ref{abdc - a,b determined})}\par}\vspace{1mm}

It suffices to determine $\left(3,c\right)$ in the following profile.\vspace{1mm}

\begin{minipage}[t]{0.17\linewidth}{\itshape%
\begin{tabular}{|c|c|c|c|}
\hline 
\textcolor{blue}{a} & \textcolor{blue}{a} & d & \textcolor{green}{b}\tabularnewline
\hline 
\textcolor{orange}{b} & \textcolor{orange}{b} & \textcolor{red}{a} & x\tabularnewline
\hline 
\textcolor{violet}{c} & \textcolor{violet}{c} & \textcolor{red}{b} & y\tabularnewline
\hline 
\textcolor{red}{d} & \textcolor{red}{d} & \textbf{c} & z\tabularnewline
\hline 
\end{tabular}}\end{minipage}\hfill\justifiedminipage{$\left(1,d\right)$, $\left(2,d\right)$,
$\left(3,a\right)$, $\left(3,b\right)$, and $\left(4,a\right)$
are determined by efficiency, and consequently $a$ is determined
by ETA. $\left(4,b\right)$ is determined by considering the profile
where agent $4$ ranks $bacd$ and applying the lemma. Then, $\left(1,b\right)$
and $\left(2,b\right)$ are determined by house}\vspace{1mm} complementarity, and
then $\left(1,c\right)$ and $\left(2,c\right)$ are
determined by agent complementarity. It then suffices to determine
$\left(4,c\right)$ in the following profile.\footnote{If $z=c$, the conclusion follows directly from SP. Otherwise, since
we assume $dP_{4}c$, we must have $x=d$. In that case, $\left(4,a\right)$
is determined in both profiles with $x=d$, and therefore, by Remark
\texttt{\textcolor{blue}{\ref{efficiency swaps}}}, it suffices to
determine $\left(4,c\right)$ at the profile where $aP_{4}c$.}\vspace{1mm}

\begin{minipage}[t]{0.17\linewidth}{\itshape%
\begin{tabular}{|c|c|c|c|}
\hline 
\textcolor{blue}{a} & \textcolor{blue}{a} & d & \textcolor{blue}{a}\tabularnewline
\hline 
\textcolor{blue}{b} & \textcolor{blue}{b} & \textcolor{red}{a} & \textcolor{blue}{b}\tabularnewline
\hline 
\textcolor{violet}{c} & \textcolor{violet}{c} & \textcolor{red}{b} & d\tabularnewline
\hline 
\textcolor{red}{d} & \textcolor{red}{d} & \textcolor{green}{c} & \textbf{\textcolor{orange}{c}}\tabularnewline
\hline 
\end{tabular}}\end{minipage}\hfill\justifiedminipage{$\left(1,d\right)$, $\left(2,d\right)$,
$\left(3,a\right)$, and $\left(3,b\right)$ are determined by efficiency,
and consequently $a$ and $b$ are determined by ETA. Then, $\left(1,c\right)$
and $\left(2,c\right)$ are determined by agent's complement. To determine
$\left(3,c\right)$, we consider the profile obtained by changing
agent $3$'s ranking to $abdc$.}\vspace{1mm} Note that this new
profile has a disagreement parameter smaller than that of the profiles
we wish to solve, and is therefore determined under the induction
hypothesis. $\left(4,c\right)$ is then determined by house complementarity,
as required.\vspace{1mm}

\item {\centering{\itshape%
\begin{tabular}{|c|c|c|c|}
\hline 
a & a & a & \textcolor{brown}{d}\tabularnewline
\hline 
b & b & b & \textcolor{brown}{x}\tabularnewline
\hline 
c & c & d & \textcolor{brown}{y}\tabularnewline
\hline 
d & d & c & \textcolor{brown}{z}\tabularnewline
\hline 
\end{tabular}}\rlap{(\textcolor{blue}{\ref{3 ranks a first, 4 not}})}\par}\vspace{1mm}

In such profiles, agents $1$ and $2$ are also not supported because
of the pair $\left\{ c,d\right\} $, and consequently those profiles
are not supported.\end{casenv}
\item When agent $3$ ranks $acbd$, note that all profiles where $bP_{4}c$
are determined by the lemma, and therefore we only need to consider
the remaining possible rankings for $P_{4}$.\vspace{1mm}

\begin{casenv}[leftmargin=0pt]\item {\centering{\itshape%
\begin{tabular}{|c|c|c|c|}
\hline 
\textcolor{brown}{a} & \textcolor{brown}{a} & \textcolor{brown}{a} & \textcolor{brown}{a}\tabularnewline
\hline 
\uline{b} & b & c & c\tabularnewline
\hline 
c & c & b & b\tabularnewline
\hline 
\textcolor{brown}{d} & \textcolor{brown}{d} & \textcolor{brown}{d} & \textcolor{brown}{d}\tabularnewline
\hline 
\end{tabular}}\rlap{(\textcolor{blue}{\ref{acbd - a,d determined}})}\par}\vspace{1mm}

It suffices to determine $\left(1,b\right)$ in the following profile.\vspace{1mm}

\begin{minipage}[t]{0.17\linewidth}{\itshape%
\begin{tabular}{|c|c|c|c|}
\hline 
\textcolor{blue}{a} & \textcolor{blue}{a} & \textcolor{blue}{a} & \textcolor{blue}{a}\tabularnewline
\hline 
\textbf{\textcolor{violet}{b}} & b & c & c\tabularnewline
\hline 
\textcolor{orange}{d} & c & b & b\tabularnewline
\hline 
\textcolor{red}{c} & \textcolor{green}{d} & \textcolor{green}{d} & \textcolor{green}{d}\tabularnewline
\hline 
\end{tabular}}\end{minipage}\hfill\justifiedminipage{$\left(1,c\right)$ is determined
by efficiency, and $a$ is determined by ETA. Note that since $\left\{ b,c\right\} $
is the only pair without near-unanimous agreement among the agents,
$\left(i,d\right)$ is determined for $i=2,3,4$, because each of
those agents can swap $\left\{ b,c\right\} $ and obtain a profile
that is determined by}\vspace{1mm} the lemma. Then, $\left(1,b\right)$
is determined by complementarity, as required.\vspace{1mm}

\item {\centering{\itshape%
\begin{tabular}{|c|c|c|c|}
\hline 
\textcolor{brown}{a} & \textcolor{brown}{a} & \textcolor{brown}{a} & \textcolor{green}{a}\tabularnewline
\hline 
\uline{b} & b & c & \textcolor{green}{c}\tabularnewline
\hline 
c & c & b & \textcolor{green}{d}\tabularnewline
\hline 
\textcolor{brown}{d} & \textcolor{brown}{d} & \textcolor{brown}{d} & \textcolor{green}{b}\tabularnewline
\hline 
\end{tabular}}\rlap{(\textcolor{blue}{\ref{acbd - a,d determined}})}\par}\vspace{1mm}

Agent $4$ is not supported because of the pair $\left\{ d,b\right\} $.
It suffices to determine $\left(1,b\right)$ in the following profile.\vspace{1mm}

\begin{minipage}[t]{0.17\linewidth}{\itshape%
\begin{tabular}{|c|c|c|c|}
\hline 
\textcolor{blue}{a} & \textcolor{blue}{a} & \textcolor{blue}{a} & \textcolor{blue}{a}\tabularnewline
\hline 
\makebox[0pt][c]{\textbf{b}} & b & c & c\tabularnewline
\hline 
\uline{d} & c & b & d\tabularnewline
\hline 
\textcolor{red}{c} & d & d & b\tabularnewline
\hline 
\end{tabular}}\end{minipage}\hfill\minipar{$\left(1,c\right)$ is determined by
efficiency, and $a$ is determined by ETA. It suffices to determine
$\left(1,d\right)$ in the following profile.}\vspace{1mm}

\begin{minipage}[t]{0.17\linewidth}{\itshape%
\begin{tabular}{|c|c|c|c|}
\hline 
b & a & a & a\tabularnewline
\hline 
a & b & c & c\tabularnewline
\hline 
\makebox[0pt][c]{\textbf{d}} & c & b & \uline{d}\tabularnewline
\hline 
c & \textcolor{green}{d} & \textcolor{green}{d} & b\tabularnewline
\hline 
\end{tabular}}\end{minipage}\hfill\minipar{$\left(i,d\right)$ for $i=2,3$ is
determined because each of these agents can swap $\left\{ b,c\right\} $
and obtain a profile that is determined by the lemma. It then suffices
to determine $\left(4,d\right)$.}\vspace{1mm}

\begin{minipage}[t]{0.17\linewidth}{\itshape%
\begin{tabular}{|c|c|c|c|}
\hline 
b & \textcolor{blue}{a} & \textcolor{blue}{a} & c\tabularnewline
\hline 
\textcolor{red}{a} & \uline{b} & \uline{c} & \textcolor{red}{a}\tabularnewline
\hline 
d & c & b & \makebox[0pt][c]{\textbf{d}}\tabularnewline
\hline 
\textcolor{red}{c} & \textcolor{green}{d} & d & \textcolor{red}{b}\tabularnewline
\hline 
\end{tabular}}\end{minipage}\hfill\minipar{$\left(1,a\right)$, $\left(1,c\right)$,
$\left(4,a\right)$, and $\left(4,b\right)$ are determined by efficiency,
and consequently $a$ is determined by ETA. $\left(2,d\right)$ is
determined since agent $2$ can swap $\left\{ b,c\right\} $ to obtain
a profile determined by the lemma. It then suffices to determine $\left(2,b\right)$
and $\left(3,c\right)$.}\vspace{1mm}
\begin{itemize}
\item \begin{minipage}[t]{0.18\linewidth}{\itshape%
\begin{tabular}{|c|c|c|c|}
\hline 
b & \textcolor{blue}{a} & \textcolor{blue}{a} & c\tabularnewline
\hline 
\textcolor{red}{a} & \textbf{b} & c & \textcolor{red}{a}\tabularnewline
\hline 
d & \uline{d} & b & d\tabularnewline
\hline 
c & \textcolor{red}{c} & d & b\tabularnewline
\hline 
\end{tabular}}\end{minipage}\hfill\modifiedminipar{0.8}{We can use this profile
to determine $\left(2,b\right)$. $\left(1,a\right)$, $\left(2,c\right)$,
and $\left(4,a\right)$ are determined by efficiency, and consequently
$a$ is determined by ETA. It suffices to determine $\left(2,d\right)$
in the following profile.}\vspace{1mm}
\begin{itemize}
\item \begin{minipage}[t]{0.19\linewidth}{\itshape%
\begin{tabular}{|c|c|c|c|}
\hline 
\textcolor{blue}{b} & \textcolor{blue}{b} & \uline{a} & c\tabularnewline
\hline 
a & a & c & \textcolor{red}{a}\tabularnewline
\hline 
d & \makebox[0pt][c]{\textbf{d}} & \textcolor{red}{b} & d\tabularnewline
\hline 
c & \textcolor{red}{c} & d & \textcolor{red}{b}\tabularnewline
\hline 
\end{tabular}}\end{minipage}\hfill\modifiedminipar{0.79}{$\left(2,c\right)$,
$\left(3,b\right)$, $\left(4,a\right)$, and $\left(4,b\right)$
are determined by efficiency, and consequently $b$ is determined
by ETA. It then suffices to determine $\left(3,a\right)$ in the following
profile (note that agents $1$ and $2$ share the same ranking).}\vspace{1mm}
\item \begin{minipage}[t]{0.19\linewidth}{\itshape%
\begin{tabular}{|c|c|c|c|}
\hline 
\textcolor{blue}{b} & \textcolor{blue}{b} & \makebox[0pt][c]{\textbf{\textcolor{violet}{a}}\textcolor{violet}{}} & \textcolor{blue}{c}\tabularnewline
\hline 
\textcolor{violet}{a} & \textcolor{violet}{a} & \textcolor{red}{b} & \textcolor{red}{a}\tabularnewline
\hline 
\textcolor{blue}{d} & \textcolor{blue}{d} & \textcolor{blue}{d} & \textcolor{red}{d}\tabularnewline
\hline 
\textcolor{red}{c} & \textcolor{red}{c} & \textcolor{red}{c} & \textcolor{red}{b}\tabularnewline
\hline 
\end{tabular}}\end{minipage}\hfill\modifiedminipar{0.79}{Efficiency and ETA determine
$b$, $c$, and $d$, and consequently the profile is determined.}\vspace{1mm}
\end{itemize}
\item \begin{minipage}[t]{0.18\linewidth}{\itshape%
\begin{tabular}{|c|c|c|c|}
\hline 
b & \textcolor{blue}{a} & \textcolor{blue}{a} & c\tabularnewline
\hline 
\textcolor{red}{a} & b & \textbf{c} & \textcolor{red}{a}\tabularnewline
\hline 
d & c & \uline{d} & d\tabularnewline
\hline 
c & d & \textcolor{red}{b} & b\tabularnewline
\hline 
\end{tabular}}\end{minipage}\hfill\modifiedminipar{0.8}{We can use this profile
to determine $\left(3,c\right)$. $\left(1,a\right)$, $\left(3,b\right)$,
and $\left(4,a\right)$ are determined by efficiency, and consequently
$a$ is determined by ETA. It then suffices to determine $\left(3,d\right)$
in the following profile.}\vspace{1mm}
\begin{itemize}
\item \begin{minipage}[t]{0.19\linewidth}{\itshape%
\begin{tabular}{|c|c|c|c|}
\hline 
b & \uline{a} & \textcolor{blue}{c} & \textcolor{blue}{c}\tabularnewline
\hline 
\textcolor{red}{a} & b & a & a\tabularnewline
\hline 
d & \textcolor{red}{c} & \makebox[0pt][c]{\textbf{d}} & d\tabularnewline
\hline 
\textcolor{red}{c} & d & \textcolor{red}{b} & b\tabularnewline
\hline 
\end{tabular}}\end{minipage}\hfill\modifiedminipar{0.79}{$\left(1,a\right)$,
$\left(1,c\right)$, $\left(2,c\right)$, and $\left(3,b\right)$
are determined by efficiency, and consequently $c$ is determined
by ETA. It then suffices to determine $\left(2,a\right)$ in the following
profile (note that agents $3$ and $4$ share the same ranking).}\vspace{1mm}
\item \begin{minipage}[t]{0.19\linewidth}{\itshape%
\begin{tabular}{|c|c|c|c|}
\hline 
\textcolor{blue}{b} & \makebox[0pt][c]{\textbf{\textcolor{violet}{a}}\textcolor{violet}{}} & \textcolor{blue}{c} & \textcolor{blue}{c}\tabularnewline
\hline 
\textcolor{red}{a} & \textcolor{red}{c} & \textcolor{violet}{a} & \textcolor{violet}{a}\tabularnewline
\hline 
\textcolor{red}{d} & \textcolor{blue}{d} & \textcolor{blue}{d} & \textcolor{blue}{d}\tabularnewline
\hline 
\textcolor{red}{c} & \textcolor{red}{b} & \textcolor{red}{b} & \textcolor{red}{b}\tabularnewline
\hline 
\end{tabular}}\end{minipage}\hfill\modifiedminipar{0.79}{Efficiency and ETA determine
$b$, $c$, and $d$, and consequently the profile is determined.}\vspace{1mm}
\end{itemize}
\end{itemize}
\item {\centering{\itshape%
\begin{tabular}{|c|c|c|c|}
\hline 
a & a & a & a\tabularnewline
\hline 
b & b & c & d\tabularnewline
\hline 
c & c & b & c\tabularnewline
\hline 
d & d & d & b\tabularnewline
\hline 
\end{tabular}}\par}\vspace{1mm}

This profile is not supported because the pair $\left\{ b,c\right\} $
is not under near-unanimous agreement among the agents, and both agents
$3$ and $4$ cannot receive $b$.\vspace{1mm}

\item {\centering{\itshape\hypertarget{case 3 subcase 4}{} %
\begin{tabular}{|c|c|c|c|}
\hline 
\textcolor{brown}{a} & \textcolor{brown}{a} & \textcolor{brown}{a} & \textcolor{brown}{c}\tabularnewline
\hline 
\uline{b} & b & c & \textcolor{brown}{a}\tabularnewline
\hline 
c & c & b & \textcolor{brown}{b}\tabularnewline
\hline 
\textcolor{brown}{d} & \textcolor{brown}{d} & \textcolor{brown}{d} & \textcolor{brown}{d}\tabularnewline
\hline 
\end{tabular}}\rlap{(\textcolor{blue}{\ref{3 ranks a first, 4 not}},\textcolor{blue}{{}
\ref{acbd - a,d determined}})}\par}\vspace{1mm}

It suffices to determine $\left(1,b\right)$ in the following profile.\vspace{1mm}

\begin{minipage}[t]{0.17\linewidth}{\itshape%
\begin{tabular}{|c|c|c|c|}
\hline 
\textcolor{blue}{a} & \textcolor{blue}{a} & \textcolor{blue}{a} & c\tabularnewline
\hline 
\textbf{\textcolor{violet}{b}} & b & c & \textcolor{red}{a}\tabularnewline
\hline 
\textcolor{orange}{d} & c & b & b\tabularnewline
\hline 
\textcolor{red}{c} & \textcolor{green}{d} & \textcolor{green}{d} & \textcolor{green}{d}\tabularnewline
\hline 
\end{tabular}}\end{minipage}\hfill\justifiedminipage{$\left(1,c\right)$ and $\left(4,a\right)$
are determined by efficiency, and consequently $a$ is determined
by ETA. $\left(i,d\right)$ for $i=2,3,4$ is determined by the lemma:
agents $2$ and $3$ can swap $\left\{ b,c\right\} $, and agent $4$
can change his ranking to $abcd$; in both cases, the resulting profile
is determined by the lemma.}\vspace{1mm} $\left(1,b\right)$ is
then determined by complementarity, as required.\vspace{1mm}

\item {\centering{\itshape%
\begin{tabular}{|c|c|c|c|}
\hline 
\textcolor{brown}{a} & \textcolor{brown}{a} & \textcolor{brown}{a} & \textcolor{brown}{c}\tabularnewline
\hline 
\uline{b} & b & c & \textcolor{brown}{a}\tabularnewline
\hline 
c & c & b & \textcolor{brown}{d}\tabularnewline
\hline 
\textcolor{brown}{d} & \textcolor{brown}{d} & \textcolor{brown}{d} & \textcolor{brown}{b}\tabularnewline
\hline 
\end{tabular}}\rlap{(\textcolor{blue}{\ref{3 ranks a first, 4 not}},\textcolor{blue}{{}
\ref{acbd - a,d determined}})}\par}\vspace{1mm}

It suffices to determine $\left(1,b\right)$ in the following profile.\vspace{1mm}

\begin{minipage}[t]{0.17\linewidth}{\itshape%
\begin{tabular}{|c|c|c|c|}
\hline 
\textcolor{blue}{a} & \textcolor{blue}{a} & \textcolor{blue}{a} & \textcolor{green}{c}\tabularnewline
\hline 
\textbf{\textcolor{violet}{b}} & b & c & \textcolor{red}{a}\tabularnewline
\hline 
\textcolor{orange}{d} & c & b & \textcolor{violet}{d}\tabularnewline
\hline 
\textcolor{red}{c} & \textcolor{green}{d} & \textcolor{green}{d} & \textcolor{red}{b}\tabularnewline
\hline 
\end{tabular}}\end{minipage}\hfill\justifiedminipage{$\left(1,c\right)$, $\left(4,a\right)$,
and $\left(4,b\right)$ are determined by efficiency, and consequently
$a$ is determined by ETA. $\left(i,d\right)$ for $i=2,3$ is determined
by the lemma: each of these agents can swap $\left\{ b,c\right\} $,
yielding a profile covered by the lemma. $\left(4,c\right)$ is determined
by considering the profile where agent $4$}\vspace{1mm} swaps $\left\{ d,b\right\} $;
the disagreement parameter of that profile equals that of the original
one,\footnotemark{} and in that profile agent $4$ is unsupported
because of the pair $\left\{ c,a\right\} $, so the desired result
follows by Remark \textcolor{blue}{\ref{unsupported profile equal param}}.
$\left(1,b\right)$ is then determined by complementarity, as required.\vspace{1mm}\footnotetext{This
holds because the addition from swapping the pair $\left\{ c,d\right\} $
in agent $1$'s ranking in the first transition is cancelled by the
swap of the pair $\left\{ d,b\right\} $ in agent $4$'s ranking in
the second transition.}

\item {\centering{\itshape%
\begin{tabular}{|c|c|c|c|}
\hline 
a & a & a & c\tabularnewline
\hline 
b & b & c & b\tabularnewline
\hline 
c & c & b & x\tabularnewline
\hline 
d & d & d & y\tabularnewline
\hline 
\end{tabular}}\par}\vspace{1mm}

We use similar arguments to those applied in \hyperlink{case 3 subcase 4}{{\itshape Case} iv}.
When agent $4$ does not rank $d$ last, we make a slight adjustment:
in the second step, to determine $\left(4,d\right)$, we may use the
lemma on the profile obtained when agent $4$ swaps $\left\{ b,c\right\} $
(that is, the profile where he ranks $bcda$).\vspace{1mm}

\item {\centering{\itshape%
\begin{tabular}{|c|c|c|c|}
\hline 
\textcolor{brown}{a} & \textcolor{brown}{a} & \textcolor{brown}{a} & \textcolor{brown}{c}\tabularnewline
\hline 
\uline{b} & b & c & \textcolor{brown}{d}\tabularnewline
\hline 
c & c & b & \textcolor{brown}{x}\tabularnewline
\hline 
\textcolor{brown}{d} & \textcolor{brown}{d} & \textcolor{brown}{d} & \textcolor{brown}{y}\tabularnewline
\hline 
\end{tabular}}\rlap{(\textcolor{blue}{\ref{3 ranks a first, 4 not}},\textcolor{blue}{{}
\ref{acbd - a,d determined}})}\par}\vspace{1mm}

It suffices to determine $\left(1,b\right)$ in the following profile.\vspace{1mm}

\begin{minipage}[t]{0.17\linewidth}{\itshape%
\begin{tabular}{|c|c|c|c|}
\hline 
\textcolor{blue}{a} & \textcolor{blue}{a} & \textcolor{blue}{a} & \textcolor{green}{c}\tabularnewline
\hline 
\textbf{\textcolor{violet}{b}} & b & c & \textcolor{violet}{d}\tabularnewline
\hline 
\textcolor{orange}{d} & c & b & \textcolor{red}{x}\tabularnewline
\hline 
\textcolor{red}{c} & \textcolor{green}{d} & \textcolor{green}{d} & \textcolor{red}{y}\tabularnewline
\hline 
\end{tabular}}\end{minipage}\hfill\justifiedminipage{$\left(1,c\right)$, $\left(4,a\right)$,
and $\left(4,b\right)$ are determined by efficiency, and consequently
$a$ is determined by ETA. $\left(i,d\right)$ for $i=2,3$ is determined
by the lemma: for each of these agents, swapping $\left\{ b,c\right\} $
yields a profile covered by the lemma. $\left(4,c\right)$ is determined
by considering the profile where agent $4$}\vspace{1mm} swaps $\left\{ d,x\right\} $;
the disagreement parameter of that profile equals that of the original
one,\footnotemark{} and in that profile agent $4$ is unsupported
because of the adjacent pair containing $a$ and the house ranked
immediately above it. $\left(1,b\right)$ is then determined by complementarity,
as required.\vspace{1mm}\footnotetext{This holds because the addition
from swapping the pair $\left\{ c,d\right\} $ in agent $1$'s ranking
in the first transition is cancelled by the swap of the pair $\left\{ d,x\right\} $
in agent $4$'s ranking in the second transition.}

\item {\centering{\itshape%
\begin{tabular}{|c|c|c|c|}
\hline 
\textcolor{brown}{a} & \textcolor{brown}{a} & \textcolor{brown}{a} & \textcolor{brown}{d}\tabularnewline
\hline 
b & b & c & \textcolor{brown}{x}\tabularnewline
\hline 
c & c & \textcolor{red}{b} & \textcolor{brown}{y}\tabularnewline
\hline 
\textcolor{brown}{d} & \textcolor{brown}{d} & \textcolor{brown}{d} & \textcolor{brown}{z}\tabularnewline
\hline 
\end{tabular}}\rlap{(\textcolor{blue}{\ref{3 ranks a first, 4 not}},\textcolor{blue}{{}
\ref{acbd - a,d determined}})}\par}\vspace{1mm}

$\left(3,b\right)$ is determined by efficiency, and consequently
the profile is determined.\end{casenv}
\item When agent $3$ ranks $acdb$, all such profiles are not supported
since both agents $3$ and $4$ are not supported. Agent $3$ is not
supported because of the pair $\left\{ d,b\right\} $. For agent $4$,
Remark \textcolor{blue}{\ref{3 ranks a first, 4 not}} shows that
he is not supported whenever he does not rank $a$ first, and when
he does, he is still not supported because of the pair containing
$b$ and the house immediately above it (note that we only consider
rankings for agent $4$ that do not precede $acdb$ in lexicographic
order).
\item When agent $3$ ranks $adbc$, all such profiles are not supported
as well. If agent $4$ ranks $adxy$, then agents $1$ and $2$ are
not supported because of the pair $\left\{ c,d\right\} $, which suffices
for the case where he ranks $adbc$, since in that case agents $3$
and $4$ share the same ranking.\footnote{Note also that this profile can be obtained from the one where agents
$3$ and $4$ rank $acdb$ after suitable renamings of the agents
and houses.} When agent $4$ ranks $adcb$, he is not supported because of the
pair $\left\{ c,b\right\} $. In the remaining cases, agent $4$ is
not supported by Remark \textcolor{blue}{\ref{3 ranks a first, 4 not}}.
When he ranks $d$ first, agents $1$ and $2$ are not supported because
of the pair $\left\{ c,d\right\} $; when he ranks either $b$ or
$c$ first, agent $3$ is not supported because of the pair $\left\{ d,b\right\} $.\footnote{Note that in this case, agent $3$ cannot receive $b$, since another
agent must receive $d$ beforehand, and because agents $1$ and $2$
rank $d$ last, that agent must be agent $4$. However, he would not
receive $d$ unless agents $1$ and $2$ received his first preference
earlier. If that preference is $b$, agent $3$ would not receive
it afterward; if it is $c$, then agents $1$ and $2$ cannot receive
it when they are among the first two agents in the ordering, as it
is only their third choice.}
\item When agent $3$ ranks $adcb$, all such profiles are not supported
as well. Agent $3$ is not supported because of the pair $\left\{ c,b\right\} $.
Agent $4$ is not supported either: if he shares the same ranking
as agent $3$, then it is because of the pair $\left\{ c,b\right\} $;
otherwise, he does not rank $a$ first, and hence is not supported
by Remark \textcolor{blue}{\ref{3 ranks a first, 4 not}}.
\end{casenv}
We now move to the cases where agent $3$ does not rank $a$ first,
most of which are resolved by the following remark.
\begin{rem}
It suffices to check the cases where agent $3$ ranks $a$ second
and where agents $3$ and $4$ have the same top choice. Indeed, whenever
the top choices of agents $3$ and $4$ differ, neither of them would
receive $a$, and both would be unsupported because of the pair containing
$a$ and the house ranked immediately above it in their orderings.
Moreover, even when their top choice is the same, they could receive
$a$ only when they rank it second. Since we assume that the ranking
of agent $4$ does not precede that of agent $3$ in the lexicographic
order, it is impossible that agent $4$ ranks $a$ second while agent
$3$ does not.
\end{rem}

We now verify the cases where agent $3$ ranks $a$ second and agent
$4$'s top choice is the same as agent $3$'s.

\begin{casenv}[resume]\item When agent $3$ ranks $bacd$.

\begin{casenv}[leftmargin=0pt]\item {\centering{\itshape%
\begin{tabular}{|c|c|c|c|}
\hline 
\uline{a} & a & b & b\tabularnewline
\hline 
b & b & a & a\tabularnewline
\hline 
\textcolor{brown}{c} & \textcolor{brown}{c} & \textcolor{brown}{c} & \textcolor{brown}{c}\tabularnewline
\hline 
\textcolor{brown}{d} & \textcolor{brown}{d} & \textcolor{brown}{d} & \textcolor{brown}{d}\tabularnewline
\hline 
\end{tabular}}\rlap{(\textcolor{blue}{\ref{bacd - c,d determined}})}\par}\vspace{1mm}

It suffices to determine $\left(1,a\right)$ in the following profile.\vspace{1mm}

\begin{minipage}[t]{0.17\linewidth}{\itshape%
\begin{tabular}{|c|c|c|c|}
\hline 
\textbf{\textcolor{violet}{a}} & a & b & b\tabularnewline
\hline 
\textcolor{orange}{c} & b & a & a\tabularnewline
\hline 
\textcolor{red}{b} & \textcolor{green}{c} & \textcolor{green}{c} & \textcolor{green}{c}\tabularnewline
\hline 
\textcolor{blue}{d} & \textcolor{blue}{d} & \textcolor{blue}{d} & \textcolor{blue}{d}\tabularnewline
\hline 
\end{tabular}}\end{minipage}\hfill\minipar{$\left(1,b\right)$ is determined by
efficiency, $d$ is determined by ETA, and $\left(i,c\right)$ for
$i=2,3,4$ is determined by considering the profile where agent $i$
swaps $\left\{ a,b\right\} $, which is determined by the lemma. $\left(1,a\right)$
is then determined by complementarity, as required.}\vspace{1mm}

\item {\centering{\itshape\hypertarget{bacd, badc}{} %
\begin{tabular}{|c|c|c|c|}
\hline 
\uline{a} & a & b & \textcolor{green}{b}\tabularnewline
\hline 
b & b & a & \textcolor{green}{a}\tabularnewline
\hline 
\textcolor{brown}{c} & \textcolor{brown}{c} & \textcolor{brown}{c} & \textcolor{brown}{d}\tabularnewline
\hline 
\textcolor{brown}{d} & \textcolor{brown}{d} & \textcolor{brown}{d} & \textcolor{brown}{c}\tabularnewline
\hline 
\end{tabular}}\rlap{(\textcolor{blue}{\ref{bacd - c,d determined}})}\par}\vspace{1mm}

Agent $4$ is not supported because of the pair $\left\{ d,c\right\} $.
It then suffices to determine $\left(1,a\right)$ in the following
profile.\vspace{1mm}

\begin{minipage}[t]{0.17\linewidth}{\itshape%
\begin{tabular}{|c|c|c|c|}
\hline 
\textbf{\textcolor{violet}{a}} & a & b & b\tabularnewline
\hline 
\textcolor{orange}{c} & b & a & a\tabularnewline
\hline 
\textcolor{red}{b} & \textcolor{green}{c} & \textcolor{green}{c} & d\tabularnewline
\hline 
\textcolor{green}{d} & d & d & \textcolor{red}{c}\tabularnewline
\hline 
\end{tabular}}\end{minipage}\hfill\justifiedminipage{$\left(1,b\right)$ and $\left(4,c\right)$
are determined by efficiency. $\left(1,d\right)$ is determined by
considering the profile where agent $1$ ranks $bacd$ and by the
lemma, and similarly $\left(i,c\right)$ for $i=2,3$ are determined
by considering the profiles where each agent $i$ swaps $\left\{ a,b\right\} $.
$\left(1,a\right)$ is then determined by}\vspace{1mm} complementarity,
as required.

\item {\centering{\itshape%
\begin{tabular}{|c|c|c|c|}
\hline 
\uline{a} & a & b & \textcolor{green}{b}\tabularnewline
\hline 
b & b & a & \textcolor{green}{c}\tabularnewline
\hline 
\textcolor{brown}{c} & \textcolor{brown}{c} & \textcolor{brown}{c} & \textcolor{green}{x}\tabularnewline
\hline 
\textcolor{brown}{d} & \textcolor{brown}{d} & \textcolor{brown}{d} & \textcolor{green}{y}\tabularnewline
\hline 
\end{tabular}}\rlap{(\textcolor{blue}{\ref{bacd - c,d determined}})}\par}\vspace{1mm}

Agent $4$ is not supported because of the pair containing $a$ and
the house immediately above it. It then suffices to determine $\left(1,a\right)$
in the following profile.\vspace{1mm}

\begin{minipage}[t]{0.17\linewidth}{\itshape%
\begin{tabular}{|c|c|c|c|}
\hline 
\textbf{\textcolor{violet}{a}} & a & b & b\tabularnewline
\hline 
\textcolor{orange}{d} & b & a & c\tabularnewline
\hline 
\textcolor{red}{b} & c & c & x\tabularnewline
\hline 
\textcolor{red}{c} & \textcolor{green}{d} & \textcolor{green}{d} & y\tabularnewline
\hline 
\end{tabular}}\end{minipage}\hfill\justifiedminipage{$\left(1,b\right)$ and $\left(1,c\right)$
are determined by efficiency. Then, $\left(i,d\right)$ for $i=2,3$
is determined by considering the profile where agent $i$ swaps $\left\{ a,b\right\} $
and by the lemma. Since $\left(4,a\right)$ is determined by efficiency
in both profiles, Remark \textcolor{blue}{\ref{efficiency swaps}}
implies that it suffices to determine $\left(4,d\right)$ at the profile
where}\vspace{1mm} $aP_{4}d$. In that profile, $\left(4,d\right)$
is determined by considering the profile where agent $4$ changes
his ranking to $abcd$ and by the lemma. $\left(1,a\right)$ is then
determined (in both profiles) by complementarity, as required.\vspace{1mm}

\item {\centering{\itshape%
\begin{tabular}{|c|c|c|c|}
\hline 
a & a & b & b\tabularnewline
\hline 
b & b & a & d\tabularnewline
\hline 
c & c & c & x\tabularnewline
\hline 
d & d & d & y\tabularnewline
\hline 
\end{tabular}}\par}\vspace{1mm}

We use similar arguments to those applied in \hyperlink{bacd, badc}{{\itshape Case} ii}.\footnote{This time, agent $4$ is not supported because of the pair containing
$a$ and the house ranked immediately above it.}\end{casenv}

\item When agent $3$ ranks $badc$, all such profiles are not supported.
\begin{itemize}
\item When agent $4$ ranks either $badc$ or $bdxy$, agents $1$, $2$,
and $3$ are not supported because of the pairs $\left\{ a,b\right\} $
and $\left\{ c,d\right\} $.
\item Otherwise, agent $4$ ranks $bcxy$. In that case, agent $3$ is not
supported because of the pair $\left\{ d,c\right\} $, and agent $4$
is not supported because of the pair containing $a$ and the house
ranked immediately above it.
\end{itemize}
\item Otherwise, let $h\in\left\{ c,d\right\} $ denote the top choice
of agents $3$ and $4$. Note that agents $1$ and $2$ are not supported
because of the pair containing $h$ and the house ranked immediately
above it.
\begin{itemize}
\item When agent $4$ does not rank $a$ second, he is not supported because
of the pair containing $a$ and the house ranked immediately above
it.
\item When agent $4$ does rank $a$ second, if he has the same ranking
as agent $3$, then it suffices that only agents $1$ and $2$ are
not supported for the profile to be determined. Otherwise, let $h^{\prime}$
and $h^{\prime\prime}$ denote the two houses in $H\setminus\left\{ a,h\right\} $.
Note that either agent $3$ or agent $4$ is not supported because
of the pair $\left\{ h^{\prime},h^{\prime\prime}\right\} $.\footnote{Exactly one of them orders $\left\{ h^{\prime},h^{\prime\prime}\right\} $
in the same way as agents $1$ and $2$, so the other must be unsupported.}
\end{itemize}
\end{casenv}

\subsection{Profiles with all agents holding distinct rankings}

In this part we examine the profiles in which all four agents hold
different rankings. Our task is to determine whether such a profile
is supported and to show that, whenever a profile is supported, the
axioms uniquely determine its corresponding assignment matrix.

Recall that a profile is supported if it contains at least two supported
agents with different rankings. Moreover, the axioms are invariant
under renamings of both agents and houses. Since in the present setting
all four agents have distinct rankings, the agents play symmetric
roles. Taken together, these facts imply that we may assume, without
loss of generality, that two supported agents are indexed as agents
$1$ and $2$, that the houses are renamed so that agent $1$'s ranking
is $abcd$, and that the ranking of agent $3$ precedes that of agent
$4$ in the lexicographic order. Any supported profile with four distinct
rankings can be brought into this normalized form by appropriate renamings
of agents and houses.

The assumption that agents $1$ and $2$ are supported imposes a system
of constraints on the rankings of agents $3$ and $4$. For each agent
$i\in\left\{ 1,2\right\} $ and each adjacent pair $xP_{i}^{+}y$
in his ranking, at least two other agents must agree with his comparison,
except that at most one such requirement may be relaxed by one, provided
that both houses in that pair are efficient for that agent. Once the
rankings of agents $1$ and $2$ are fixed, these conditions determine
certain comparisons that must be satisfied by agents $3$ and $4$.
For example, if agent $2$ ranks $adbc$, then incorporating his preferences
into the constraints for agent $1$ yields the following remaining
requirements for agents $3$ and $4$: one of them must prefer $a$
to $b$, one must prefer $b$ to $c$, and both must prefer $c$ to
$d$. As before, at most one of these requirements may be relaxed
by one when both houses in the corresponding pair are efficient for
agent $1$.

We now formalize the notation used to describe these constraints.
Fix two rankings $P_{1},P_{2}\in\mathcal{R}$, a supported agent $i\in\left\{ 1,2\right\} $,
and an adjacent pair $xP_{i}^{+}y$ in his ranking. Define $\widetilde{R}_{xy}^{i}\in\left\{ 1,2\right\} $
to be the number of agents in $\left\{ 3,4\right\} $ who are required
to prefer $x$ to $y$ so that agent $i$ has at least two other agents
agreeing with him on that comparison. Formally, 
\[
\widetilde{R}_{xy}^{i}\coloneqq\begin{cases}
1 & \text{if \ensuremath{xP_{3-i}^{+}y}},\\
2 & \text{otherwise}.
\end{cases}
\]
For illustration, if agent $2$ ranks $adbc$, then the constraints
from agent $1$ are $\widetilde{R}_{ab}^{1}=1$, $\widetilde{R}_{bc}^{1}=1$,
and $\widetilde{R}_{cd}^{1}=2$, and the constraints from agent $2$
are $\widetilde{R}_{ad}^{2}=1$, $\widetilde{R}_{db}^{2}=2$, and
$\widetilde{R}_{bc}^{2}=1$.

For each supported agent $i\in\left\{ 1,2\right\} $, at most one
of the constraints $\widetilde{R}_{xy}^{i}$ may be relaxed by one,
and this is allowed only when both houses in that pair are efficient
for that agent. We denote the chosen relaxation pair by $L_{i}\in\left(H\times H\right)\cup\left\{ \varnothing\right\} $,
where $L_{i}=\varnothing$ means that no relaxation is applied for
agent $i$. The relaxed constraints are 
\[
R_{xy}^{i}\coloneqq\begin{cases}
\widetilde{R}_{xy}^{i}-1 & \text{if \ensuremath{L_{i}=\left(x,y\right)}},\\
\widetilde{R}_{xy}^{i} & \text{otherwise}.
\end{cases}
\]
Finally, if $\left(x,y\right)$ is an adjacent pair with $x$ immediately
above $y$ in either $P_{1}$ or $P_{2}$, then the requirement imposed
on agents $3$ and $4$ after applying $L_{1}$ and $L_{2}$ is 
\[
R_{xy}\coloneqq\max_{i\in\left\{ 1,2\right\} }\left\{ R_{xy}^{i}\right\} .
\]
Here, we adopt the convention that $R_{xy}^{i}\coloneqq0$ whenever
$\left\{ x,y\right\} $ is not an adjacent pair in $P_{i}$ with $xP_{i}^{+}y$.
We record some remarks that will be used repeatedly throughout the
case analysis.
\begin{rem}
\label{rem:same param share ranking}When determining the assignment
matrix of a given profile, we may assume that any profile with the
same disagreement parameter in which at least two agents share the
same ranking is already determined, since all such profiles were exhaustively
analyzed in the previous section.
\end{rem}

\begin{rem}
\label{both adjacent, reversed order}Suppose agents $1$ and $2$
share an adjacent pair $\left\{ x,y\right\} $, with agent $1$ preferring
$x$ and agent $2$ preferring $y$. Then $\widetilde{R}_{xy}^{1}=2$
and $\widetilde{R}_{yx}^{2}=2$. Since agents $3$ and $4$ can satisfy
both requirements $R_{xy}$ and $R_{yx}$ only when $R_{xy}+R_{yx}\leq2$,
supportedness of both agents $1$ and $2$ forces $L_{1}=\left(x,y\right)$
and $L_{2}=\left(y,x\right)$. Moreover, if more than one such conflicting
pair exists, the assumption that both agents $1$ and $2$ are supported
cannot hold.
\end{rem}

\begin{rem}
\label{both adjacent, same order}Suppose agents $1$ and $2$ share
an adjacent pair $\left\{ x,y\right\} $ and both prefer $x$. Then
$\widetilde{R}_{xy}^{1}=\widetilde{R}_{xy}^{2}=1$. In this situation,
relaxing the requirement at $\left(x,y\right)$ for only one of the
two agents is formally possible but always redundant. For instance,
if $L_{1}=\left(x,y\right)$ and $L_{2}\neq\left(x,y\right)$, then
$R_{xy}=1$, so at least one of agents $3$ and $4$ prefers $x$
to $y$; in such a profile, agent $1$ already has two agents agreeing
with him, so taking $\left(x,y\right)$ as his relaxation pair removes
no constraint and moreover adds the requirement that both $x$ and
$y$ must be efficient for him. It follows that relaxing the requirement
at $\left(x,y\right)$ for only one agent never produces any new profiles
to check; it only excludes profiles by adding unnecessary efficiency
conditions. We may therefore assume without loss of generality that
the relaxation at $\left(x,y\right)$ is either applied for both agents
or for neither of them.
\end{rem}

\begin{rem}
\label{x,y,z}For any three distinct houses $x,y,z\in H$, agents
$3$ and $4$ can satisfy the three requirements $R_{xy}$, $R_{yz}$,
and $R_{zx}$ only when $R_{xy}+R_{yz}+R_{zx}\leq4$. This is because
each of agents $3$ and $4$ can satisfy at most two of the corresponding
binary comparisons in his ranking. In particular, for any $i\in\left\{ 1,2\right\} $,
whenever $\widetilde{R}_{xy}^{i}=2$, $\widetilde{R}_{yz}^{3-i}=2$,
and $\widetilde{R}_{zx}^{1}=\widetilde{R}_{zx}^{2}=1$, one of these
requirements must be relaxed through the choice of the relaxation
pairs. Thus, we can assume that either $L_{i}=\left(x,y\right)$,
or $L_{3-i}=\left(y,z\right)$, or $L_{1}=L_{2}=\left(z,x\right)$.
\end{rem}

\begin{rem}
\label{abdc 1,2 a,b}\label{acbd 1,2 a,d}\label{bacd 1,2 c,d}Whenever
agent $2$'s ranking is obtained from $abcd$ by an adjacent swap,
the probabilities associated with the two houses that are not involved
in that swap for agents $1$ and $2$ are determined. For example,
when agent $2$ ranks $abdc$, for agents $1$ and $2$ to be supported,
exactly one of the agents $3$ and $4$ must rank $c$ above $d$.
Consequently, whenever $c$ and $d$ are adjacent in the ranking of
an agent, swapping them yields a profile with a strictly smaller disagreement
parameter. Thus, the probabilities $\left(i,h\right)$ for $i\in\left\{ 1,2\right\} $
and $h\in\left\{ a,b\right\} $ are determined by the induction hypothesis
and SP. The same argument applies when agent $2$ ranks $acbd$ and
$h\in\left\{ a,d\right\} $ and when agent $2$ ranks $bacd$ and
$h\in\left\{ c,d\right\} $.
\end{rem}

\begin{rem}
\label{rem:R_xd>0 and 3 prefers d}Note that whenever $R_{xd}>0$
for some $x\in H\setminus\left\{ d\right\} $, we may assume that
agent $3$ does not rank $d$ first. This is because we assume that
his ranking precedes that of agent $4$ in the lexicographic order,
so if he ranks $d$ first, agent $4$ ranks $d$ first as well, and
the constraint is not satisfied.
\end{rem}

\begin{defn}
A profile is called \emph{degenerate} if each agent ranks a different
house first.
\end{defn}

\begin{rem}
\label{rem:degenerate}Every degenerate profile is determined, because
efficiency forces each agent to receive his top-ranked house.
\end{rem}

\begin{rem}
\label{rem:R_xc>0, 2-b and 3-c}Note that whenever agent $2$ ranks
$b$ first and $R_{xc}>0$ for some $x\in H\setminus\left\{ c\right\} $,
we may assume that agent $3$ does not rank $c$ first. This is because
we assume that his ranking precedes that of agent $4$ in the lexicographic
order, so if he ranks $c$ first, in order to satisfy the constraint
agent $4$ must rank $d$ first, and we obtain a degenerate profile;
hence it is determined by Remark \textcolor{blue}{\ref{rem:degenerate}.}
\end{rem}

\begin{rem}
\label{rem:2 and 3 prefer b, 4 not, and b in L_1}Note that whenever
agent $2$ ranks $b$ first and $b$ is involved in the relaxation
pair of agent $1$ (that is, $L_{1}=\left(a,b\right)$ or $L_{1}=\left(b,c\right)$),
we may assume that agent $3$ does not rank $b$ first. This is because
we assume that his ranking precedes that of agent $4$ in the lexicographic
order, so if he ranks $b$ first, agent $4$ cannot rank $a$ first,
and then agent $1$ cannot receive $b$.
\end{rem}

We now proceed to examine all possible profiles, grouped according
to the ranking of agent $2$.
\begin{casenv}
\item When agent $2$ ranks $abdc$, $\left\{ c,d\right\} $ is an adjacent
pair in both $P_{1}$ and $P_{2}$, with agent $1$ preferring $c$
and agent $2$ preferring $d$. By Remark \textcolor{blue}{\ref{both adjacent, reversed order}},
we must therefore have $L_{1}=\left(c,d\right)$ and $L_{2}=\left(d,c\right)$.
The constraints imposed on the rankings of agents $3$ and $4$ are
\[
R_{ab}=R_{bc}=R_{cd}=R_{bd}=R_{dc}=1.
\]
Since $L_{1}=\left(c,d\right)$ and $L_{2}=\left(d,c\right)$, the
houses $c$ and $d$ must both be efficient for agents $1$ and $2$,
and consequently agents $3$ and $4$ cannot rank either $c$ or $d$
in the top position. Moreover, agent $3$ cannot rank $b$ first either:
since its ranking precedes that of agent $4$, this would force agent
$4$ to rank $b$ first, but then the requirement $R_{ab}=1$ would
not be satisfied. We now consider subcases according to the ranking
of agent $3$ (recall that we do not examine cases where two agents
share the same ranking).
\begin{casenv}
\item When agent $3$ ranks $acbd$, this ranking satisfies $aP_{3}b$,
$cP_{3}d$, and $bP_{3}d$. The constraints therefore require that
agent $4$'s ranking satisfy $bP_{4}c$ and $dP_{4}c$. After taking
these restrictions into account as well, the remaining possibilities
for $P_{4}$ are $adbc$, $badc$, and rankings of the form $bdxy$.\vspace{1mm}

\begin{casenv}[leftmargin=0pt]\item {\centering{\itshape%
\begin{tabular}{|c|c|c|c|}
\hline 
\textcolor{brown}{a} & \textcolor{brown}{a} & \textcolor{green}{a} & \textcolor{green}{a}\tabularnewline
\hline 
\textcolor{brown}{b} & \textcolor{brown}{b} & \textcolor{green}{c} & \textcolor{green}{d}\tabularnewline
\hline 
c & d & \textcolor{red}{b} & \textcolor{red}{b}\tabularnewline
\hline 
\uline{d} & c & \textcolor{green}{d} & \textcolor{green}{c}\tabularnewline
\hline 
\end{tabular}}\rlap{(\textcolor{blue}{\ref{abdc 1,2 a,b}})}\par}\vspace{1mm}In
this profile, agents $3$ and $4$ are not supported because of the
pairs $\left\{ c,b\right\} $ and $\left\{ d,b\right\} $, respectively.
It then suffices to determine $\left(1,d\right)$ in the following
profile.\vspace{1mm}

\begin{minipage}[t]{0.19\linewidth}{\itshape%
\begin{tabular}{|c|c|c|c|}
\hline 
\textcolor{blue}{a} & \textcolor{blue}{a} & \textcolor{blue}{a} & \textcolor{blue}{a}\tabularnewline
\hline 
\textcolor{blue}{c} & \textcolor{green}{b} & \textcolor{blue}{c} & d\tabularnewline
\hline 
\textcolor{orange}{b} & d & \textcolor{orange}{b} & \textcolor{red}{b}\tabularnewline
\hline 
\textbf{\textcolor{violet}{d}} & \textcolor{red}{c} & d & \textcolor{red}{c}\tabularnewline
\hline 
\end{tabular}

}\end{minipage}\hfill\modifiedjustifiedminipage{0.79}{$\left(2,c\right),\left(4,b\right),$
and $\left(4,c\right)$ are determined by efficiency, and consequently
$a$ and $c$ are determined by ETA. $\left(2,b\right)$ is determined
by considering the profile where agent $2$ swaps $\left\{ c,d\right\} $.
That profile's disagreement parameter equals that of}\vspace{1mm} the original
profile, and it is not supported.\footnote{In that profile, both agents $2$ and $4$ are not supported due to
the pair $\left\{ b,c\right\} $, and the remaining agents share the
same ranking.} $\left(1,d\right)$ is then determined by complementarity, since
agents $1$ and $3$ share the same ranking.\vspace{1mm}

\item {\centering{\itshape%
\begin{tabular}{|c|c|c|c|}
\hline 
\textcolor{brown}{a} & \textcolor{brown}{a} & \textcolor{green}{a} & \textcolor{green}{b}\tabularnewline
\hline 
\textcolor{brown}{b} & \textcolor{brown}{b} & \textcolor{green}{c} & \textcolor{red}{a}\tabularnewline
\hline 
c & d & \textcolor{red}{b} & \textcolor{green}{d}\tabularnewline
\hline 
\uline{d} & c & \textcolor{green}{d} & \textcolor{green}{c}\tabularnewline
\hline 
\end{tabular}}\rlap{(\textcolor{blue}{\ref{abdc 1,2 a,b}})}\par}\vspace{1mm}In
this profile, agents $3$ and $4$ are not supported because of the
pairs $\left\{ c,b\right\} $ and $\left\{ b,a\right\} $, respectively.
It then suffices to determine $\left(1,d\right)$ in the following
profile.\vspace{1mm}

\begin{minipage}[t]{0.19\linewidth}{\itshape%
\begin{tabular}{|c|c|c|c|}
\hline 
c & a & a & b\tabularnewline
\hline 
a & b & c & a\tabularnewline
\hline 
b & \uline{d} & b & \uline{d}\tabularnewline
\hline 
\textbf{d} & c & \uline{d} & c\tabularnewline
\hline 
\end{tabular}

}\end{minipage}\hfill\modifiedminipar{0.79}{By house complementarity,
it suffices to determine $\left(i,d\right)$ for $i=2,3,4$.}\vspace{1mm}
\begin{itemize}
\item \begin{minipage}[t]{0.20\linewidth}{\itshape%
\begin{tabular}{|c|c|c|c|}
\hline 
c & \textcolor{blue}{b} & \uline{a} & \textcolor{blue}{b}\tabularnewline
\hline 
\textcolor{red}{a} & a & c & a\tabularnewline
\hline 
\textcolor{red}{b} & \textbf{d} & \textcolor{red}{b} & d\tabularnewline
\hline 
d & \textcolor{red}{c} & d & c\tabularnewline
\hline 
\end{tabular}}\end{minipage}\hfill\modifiedjustifiedminipage{0.78}{We can use
this profile to determine $\left(2,d\right)$. $\left(1,a\right)$,
$\left(1,b\right)$, $\left(2,c\right)$, and $\left(3,b\right)$
are determined by efficiency, and consequently $b$ is determined
by ETA. It suffices to determine $\left(3,a\right)$ in the following
profile (since agents $2$}\vspace{1mm} and $4$ share the same
ranking).\vspace{1mm}
\begin{itemize}
\item \begin{minipage}[t]{0.205\linewidth}{\itshape%
\begin{tabular}{|c|c|c|c|}
\hline 
\textcolor{blue}{c} & \textcolor{blue}{b} & \textbf{\textcolor{violet}{a}} & \textcolor{blue}{b}\tabularnewline
\hline 
\textcolor{red}{a} & \textcolor{violet}{a} & \textcolor{red}{b} & \textcolor{violet}{a}\tabularnewline
\hline 
\textcolor{red}{b} & \textcolor{blue}{d} & \textcolor{blue}{d} & \textcolor{blue}{d}\tabularnewline
\hline 
\textcolor{red}{d} & \textcolor{red}{c} & \textcolor{red}{c} & \textcolor{red}{c}\tabularnewline
\hline 
\end{tabular}}\end{minipage}\hfill\modifiedminipar{0.775}{Efficiency and ETA determine
$b$, $c$, and $d$, and consequently the profile is determined.}\vspace{1mm}
\end{itemize}
\item \begin{minipage}[t]{0.20\linewidth}{\itshape%
\begin{tabular}{|c|c|c|c|}
\hline 
\textcolor{violet}{c} & \textcolor{blue}{a} & \textcolor{blue}{a} & b\tabularnewline
\hline 
\textcolor{red}{a} & b & \textcolor{green}{b} & \textcolor{red}{a}\tabularnewline
\hline 
\textcolor{red}{b} & d & \textcolor{orange}{c} & d\tabularnewline
\hline 
\textcolor{green}{d} & \textcolor{red}{c} & \textbf{\textcolor{violet}{d}} & \textcolor{red}{c}\tabularnewline
\hline 
\end{tabular}}\end{minipage}\hfill\modifiedjustifiedminipage{0.78}{We can use
this profile to determine $\left(3,d\right)$. $\left(1,a\right)$,
$\left(1,b\right)$, $\left(2,c\right)$, $\left(4,a\right)$, and
$\left(4,c\right)$ are determined by efficiency, and consequently
$a$ is determined by ETA. $\left(3,b\right)$ is determined by considering
the profile where agent $3$}\vspace{1mm} swaps $\left\{ c,d\right\} $ and
applying the lemma. $\left(1,d\right)$ is determined by considering
the profile where agent $1$ swaps $\left\{ c,a\right\} $. That profile's
disagreement parameter equals that of the original profile, and agent
$1$ is not supported there.\footnote{Agent $1$ is not supported in that profile because of the pair $\left\{ c,b\right\} $.}
By Remark \textcolor{blue}{\ref{unsupported profile equal param}},
the entries of that agent are determined in the swapped profile, and
therefore $\left(1,d\right)$ is determined in the current profile.
$\left(3,d\right)$ is then determined by complementarity.\vspace{1mm}
\item \begin{minipage}[t]{0.20\linewidth}{\itshape%
\begin{tabular}{|c|c|c|c|}
\hline 
c & a & a & a\tabularnewline
\hline 
a & b & c & b\tabularnewline
\hline 
b & \textcolor{orange}{d} & b & \textbf{\textcolor{orange}{d}}\tabularnewline
\hline 
\textcolor{green}{d} & c & \textcolor{green}{d} & c\tabularnewline
\hline 
\end{tabular}}\end{minipage}\hfill\modifiedjustifiedminipage{0.78}{We can use
this profile to determine $\left(4,d\right)$. $\left(1,d\right)$
is determined by considering the profile where agent $1$ swaps $\left\{ c,a\right\} $.
That profile has a smaller disagreement parameter than the original
one. $\left(3,d\right)$ is determined by}\vspace{1mm} considering
the profile where agent $3$ swaps $\left\{ c,b\right\} $. That profile's
disagreement parameter equals that of the original profile, and it
is not supported.\footnote{That profile is not supported because, for agent $1$ we may consider
$\left\{ c,a\right\} $, and for agents $2$ and $4$ we may consider
$\left\{ d,c\right\} $.} Since agents $2$ and $4$ share the same ranking, $\left(4,d\right)$
is then determined.\vspace{1mm}
\end{itemize}
\item {\centering{\itshape%
\begin{tabular}{|c|c|c|c|}
\hline 
\textcolor{brown}{a} & \textcolor{brown}{a} & \textcolor{green}{a} & \textcolor{green}{b}\tabularnewline
\hline 
\textcolor{brown}{b} & \textcolor{brown}{b} & \textcolor{green}{c} & \textcolor{green}{d}\tabularnewline
\hline 
c & d & \textcolor{red}{b} & \textcolor{green}{x}\tabularnewline
\hline 
d & \uline{c} & \textcolor{green}{d} & \textcolor{green}{y}\tabularnewline
\hline 
\end{tabular}}\rlap{(\textcolor{blue}{\ref{abdc 1,2 a,b}})}\par}\vspace{1mm}In
this profile, agents $3$ and $4$ are not supported: agent $3$ because
of the pair $\left\{ c,b\right\} $, and agent $4$ because of the
pair containing $a$ and the house immediately above it. It then suffices
to determine $\left(2,c\right)$ in the following profile.\vspace{1mm}

\begin{minipage}[t]{0.19\linewidth}{\itshape

\begin{tabular}{|c|c|c|c|}
\hline 
a & \textcolor{orange}{d} & a & \textcolor{green}{b}\tabularnewline
\hline 
b & \textcolor{red}{a} & c & \textcolor{violet}{d}\tabularnewline
\hline 
c & \textcolor{red}{b} & \textcolor{red}{b} & x\tabularnewline
\hline 
\textcolor{red}{d} & \textbf{\textcolor{violet}{c}} & \textcolor{red}{d} & y\tabularnewline
\hline 
\end{tabular}}\end{minipage}\hfill\modifiedjustifiedminipage{0.79}{$\left(1,d\right)$,
$\left(2,a\right)$, $\left(2,b\right)$, $\left(3,b\right)$, $\left(3,d\right)$,
and $\left(4,a\right)$ are determined by efficiency. $\left(4,b\right)$
is determined by considering the profile where agent $4$ ranks $bacd$
and applying the lemma. Since $\left(4,a\right)$ is determined by
efficiency in both}\vspace{1mm} profiles, Remark \textcolor{blue}{\ref{efficiency swaps}}
implies that it suffices to determine $\left(4,c\right)$
at the profile where $aP_{4}c$. In that profile, $\left(4,c\right)$
is determined by considering the profile where agent $4$ ranks $abdc$.
That profile's disagreement parameter is at most that of the original
profile, and it is not supported.\footnote{In that profile, agents $1$, $2$, and $3$ are not supported because
of the pairs $\left\{ c,d\right\} $, $\left\{ d,a\right\} $, and
$\left\{ c,b\right\} $, respectively. } $\left(2,c\right)$ is then determined by complementarity.

\end{casenv}
\item When agent $3$ ranks $acdb$, the constraints require that $bP_{4}dP_{4}c$.
Thus, the remaining possibilities for $P_{4}$ are of the form $bxyz$
with $dP_{4}c$.
\begin{rem}
\label{abdc,a*** 3,4 not supported}In this case and the next two
cases, agent $3$ is not supported because of the pair containing
$b$, and agent $4$ is not supported because of the pair containing
$a$, where in both instances we refer to the pair formed with the
house immediately above them in their respective rankings.\vspace{1mm}
\end{rem}

\begin{casenv}[leftmargin=0pt]

\item[] {\centering{\itshape%
\begin{tabular}{|c|c|c|c|}
\hline 
\textcolor{brown}{a} & \textcolor{brown}{a} & \textcolor{brown}{a} & \textcolor{brown}{b}\tabularnewline
\hline 
\textcolor{brown}{b} & \textcolor{brown}{b} & \textcolor{brown}{c} & \textcolor{brown}{x}\tabularnewline
\hline 
c & d & \textcolor{brown}{d} & \textcolor{brown}{y}\tabularnewline
\hline 
\uline{d} & c & \textcolor{brown}{b} & \textcolor{brown}{z}\tabularnewline
\hline 
\end{tabular}}\rlap{(\textcolor{blue}{\ref{abdc 1,2 a,b}},\textcolor{blue}{\ref{abdc,a*** 3,4 not supported}})}\par}\vspace{1mm}It
suffices to determine $\left(1,d\right)$ in the following profile.\vspace{1mm}

\begin{minipage}[t]{0.19\linewidth}{\itshape%
\begin{tabular}{|c|c|c|c|}
\hline 
\uline{c} & \textcolor{blue}{a} & \textcolor{blue}{a} & \uline{b}\tabularnewline
\hline 
\textcolor{red}{a} & b & c & x\tabularnewline
\hline 
b & \uline{d} & d & y\tabularnewline
\hline 
\textbf{d} & \textcolor{red}{c} & \textcolor{red}{b} & z\tabularnewline
\hline 
\end{tabular}}\end{minipage}\hfill\modifiedminipar{0.79}{$\left(1,a\right)$,
$\left(2,c\right)$, $\left(3,b\right)$, and $\left(4,a\right)$
are determined by efficiency, and consequently $a$ is determined
by ETA. It then suffices to determine $\left(1,c\right)$, $\left(2,d\right)$,
and $\left(4,b\right)$.}\vspace{1mm}
\begin{itemize}
\item \begin{minipage}[t]{0.20\linewidth}{\itshape%
\begin{tabular}{|c|c|c|c|}
\hline 
\textbf{c} & a & a & b\tabularnewline
\hline 
\textcolor{red}{a} & b & c & x\tabularnewline
\hline 
\uline{d} & d & d & y\tabularnewline
\hline 
\textcolor{red}{b} & c & b & z\tabularnewline
\hline 
\end{tabular}}\end{minipage}\hfill\modifiedminipar{0.78}{We can use this profile
to determine $\left(1,c\right)$. $\left(1,a\right)$ and $\left(1,b\right)$
are determined by efficiency. It then suffices to determine $\left(1,d\right)$
in the following profile.}\vspace{1mm}
\begin{itemize}
\item \begin{minipage}[t]{0.205\linewidth}{\itshape%
\begin{tabular}{|c|c|c|c|}
\hline 
\textcolor{blue}{a} & \textcolor{blue}{a} & \textcolor{blue}{a} & b\tabularnewline
\hline 
\textcolor{blue}{c} & b & \textcolor{blue}{c} & x\tabularnewline
\hline 
\textbf{\textcolor{violet}{d}} & d & d & y\tabularnewline
\hline 
\textcolor{red}{b} & \textcolor{red}{c} & b & z\tabularnewline
\hline 
\end{tabular}}\end{minipage}\hfill\modifiedminipar{0.775}{ $\left(1,b\right)$,
$\left(2,c\right)$, and $\left(4,a\right)$ are determined by efficiency.
Since $dP_{4}c$, $\left(4,c\right)$ is also determined by efficiency,
and consequently $a$ and $c$ are determined by ETA. $\left(1,d\right)$
is then determined by complementarity.}\vspace{1mm}
\end{itemize}
\item \begin{minipage}[t]{0.20\linewidth}{\itshape%
\begin{tabular}{|c|c|c|c|}
\hline 
c & \textcolor{blue}{b} & \textcolor{green}{a} & \textcolor{blue}{b}\tabularnewline
\hline 
\textcolor{red}{a} & \textcolor{orange}{a} & c & x\tabularnewline
\hline 
\textcolor{red}{b} & \textbf{\textcolor{violet}{d}} & d & y\tabularnewline
\hline 
d & \textcolor{red}{c} & \textcolor{red}{b} & z\tabularnewline
\hline 
\end{tabular}}\end{minipage}\hfill\modifiedjustifiedminipage{0.78}{We can use
this profile to determine $\left(2,d\right)$. $\left(1,a\right)$,
$\left(1,b\right)$, $\left(2,c\right)$, and $\left(3,b\right)$
are determined by efficiency, and consequently $b$ is determined
by ETA. $\left(3,a\right)$ is determined by considering the profile
where agent $3$}\vspace{1mm} ranks $abdc$. That profile's disagreement
parameter is at most that of the original profile, and agent $3$
is not supported there.\footnote{Agent $3$ is not supported in that profile because of the pair $\left\{ a,b\right\} $.
} Note that if $x=a$, then agents $2$ and $4$ share the same ranking,
and otherwise, since $dP_{4}c$, $\left(4,a\right)$ is determined
by efficiency. In either case, $\left(2,d\right)$ is then determined
by complementarity.\vspace{1mm}
\item \begin{minipage}[t]{0.20\linewidth}{\itshape%
\begin{tabular}{|c|c|c|c|}
\hline 
c & a & a & \textbf{\textcolor{orange}{b}}\tabularnewline
\hline 
a & \textcolor{green}{b} & c & a\tabularnewline
\hline 
\textcolor{green}{b} & d & d & c\tabularnewline
\hline 
d & c & \textcolor{red}{b} & d\tabularnewline
\hline 
\end{tabular}}\end{minipage}\hfill\modifiedjustifiedminipage{0.78}{We can use
this profile to determine $\left(4,b\right)$. $\left(3,b\right)$
is determined by efficiency. $\left(1,b\right)$ is determined by
considering the profile where agent $1$ swaps $\left\{ c,a\right\} $.
That profile's disagreement parameter is at most that of the}\vspace{1mm} original
profile, and it is not supported.\footnote{In that profile, agents $2$, $3$, and $4$ are not supported because
of the pairs $\left\{ d,c\right\} $, $\left\{ d,b\right\} $, and
$\left\{ b,a\right\} $, respectively. } Similarly, $\left(2,b\right)$ is determined by considering the profile
where agent $2$ swaps $\left\{ d,c\right\} $.\footnote{In that profile, agents $1$, $3$, and $4$ are not supported because
of the pairs $\left\{ c,a\right\} $, $\left\{ d,b\right\} $, and
$\left\{ b,a\right\} $, respectively. } $\left(4,b\right)$ is then determined by complementarity.
\end{itemize}
\end{casenv}
\item When agent $3$ ranks $adbc$, the constraints require that $cP_{4}d$
and $bP_{4}d$. Thus, the remaining possibilities for $P_{4}$ are
$bacd$ and rankings of the form $bcxy$.\vspace{1mm}

\begin{casenv}[leftmargin=0pt]\item {\centering{\itshape%
\begin{tabular}{|c|c|c|c|}
\hline 
\textcolor{brown}{a} & \textcolor{brown}{a} & \textcolor{brown}{a} & \textcolor{brown}{b}\tabularnewline
\hline 
\textcolor{brown}{b} & \textcolor{brown}{b} & \textcolor{brown}{d} & \textcolor{brown}{a}\tabularnewline
\hline 
\uline{c} & d & \textcolor{brown}{b} & \textcolor{brown}{c}\tabularnewline
\hline 
d & c & \textcolor{brown}{c} & \textcolor{brown}{d}\tabularnewline
\hline 
\end{tabular}}\rlap{(\textcolor{blue}{\ref{abdc 1,2 a,b}},\textcolor{blue}{\ref{abdc,a*** 3,4 not supported}})}\par}\vspace{1mm}It
suffices to determine $\left(1,c\right)$ in the following profile.\vspace{1mm}

\begin{minipage}[t]{0.19\linewidth}{\itshape%
\begin{tabular}{|c|c|c|c|}
\hline 
b & a & a & b\tabularnewline
\hline 
a & b & d & a\tabularnewline
\hline 
\textbf{c} & d & b & c\tabularnewline
\hline 
d & \uline{c} & \textcolor{green}{c} & d\tabularnewline
\hline 
\end{tabular}}\end{minipage}\hfill\modifiedjustifiedminipage{0.79}{$\left(3,c\right)$
is determined by considering the profile where agent $3$ swaps $\left\{ d,b\right\} $.
That profile's disagreement parameter is less than that of the original
profile. It then suffices to determine $\left(2,c\right)$ in the
following profile (since agents $1$ and $4$ share the}\vspace{1mm}
same ranking).\vspace{1mm}

\begin{minipage}[t]{0.19\linewidth}{\itshape%
\begin{tabular}{|c|c|c|c|}
\hline 
b & \textcolor{orange}{d} & \textcolor{green}{a} & b\tabularnewline
\hline 
a & \textcolor{red}{b} & \textcolor{violet}{d} & a\tabularnewline
\hline 
c & \textcolor{red}{a} & \textcolor{red}{b} & c\tabularnewline
\hline 
\textcolor{red}{d} & \textbf{\textcolor{violet}{c}} & \textcolor{green}{c} & \textcolor{red}{d}\tabularnewline
\hline 
\end{tabular}}\end{minipage}\hfill\modifiedjustifiedminipage{0.79}{$\left(1,d\right)$,
$\left(2,b\right)$, $\left(2,a\right)$, $\left(3,b\right)$, and
$\left(4,d\right)$ are determined by efficiency. $\left(3,a\right)$
is determined by considering the profile where agent $3$ ranks $abcd$.
That profile is determined by the lemma. $\left(3,c\right)$ is determined
by considering the profile where}\vspace{1mm} agent $3$ ranks $badc$.
That profile's disagreement parameter equals that of the original
profile, and it is not supported.\footnote{In that profile, agents $1$ and $4$ are not supported because of
the pair $\left\{ c,d\right\} $, and agent $2$ is not supported
because of the pair $\left\{ d,b\right\} $.} $\left(2,c\right)$ is then determined by complementarity.\vspace{1mm}

\item {\centering{\itshape%
\begin{tabular}{|c|c|c|c|}
\hline 
\textcolor{brown}{a} & \textcolor{brown}{a} & \textcolor{brown}{a} & \textcolor{brown}{b}\tabularnewline
\hline 
\textcolor{brown}{b} & \textcolor{brown}{b} & \textcolor{brown}{d} & \textcolor{brown}{c}\tabularnewline
\hline 
c & d & \textcolor{brown}{b} & \textcolor{brown}{x}\tabularnewline
\hline 
\uline{d} & c & \textcolor{brown}{c} & \textcolor{brown}{y}\tabularnewline
\hline 
\end{tabular}}\rlap{(\textcolor{blue}{\ref{abdc 1,2 a,b}},\textcolor{blue}{\ref{abdc,a*** 3,4 not supported}})}\par}\vspace{1mm}It
suffices to determine $\left(1,d\right)$ in the following profile.\vspace{1mm}

\begin{minipage}[t]{0.19\linewidth}{\itshape%
\begin{tabular}{|c|c|c|c|}
\hline 
c & \textcolor{blue}{a} & \textcolor{blue}{a} & \textcolor{green}{b}\tabularnewline
\hline 
\textcolor{red}{a} & \textcolor{orange}{b} & \textcolor{violet}{d} & c\tabularnewline
\hline 
\textcolor{red}{b} & \textcolor{violet}{d} & \textcolor{red}{b} & x\tabularnewline
\hline 
\textbf{\textcolor{orange}{d}} & \textcolor{red}{c} & \textcolor{red}{c} & y\tabularnewline
\hline 
\end{tabular}}\end{minipage}\hfill\modifiedjustifiedminipage{0.79}{$\left(1,a\right)$,
$\left(1,b\right)$, $\left(2,c\right)$, $\left(3,b\right)$, $\left(3,c\right)$,
and $\left(4,a\right)$ are determined by efficiency, and consequently
$a$ is determined by ETA. $\left(4,b\right)$ is determined by considering
the profile where agent $4$ ranks $badc$. That profile is determined
by the lemma. $\left(4,d\right)$}\vspace{1mm} is determined by
considering the profile where agent $4$ ranks $abcd$.\footnote{Note that $\left(4,a\right)$ is determined by efficiency in both
profiles of that form. Hence, by Remark \textcolor{blue}{\ref{efficiency swaps}},
we may use this transition to determine $\left(4,d\right)$.} That profile's disagreement parameter is at most that of the original
profile, and it is not supported.\footnote{In that profile, agents $1$, $2$, and $3$ are not supported because
of the pairs $\left\{ c,a\right\} $, $\left\{ d,c\right\} $, and
$\left\{ d,b\right\} $, respectively.} $\left(1,d\right)$ is then determined by complementarity.

\end{casenv}
\item When agent $3$ ranks $adcb$, the constraints require that $bP_{4}cP_{4}d$.
Thus, the remaining possibilities for $P_{4}$ are of the form $bxyz$
with $cP_{4}d$.\vspace{1mm}

\begin{casenv}[leftmargin=0pt]

\item[] {\centering{\itshape%
\begin{tabular}{|c|c|c|c|}
\hline 
\textcolor{brown}{a} & \textcolor{brown}{a} & \textcolor{brown}{a} & \textcolor{brown}{b}\tabularnewline
\hline 
\textcolor{brown}{b} & \textcolor{brown}{b} & \textcolor{brown}{d} & \textcolor{brown}{x}\tabularnewline
\hline 
c & d & \textcolor{brown}{c} & \textcolor{brown}{y}\tabularnewline
\hline 
\uline{d} & c & \textcolor{brown}{b} & \textcolor{brown}{z}\tabularnewline
\hline 
\end{tabular}}\rlap{(\textcolor{blue}{\ref{abdc 1,2 a,b}},\textcolor{blue}{\ref{abdc,a*** 3,4 not supported}})}\par}\vspace{1mm}It
suffices to determine $\left(1,d\right)$ in the following profile.\vspace{1mm}

\begin{minipage}[t]{0.19\linewidth}{\itshape%
\begin{tabular}{|c|c|c|c|}
\hline 
\textcolor{blue}{a} & \textcolor{blue}{a} & \textcolor{blue}{a} & b\tabularnewline
\hline 
c & \textcolor{green}{b} & \textcolor{green}{d} & x\tabularnewline
\hline 
\textcolor{red}{b} & d & c & y\tabularnewline
\hline 
\textbf{d} & \uline{c} & \textcolor{red}{b} & z\tabularnewline
\hline 
\end{tabular}}\end{minipage}\hfill\modifiedjustifiedminipage{0.79}{$\left(1,b\right)$,
$\left(3,b\right)$, and $\left(4,a\right)$ are determined by efficiency,
and consequently $a$ is determined by ETA. $\left(2,b\right)$, $\left(3,d\right)$,
and $\left(4,d\right)$ are determined by considering, for each of
these agents, a modified profile whose disagreement parameter is at
most}\vspace{1mm} that of the original profile and which is either
not supported or in which that agent is not supported.\footnote{For $\left(2,b\right)$, we consider the profile where agent $2$
swaps $\left\{ d,c\right\} $. In that profile, agents $1$ and $3$
are not supported because of the pair $\left\{ c,b\right\} $, and
agent $4$ is not supported because of the pair containing $a$ and
the house immediately above it. For $\left(3,d\right)$, we consider
the profile where agent $3$ swaps $\left\{ c,b\right\} $. In that
profile, agent $3$ is not supported because of the pair $\left\{ d,b\right\} $.
For $\left(4,d\right)$, since $cP_{4}d$, by Remark \textcolor{blue}{\ref{efficiency swaps}}
when $z=a$, we may consider the profile where agent $4$ ranks $cabd$.
In that profile, agent $4$ is not supported because of the pair $\left\{ c,a\right\} $.} It then suffices to determine $\left(2,c\right)$ in the following
profile.\vspace{1mm}

\begin{minipage}[t]{0.19\linewidth}{\itshape%
\begin{tabular}{|c|c|c|c|}
\hline 
\textcolor{blue}{a} & \textcolor{blue}{a} & \textcolor{blue}{a} & b\tabularnewline
\hline 
\textcolor{violet}{c} & \textcolor{blue}{d} & \textcolor{blue}{d} & x\tabularnewline
\hline 
\textcolor{red}{b} & b & \textcolor{violet}{c} & y\tabularnewline
\hline 
\textcolor{red}{d} & \textbf{\textcolor{orange}{c}} & \textcolor{red}{b} & z\tabularnewline
\hline 
\end{tabular}}\end{minipage}\hfill\modifiedjustifiedminipage{0.79}{$\left(1,b\right)$,
$\left(1,d\right)$, $\left(3,b\right)$, $\left(4,a\right)$, and
$\left(4,d\right)$ are determined by efficiency, and consequently
$a$ and $d$ are determined by ETA. $\left(4,c\right)$ is determined
by considering the profile where agent $4$ ranks $abcd$.\footnote{Since $\left(4,a\right)$ is determined by efficiency at all the relevant
profiles, by Remark \textcolor{blue}{\ref{efficiency swaps}} we may
assume $aP_{4}c$.} That profile's disagreement parameter is less}\vspace{1mm} than
that of the original profiles. $\left(2,c\right)$ is then determined
by complementarity.

\end{casenv}
\end{casenv}
\item When agent $2$ ranks $acbd$, by Remark \textcolor{blue}{\ref{both adjacent, reversed order}},
we must have $L_{1}=\left(b,c\right)$ and $L_{2}=\left(c,b\right)$.
The constraints imposed on the rankings of agents $3$ and $4$ are
\[
R_{ab}=R_{bc}=R_{cd}=R_{ac}=R_{cb}=R_{bd}=1.
\]

\begin{casenv}
\item When agent $3$ ranks $abdc$, the constraints require that $cP_{4}d$
and $cP_{4}b$. Thus, the remaining possibilities for $P_{4}$ are
$acdb$ and rankings of the form $cxyz$.\vspace{1mm}

\begin{casenv}[leftmargin=0pt]

\item {\centering{\itshape%
\begin{tabular}{|c|c|c|c|}
\hline 
\textcolor{brown}{a} & \textcolor{brown}{a} & \textcolor{green}{a} & \textcolor{green}{a}\tabularnewline
\hline 
\uline{b} & c & \textcolor{green}{b} & \textcolor{green}{c}\tabularnewline
\hline 
c & b & \textcolor{green}{d} & \textcolor{green}{d}\tabularnewline
\hline 
\textcolor{brown}{d} & \textcolor{brown}{d} & \textcolor{red}{c} & \textcolor{red}{b}\tabularnewline
\hline 
\end{tabular}}\rlap{(\textcolor{blue}{\ref{acbd 1,2 a,d}})}\par}\vspace{1mm}In
this profile, agents $3$ and $4$ are not supported because of the
pairs $\left\{ d,c\right\} $ and $\left\{ d,b\right\} $, respectively.
It then suffices to determine $\left(1,b\right)$ in the following
profile.\vspace{1mm}

\begin{minipage}[t]{0.19\linewidth}{\itshape%
\begin{tabular}{|c|c|c|c|}
\hline 
\textcolor{blue}{a} & \textcolor{blue}{a} & \textcolor{blue}{a} & \textcolor{blue}{a}\tabularnewline
\hline 
\textbf{\textcolor{orange}{b}} & \textcolor{blue}{c} & \textcolor{orange}{b} & \textcolor{blue}{c}\tabularnewline
\hline 
d & \textcolor{violet}{b} & d & d\tabularnewline
\hline 
\textcolor{red}{c} & \textcolor{green}{d} & \textcolor{red}{c} & \textcolor{red}{b}\tabularnewline
\hline 
\end{tabular}}\end{minipage}\hfill\modifiedjustifiedminipage{0.79}{$\left(1,c\right)$,
$\left(3,c\right)$, and $\left(4,b\right)$ are determined by efficiency,
and consequently $a$ and $c$ are determined by ETA. $\left(2,d\right)$
is determined by considering the profile where agent $2$ ranks $abcd$.
That profile's disagreement parameter equals that of}\vspace{1mm} the original
profile, and agents $1$ and $3$ share the same ranking in it, so
that profile is determined by Remark \textcolor{blue}{\ref{rem:same param share ranking}.}\vspace{1mm}

\item {\centering{\itshape%
\begin{tabular}{|c|c|c|c|}
\hline 
\textcolor{brown}{a} & \textcolor{brown}{a} & \textcolor{green}{a} & \textcolor{green}{c}\tabularnewline
\hline 
\uline{b} & c & \textcolor{green}{b} & \textcolor{green}{x}\tabularnewline
\hline 
c & b & \textcolor{green}{d} & \textcolor{green}{y}\tabularnewline
\hline 
\textcolor{brown}{d} & \textcolor{brown}{d} & \textcolor{red}{c} & \textcolor{green}{z}\tabularnewline
\hline 
\end{tabular}}\rlap{(\textcolor{blue}{\ref{acbd 1,2 a,d}})}\par}\vspace{1mm}In
profiles of this form, agents $3$ and $4$ are not supported: agent
$3$ because of the pair $\left\{ d,c\right\} $, and agent $4$ because
of the pair containing $a$ and the house immediately above it. It
then suffices to determine $\left(1,b\right)$ in the following profile.\vspace{1mm}

\begin{minipage}[t]{0.19\linewidth}{\itshape%
\begin{tabular}{|c|c|c|c|}
\hline 
\textcolor{blue}{a} & \textcolor{blue}{a} & \textcolor{blue}{a} & c\tabularnewline
\hline 
\textbf{\textcolor{violet}{b}} & c & b & x\tabularnewline
\hline 
\textcolor{orange}{d} & b & \textcolor{orange}{d} & y\tabularnewline
\hline 
\textcolor{red}{c} & \textcolor{green}{d} & c & z\tabularnewline
\hline 
\end{tabular}}\end{minipage}\hfill\modifiedjustifiedminipage{0.79}{$\left(1,c\right)$
and $\left(4,a\right)$ are determined by efficiency, and consequently
$a$ is determined by ETA. $\left(2,d\right)$ is determined by the
profile where agent $2$ swaps $\left\{ c,b\right\} $. That profile's
disagreement parameter equals that of the original profile,\footnote{This is true in all options, since agent $4$ ranks $c$ first, and
therefore his binary preference over the pairs $\left\{ c,d\right\} $
and $\left\{ c,b\right\} $ does not change across them.}}\vspace{1mm} and agents $1$ and $3$ share the same ranking in
it, so that profile is determined by Remark \textcolor{blue}{\ref{rem:same param share ranking}.}
$\left(4,d\right)$ is determined by considering a profile where agent
$4$ ranks either $acbd$ or $acdb$, while preserving his binary
preference between $b$ and $d$ from the given profile.\footnote{\label{cxyz assume aP_4d}Since agent $4$ cannot receive $a$ under
any ranking of the form $cxyz$, we may assume $aP_{4}d$. This is
justified as follows: if $a$ and $d$ are not adjacent in $P_{4}$,
we may first swap $\left\{ a,b\right\} $; by SP, this does not affect
$\left(4,d\right)$. Once $a$ and $d$ become adjacent in $P_{4}$,
swapping them leaves $\left(4,d\right)$ unchanged by Remark \textcolor{blue}{\ref{efficiency swaps}}.} That profile's disagreement parameter is less than that of the original
profile.\footnote{This is because all the swaps considered in agent $4$'s ranking elevate
$a$, and each such swap reduces the disagreement parameter by $3$.} Since agents $1$ and $3$ share the same ranking, $\left(1,b\right)$
is then determined by complementarity.

\end{casenv}
\item When agent $3$ ranks $acdb$, the constraints require that $bP_{4}c$
and $bP_{4}d$. Thus, the remaining possibilities for $P_{4}$ are
of the form $bxyz$.\vspace{1mm}

\begin{casenv}[leftmargin=0pt]

\item[] {\centering{\itshape%
\begin{tabular}{|c|c|c|c|}
\hline 
\textcolor{brown}{a} & \textcolor{brown}{a} & \textcolor{green}{a} & \textcolor{green}{b}\tabularnewline
\hline 
\uline{b} & \uline{c} & \textcolor{green}{c} & \textcolor{green}{x}\tabularnewline
\hline 
c & b & \textcolor{green}{d} & \textcolor{green}{y}\tabularnewline
\hline 
\textcolor{brown}{d} & \textcolor{brown}{d} & \textcolor{red}{b} & \textcolor{green}{z}\tabularnewline
\hline 
\end{tabular}}\rlap{(\textcolor{blue}{\ref{acbd 1,2 a,d}})}\par}\vspace{1mm}In
profiles of this form, agents $3$ and $4$ are not supported: agent
$3$ because of the pair $\left\{ d,b\right\} $, and agent $4$ because
of the pair containing $a$ and the house immediately above it. It
then suffices to determine either $\left(1,b\right)$ or $\left(2,c\right)$.\vspace{1mm}
\begin{itemize}
\item \begin{minipage}[t]{0.20\linewidth}{\itshape%
\begin{tabular}{|c|c|c|c|}
\hline 
\textcolor{blue}{a} & \textcolor{blue}{a} & \textcolor{blue}{a} & b\tabularnewline
\hline 
\textbf{\textcolor{violet}{b}} & c & \textcolor{green}{c} & x\tabularnewline
\hline 
\textcolor{orange}{d} & b & \textcolor{violet}{d} & y\tabularnewline
\hline 
\textcolor{red}{c} & \textcolor{green}{d} & \textcolor{red}{b} & z\tabularnewline
\hline 
\end{tabular}}\end{minipage}\hfill\modifiedjustifiedminipage{0.78}{When $cP_{4}d$,
we use this profile to determine $\left(1,b\right)$. $\left(1,c\right)$,
$\left(3,b\right)$, and $\left(4,a\right)$ are determined by efficiency,
and consequently $a$ is determined by ETA. $\left(2,d\right)$ is
determined by considering the profile where agent $2$}\vspace{1mm} swaps $\left\{ c,b\right\} $,
and $\left(4,d\right)$ is determined by considering
the profile where agent $4$ ranks $acbd$.\footnote{Since agent $4$ cannot receive $a$ under any ranking of the form
$bxyz$, we may assume $aP_{4}d$ by Remark \textcolor{blue}{\ref{efficiency swaps}}.} Both profiles are determined by the lemma.\footnote{In the profile where agent $2$ swaps $\left\{ c,b\right\} $, note
that three agents prefer $a$ to every other house (agents $1$, $2$,
and $3$), three agents prefer $b$ to both $c$ and $d$ (agents
$1$, $2$, and $4$), and three agents prefer $c$ to $d$ (agents
$2$, $3$, and, by assumption, agent $4$).} $\left(3,c\right)$ is determined by considering the profile where
agent $3$ swaps $\left\{ d,b\right\} $. That profile's disagreement
parameter equals that of the original profile,\footnote{This is true in all relevant options, since agent $4$'s binary preference
over the pairs $\left\{ c,d\right\} $ and $\left\{ d,b\right\} $
does not change across them.} and agents $2$ and $3$ share the same ranking in it, so that profile
is determined by Remark \textcolor{blue}{\ref{rem:same param share ranking}.}
$\left(1,b\right)$ is then determined by complementarity.\vspace{1mm}
\item \begin{minipage}[t]{0.20\linewidth}{\itshape%
\begin{tabular}{|c|c|c|c|}
\hline 
\textcolor{blue}{a} & \textcolor{blue}{a} & \textcolor{blue}{a} & \textcolor{green}{b}\tabularnewline
\hline 
b & \textbf{\textcolor{violet}{c}} & c & x\tabularnewline
\hline 
c & \textcolor{orange}{d} & \textcolor{orange}{d} & y\tabularnewline
\hline 
\textcolor{green}{d} & \textcolor{red}{b} & \textcolor{red}{b} & z\tabularnewline
\hline 
\end{tabular}}\end{minipage}\hfill\modifiedjustifiedminipage{0.78}{When $dP_{4}c$,
we can use this profile to determine $\left(2,c\right)$. $\left(2,b\right)$,
$\left(3,b\right)$, $\left(4,a\right)$, and $\left(4,c\right)$
are determined by efficiency, and consequently $a$ is determined
by ETA. $\left(1,d\right)$ is determined by considering the profile
where agent}\vspace{1mm} $1$ swaps $\left\{ b,c\right\} $. That
profile's disagreement parameter equals that of the original profile,
and it is not supported.\footnote{In that profile, agents $2$ and $3$ are not supported because of
the pair $\left\{ d,b\right\} $, and agent $4$ is not supported
because of the pair containing $a$ and the house immediately above
it. Moreover, the disagreement parameter equals that of the original
profile in all options, since agent $4$'s binary preference over
the pairs $\left\{ b,d\right\} $ and $\left\{ b,c\right\} $ does
not change between them, as he ranks $b$ first.} $\left(4,b\right)$ is determined by considering the profile where
agent $4$ ranks $bacd$. The disagreement parameter of that profile
is less than that of the original one.\footnote{In order to reach that ranking, agent $4$ must swap the pair $\left\{ c,d\right\} $,
and that swap lowers the disagreement parameter by $3$. Agent $4$
may also perform additional swaps that elevate $a$, each of which
decreases it even more. Thus this transition yields a net reduction
of at least $3$, whereas the earlier $\left\{ b,d\right\} $ swap
in agent $2$'s ranking raised it by only $1$.} Since agents $2$ and $3$ share the same ranking, $\left(2,c\right)$
is then determined by complementarity.
\end{itemize}
\end{casenv}
\item When agent $3$ ranks $adxy$, the constraints require that $yP_{4}xP_{4}d$.
But then the agent $i\in\left\{ 1,2\right\} $ who ranks $axyd$ cannot
receive $y$,\footnote{Agent $i$ cannot receive $y$ because agent $3$ is the only other
agent who prefers $x$ to $y$, so agent $3$ must receive $x$ first.
For this to occur, either agent $4$ or agent $3-i$ must receive
$d$ beforehand, yet both of them prefer $y$ over $d$.} which contradicts the requirement imposed by that agent's relaxation
pair.
\item Otherwise, agent $3$ does not rank $a$ first. We first show that
every profile satisfying the constraints in this case must be one
in which agent $3$ ranks $b$ first and prefers $a$ to $c$, and
agent $4$ ranks $c$ first and prefers $a$ to $b$.

Since $R_{cd}=1$, by Remark \textcolor{blue}{\ref{rem:R_xd>0 and 3 prefers d}},
we may assume that agent $3$ does not rank $d$ first. If agent $3$
were to rank $c$ first, then the requirement $R_{bc}=1$ would imply
that agent $4$ cannot rank $c$ first, so he must rank $d$ first;
but then agent $1$ would be unable to receive $c$.\footnote{Agent $1$ cannot receive $c$ because agent $4$ is the only other
agent who prefers $b$ to $c$, so agent $4$ must receive $b$ first.
For this to occur, either agent $2$ or $3$ must receive $d$ beforehand,
yet both of them prefer $c$ over $d$.} This contradicts the requirement imposed by the relaxation pair $L_{1}=\left(b,c\right)$.
Thus, agent $3$ must rank $b$ first. This, in turn, implies that
the constraints require $aP_{4}b$ and $cP_{4}b$. In particular,
agent $4$ cannot rank $b$ first. He also cannot rank $d$ first,
since in that case agent $2$ would be unable to receive $b$.\footnote{To see this, apply the argument used in the previous footnote, and
exchange the roles of agents $1$ and $2$ and houses $b$ and $c$.} Thus, he must rank $c$ first. This in turn implies that the constraints
require $aP_{3}c$, so agent $3$ must rank $b$ first and prefer
$a$ to $c$, while agent $4$ must rank $c$ first and prefer $a$
to $b$. We now consider all remaining possibilities collectively.\vspace{1mm}

\begin{casenv}[leftmargin=0pt]

\item[] {\centering{\itshape%
\begin{tabular}{|c|c|c|c|}
\hline 
\textcolor{brown}{a} & \textcolor{brown}{a} & \textcolor{green}{b} & \textcolor{green}{c}\tabularnewline
\hline 
\uline{b} & c & \textcolor{green}{x} & \textcolor{green}{$\widetilde{x}$}\tabularnewline
\hline 
c & b & \textcolor{green}{y} & \textcolor{green}{$\widetilde{y}$}\tabularnewline
\hline 
\textcolor{brown}{d} & \textcolor{brown}{d} & \textcolor{green}{z} & \textcolor{green}{$\widetilde{z}$}\tabularnewline
\hline 
\end{tabular}}\rlap{(\textcolor{blue}{\ref{acbd 1,2 a,d}})}\par}\vspace{1mm}In
profiles of this form, both agents $3$ and $4$ are not supported
because of the pair containing $a$ and the house immediately above
it. It then suffices to determine $\left(1,b\right)$ in the following
profile.\vspace{1mm}

\begin{minipage}[t]{0.19\linewidth}{\itshape%
\begin{tabular}{|c|c|c|c|}
\hline 
\textcolor{blue}{a} & \textcolor{blue}{a} & b & c\tabularnewline
\hline 
\textbf{\textcolor{violet}{b}} & c & x & $\widetilde{x}$\tabularnewline
\hline 
\textcolor{orange}{d} & b & y & $\widetilde{y}$\tabularnewline
\hline 
\textcolor{red}{c} & \textcolor{green}{d} & z & $\widetilde{z}$\tabularnewline
\hline 
\end{tabular}}\end{minipage}\hfill\modifiedjustifiedminipage{0.79}{$\left(1,c\right)$,
$\left(3,a\right)$, and $\left(4,a\right)$ are determined by efficiency,
and consequently $a$ is determined by ETA. It then suffices to determine
$\left(i,d\right)$ for $i=2,3,4$. $\left(2,d\right)$ is determined
by considering the profile where agent $2$ swaps $\left\{ c,b\right\} $.
If $cP_{3}d$}\vspace{1mm} (i.e., the ranking of agent $3$ is $bacd$),
then that profile is determined by the lemma.\footnote{This is because agents $1$, $2$, and $4$ prefer $a$ to $b$; agents
$1$, $2$, and $3$ prefer $a$ and $b$ to both $c$ and $d$; and
agents $2$, $3$, and $4$ prefer $c$ to $d$.} Otherwise, that profile's disagreement parameter equals that of the
original profile, and it is not supported.\footnote{In this case, the swap of the pair $\left\{ c,d\right\} $ in agent
$1$'s ranking increases the disagreement parameter by only $1$,
so the subsequent swap of $\left\{ c,b\right\} $ in agent $2$'s
ranking offsets that increase entirely. In the resulting profile,
agent $1$ is not supported because of the pair $\left\{ d,c\right\} $,
and agents $3$ and $4$ are not supported because of the pair containing
$a$ and the house immediately above it.} $\left(3,d\right)$ is determined as follows: if $cP_{3}d$, then
it is determined by considering the profile where agent $3$ ranks
$acbd$. That profile's disagreement parameter is less than that of
the original profile.\footnote{The transition from $bacd$ to $acbd$ in agent $3$'s ranking reduces
the disagreement parameter by $4$, whereas the swap of the pair $\left\{ c,d\right\} $
in agent $1$'s ranking increased it by only $3$.} Otherwise, $\left(3,c\right)$ is determined by efficiency, and $\left(3,b\right)$
is determined by considering the profile where agent $3$ ranks $bacd$.
That profile's disagreement parameter is at most that of the original
profile, and agent $3$ is not supported there.\footnote{In this case, agent $3$'s transition reduces the parameter by at
least the amount by which swapping the pair $\left\{ c,d\right\} $
in agent $1$'s ranking increased it, since agent $3$'s transition
swaps $\left\{ c,d\right\} $ in the opposite direction, possibly
after also swapping the pair $\left\{ d,a\right\} $, which decreases
the parameter even further. In the resulting profile, agent $3$ is
not supported because of the pair $\left\{ b,a\right\} .$} $\left(3,d\right)$ is then determined by complementarity. Finally,
$\left(4,d\right)$ is determined by considering the profile where
agent $4$ ranks either $acbd$ or $acdb$, while preserving his binary
preference between $b$ and $d$ from the given profile.\footnote{A similar argument to the one used in footnote \textcolor{blue}{\ref{cxyz assume aP_4d}}
applies here.} That profile's disagreement parameter is at most that of the original
profile, and it is not supported.\footnote{The swap of the pair $\left\{ c,d\right\} $ in agent $1$'s ranking
increases the disagreement parameter by at most $3$, while the swap
of the pair $\left\{ c,a\right\} $ in agent $4$'s ranking decreases
it by exactly $3$ (since $aP_{3}c$). Any additional swaps that may
occur in agent $4$'s transition elevate $a$, and these also decrease
the parameter, as agents $1$ and $2$ rank $a$ first. In the resulting
profile, agent $1$ is not supported because of the pair $\left\{ d,c\right\} $,
and agent $3$ is not supported because of the pair containing $a$
and the house immediately above it. Moreover, if $dP_{4}b$, then
agent $4$ is not supported because of that pair, and otherwise agents
$2$ and $4$ share the same ranking.}

\end{casenv}
\end{casenv}
\item When agent $2$ ranks $acdb$, the constraints from agent $1$ are
$\widetilde{R}_{ab}^{1}=1$, $\widetilde{R}_{bc}^{1}=2$, and $\widetilde{R}_{cd}^{1}=1$,
and the constraints from agent $2$ are $\widetilde{R}_{ac}^{2}=1$,
$\widetilde{R}_{cd}^{2}=1$, and $\widetilde{R}_{db}^{2}=2$. By Remark
\textcolor{blue}{\ref{x,y,z}} applied with $i=2$ and $\left(x,y,z\right)=\left(d,b,c\right)$,
we must have either $L_{2}=\left(d,b\right)$, or $L_{1}=\left(b,c\right)$,
or $L_{1}=L_{2}=\left(c,d\right)$.
\begin{casenv}
\item When $L_{2}=\left(d,b\right)$, note that agents $3$ and $4$ cannot
rank $b$ first, because in that case agent $2$ would be unable to
receive it. We split according to the value of $L_{1}$.

If $L_{1}=\varnothing$ or $L_{1}=\left(a,b\right)$, then $R_{bc}=2$
and $R_{cd}=1$. These constraints imply that one of the agents $j\in\left\{ 3,4\right\} $
must satisfy $bP_{j}cP_{j}d$. Since agent $1$ ranks $abcd$ and
all four agents have distinct rankings, this forces agent $j$ to
rank $b$ first, but then, as noted earlier, agent $2$ would not
be supported.

Otherwise, by Remark \textcolor{blue}{\ref{both adjacent, same order}},
$L_{1}=\left(b,c\right)$. Then, the constraints imposed on the rankings
of agents $3$ and $4$ are 
\[
R_{ab}=R_{bc}=R_{cd}=R_{ac}=R_{db}=1.
\]
Note that by Remark \ref{rem:R_xd>0 and 3 prefers d}, we may assume that agent $3$ does not rank
$d$ first.
\begin{itemize}
\item When agent $3$ ranks $abdc$, the constraints require that $cP_{4}dP_{4}b$,
but in that case agent $2$ cannot receive $b$.\footnote{In that case, agent $2$ cannot receive $b$ because for that to occur,
agent $4$ would have to receive $d$ beforehand (as he is the only
other agent who ranks $d$ over $b$). For agent $4$ to receive $d$,
either agent $1$ or agent $3$ would need to receive $c$ first,
yet both of them prefer $b$ over $c$.}
\item When agent $3$ ranks $acbd$, the constraints require that $dP_{4}bP_{4}c$,
but in that case agent $1$ cannot receive $c$.\footnote{\label{fn:43}This is because agent $4$ would have to receive $b$ beforehand,
and for that to occur either agent $2$ or agent $3$ would need to
receive $d$ first, yet both prefer $c$ over $d$.}
\item When agent $3$ ranks $adxy$, the constraints require that $cP_{4}d$.
In this case, agent $4$ cannot rank $a$ first because agent $3$'s
ranking precedes his, and as noted earlier he cannot rank $b$ first
either. Thus, since $cP_{4}d$, he must rank $c$ first, but then
agent $1$ cannot receive $c$.\footnote{This is because agent $3$ would need to prefer $b$ over $c$ and
receive $b$, and for that to occur either agent $2$ or agent $4$
would need to receive $d$ first, yet both prefer $c$ over $d$.}
\item As noted earlier, when agent $3$ ranks $b$ first, agent $2$ cannot
receive it.
\item Otherwise, agent $3$ ranks $c$ first. Then the constraints require
that $bP_{4}c$, so agent $4$ cannot rank $c$ first. Since agent
$3$'s ranking precedes agent $4$'s, the latter must rank $d$ first.
In that case, agent $1$ cannot receive $c$.\footref{fn:43}
\end{itemize}
\item When $L_{1}=\left(b,c\right)$, it remains to consider the cases where
$L_{2}=\varnothing$ and $L_{2}=\left(a,c\right)$. In these cases,
we have $R_{db}=2$, $R_{cd}=1$, and $R_{bc}=1$. These constraints
imply that exactly one agent $i\in\left\{ 3,4\right\} $ must rank
$b$ over $c$, and that agent must rank $d$ over $b$. Consequently,
agent $1$ cannot receive $c$.\footnote{For that to occur, agent $i$ would need to receive $b$ beforehand,
and for this to happen, one of the two other agents would need to
receive $d$ first, yet both prefer $c$ over $d$.}
\item When $L_{1}=L_{2}=\left(c,d\right)$, we have $R_{db}=R_{bc}=2$.
Since agents $3$ and $4$ have different rankings, at least one of
them must rank $d$ first. In that situation, agent $1$ cannot receive
$d$.
\end{casenv}
\item When agent $2$ ranks $adbc$, we may rename the houses so that his
ranking becomes $abcd$. Under this renaming, agent $1$'s ranking
becomes $acdb$. Swapping the names of agents $1$ and $2$ then reduces
this case to the previous one.\footnote{We may also swap the names of agents $3$ and $4$ so that the ranking
of agent $3$ precedes that of agent $4$.}
\item When agent $2$ ranks $adcb$, agents $1$ and $2$ disagree on both
adjacent pairs $\left\{ b,c\right\} $ and $\left\{ c,d\right\} $,
and by Remark \textcolor{blue}{\ref{both adjacent, reversed order}},
agents $3$ and $4$ cannot satisfy the constraints.
\item When agent $2$ ranks $bacd$, Remark \textcolor{blue}{\ref{both adjacent, reversed order}}
implies that $L_{1}=\left(a,b\right)$ and $L_{2}=\left(b,a\right)$.
Under these conditions, the constraints become
\[
R_{ab}=R_{bc}=R_{cd}=R_{ba}=R_{ac}=1.
\]
Note that by Remarks \textcolor{blue}{\ref{rem:R_xd>0 and 3 prefers d},
}\ref{rem:R_xc>0, 2-b and 3-c}, and \textcolor{blue}{\ref{rem:2 and 3 prefer b, 4 not, and b in L_1}},
we may assume that agent $3$ ranks $a$ first.

\begin{rem}
\label{bacd agent 3 ranks a first so 4 ranks b first}In all subsequent
subcases where agent $3$ ranks $a$ first, agent $4$ must rank $b$
first. Otherwise, agent $2$ would be unable to receive $a$.
\end{rem}

\begin{casenv}
\item When agent $3$ ranks $abdc$, the constraints require that $cP_{4}d$.
Thus, by Remark \textcolor{blue}{\ref{bacd agent 3 ranks a first so 4 ranks b first}},
the remaining possibilities for agent $4$'s ranking are of the form
$bcxy$ (because agent $2$ ranks $bacd$).\vspace{1mm}

\begin{casenv}[leftmargin=0pt]

\item[] {\centering{\itshape%
\begin{tabular}{|c|c|c|c|}
\hline 
\uline{a} & b & \textcolor{green}{a} & \textcolor{green}{b}\tabularnewline
\hline 
b & a & \textcolor{green}{b} & \textcolor{green}{c}\tabularnewline
\hline 
\textcolor{brown}{c} & \textcolor{brown}{c} & \textcolor{green}{d} & \textcolor{green}{x}\tabularnewline
\hline 
\textcolor{brown}{d} & \textcolor{brown}{d} & \textcolor{red}{c} & \textcolor{green}{y}\tabularnewline
\hline 
\end{tabular}}\rlap{(\textcolor{blue}{\ref{bacd 1,2 c,d}})}\par}\vspace{1mm}In
profiles of this form, agents $3$ and $4$ are not supported: agent
$3$ because of the pair $\left\{ d,c\right\} $, and agent $4$ because
of the pair containing $a$ and the house immediately above it. It
then suffices to determine $\left(1,a\right)$ in the following profile.\vspace{1mm}

\begin{minipage}[t]{0.19\linewidth}{\itshape%
\begin{tabular}{|c|c|c|c|}
\hline 
\textbf{\textcolor{violet}{a}} & b & a & \textcolor{green}{b}\tabularnewline
\hline 
\textcolor{orange}{c} & a & b & \textcolor{violet}{c}\tabularnewline
\hline 
\textcolor{red}{b} & \textcolor{green}{c} & \textcolor{green}{d} & x\tabularnewline
\hline 
\textcolor{green}{d} & \textcolor{green}{d} & \textcolor{red}{c} & y\tabularnewline
\hline 
\end{tabular}}\end{minipage}\hfill\modifiedjustifiedminipage{0.79}{$\left(1,b\right)$,
$\left(3,c\right)$, and $\left(4,a\right)$ are determined by efficiency.
Since $\left\{ a,b\right\} $ is the only pair on which the agents
do not have near unanimous agreement, $\left(2,c\right)$, $\left(2,d\right)$,
and $\left(3,d\right)$ are determined by considering the profiles
where agent $i\in\left\{ 2,3\right\} $}\vspace{1mm} swaps $\left\{ a,b\right\} $
and applying the lemma. Likewise, $\left(1,d\right)$ is determined
by considering the profile where agent $1$ ranks $bacd$.\footnote{Note that agents $2$, $3$, and $4$ all prefer $b$ to $c$. Therefore,
agent $1$ may swap $\left\{ b,c\right\} $ without causing that pair
to lose its near unanimous agreement among the agents.} $\left(4,b\right)$ is determined by considering the profile where
agent $4$ ranks $bacd$. That profile's disagreement parameter is
at most that of the original profile,\footnote{Swapping the pair $\left\{ b,c\right\} $ in agent $1$'s ranking
increased the disagreement parameter by $3$. In agent $4$'s transition,
swapping the pair $\left\{ c,a\right\} $ decreases it by $3$, and
an additional swap of $\left\{ a,d\right\} $, if it occurs, decreases
it even further.} and agents $2$ and $4$ share the same ranking in it, so that profile
is determined by Remark \textcolor{blue}{\ref{rem:same param share ranking}.}
Finally, $\left(1,a\right)$ is determined by complementarity.

\end{casenv}
\item When agent $3$ ranks $acxy$, Remark \textcolor{blue}{\ref{bacd agent 3 ranks a first so 4 ranks b first}}
implies that agent $4$ must rank $b$ first. His ranking cannot be
$bacd$, since this is agent $2$'s ranking, and in this situation
the constraints do not impose any additional restrictions on the remaining
options.\vspace{1mm}

\begin{casenv}[leftmargin=0pt]

\item[] {\centering{\itshape%
\begin{tabular}{|c|c|c|c|}
\hline 
\uline{a} & \uline{b} & \textcolor{green}{a} & \textcolor{green}{b}\tabularnewline
\hline 
b & a & \textcolor{green}{c} & \textcolor{green}{$\widetilde{x}$}\tabularnewline
\hline 
\textcolor{brown}{c} & \textcolor{brown}{c} & \textcolor{green}{x} & \textcolor{green}{$\widetilde{y}$}\tabularnewline
\hline 
\textcolor{brown}{d} & \textcolor{brown}{d} & \textcolor{green}{y} & \textcolor{green}{$\widetilde{z}$}\tabularnewline
\hline 
\end{tabular}}\rlap{(\textcolor{blue}{\ref{bacd 1,2 c,d}})}\par}\vspace{1mm}In
profiles of this form, agents $3$ and $4$ are not supported: agent
$3$ because of the pair containing $b$ and the house immediately
above it, and agent $4$ because of the pair containing $a$ and the
house immediately above it, except in the ranking $badc$, where he
is not supported because of the pair $\left\{ d,c\right\} $. It then
suffices to determine either $\left(1,a\right)$ or $\left(2,b\right)$.\vspace{1mm}
\begin{itemize}
\item \begin{minipage}[t]{0.20\linewidth}{\itshape%
\begin{tabular}{|c|c|c|c|}
\hline 
\textbf{\textcolor{violet}{a}} & b & a & b\tabularnewline
\hline 
\textcolor{orange}{d} & a & c & c\tabularnewline
\hline 
\textcolor{red}{b} & c & b & $\widetilde{y}$\tabularnewline
\hline 
\textcolor{red}{c} & \textcolor{green}{d} & \textcolor{green}{d} & $\widetilde{z}$\tabularnewline
\hline 
\end{tabular}}\end{minipage}\hfill\modifiedjustifiedminipage{0.78}{When $cP_{4}d$
and $bP_{3}d$, we use this profile to determine $\left(1,a\right)$.\footnote{\label{fn:38}Since agent $4$'s ranking is not $bacd$, he must rank
$c$ in the second position.} $\left(1,b\right)$, $\left(1,c\right)$, and $\left(4,a\right)$
are determined by efficiency. $\left(i,d\right)$ for $i=2,4$ is
determined by considering the profile where agent $i$ ranks $abcd$,\footnote{If $dP_{4}a$, then by Remark \textcolor{blue}{\ref{efficiency swaps}}
the transition to $abcd$ is justified.} and $\left(3,d\right)$ is}\vspace{1mm} determined by considering
the profile where agent $3$ ranks $bacd$. All these profiles are
determined by the lemma.\footnote{The swap of $\left\{ b,c\right\} $ in agent $3$'s ranking and of
$\left\{ a,c\right\} $ (and possibly $\left\{ a,d\right\} $) in
agent $4$'s ranking do not cause any of these pairs to lose their
near-unanimous agreement.} $\left(1,a\right)$ is then determined by complementarity.\vspace{1mm}
\item \begin{minipage}[t]{0.20\linewidth}{\itshape%
\begin{tabular}{|c|c|c|c|}
\hline 
\textbf{a} & \textcolor{blue}{b} & a & \textcolor{blue}{b}\tabularnewline
\hline 
c & a & c & c\tabularnewline
\hline 
d & \uline{c} & d & $\widetilde{y}$\tabularnewline
\hline 
\textcolor{red}{b} & \uline{d} & \textcolor{red}{b} & $\widetilde{z}$\tabularnewline
\hline 
\end{tabular}}\end{minipage}\hfill\modifiedjustifiedminipage{0.78}{When $cP_{4}d$
and $dP_{3}b$, we use this profile to determine $\left(1,a\right)$.\footref{fn:38}
$\left(1,b\right)$, $\left(3,b\right)$, and $\left(4,a\right)$
are determined by efficiency, and consequently $b$ is determined
by ETA. It then suffices to determine $\left(2,c\right)$ and $\left(2,d\right)$
in the}\vspace{1mm} following profile.\vspace{1mm}
\begin{itemize}
\item \begin{minipage}[t]{0.205\linewidth}{\itshape%
\begin{tabular}{|c|c|c|c|}
\hline 
\textcolor{blue}{a} & \textcolor{blue}{a} & \textcolor{blue}{a} & \textcolor{green}{b}\tabularnewline
\hline 
c & \textcolor{orange}{b} & c & c\tabularnewline
\hline 
d & \textbf{\textcolor{violet}{c}} & d & $\widetilde{y}$\tabularnewline
\hline 
\textcolor{red}{b} & \textbf{\textcolor{green}{d}} & \textcolor{red}{b} & $\widetilde{z}$\tabularnewline
\hline 
\end{tabular}}\end{minipage}\hfill\modifiedjustifiedminipage{0.775}{$\left(1,b\right)$,
$\left(3,b\right)$, and $\left(4,a\right)$ are determined by efficiency,
and consequently $a$ is determined by ETA. $\left(2,d\right)$ is
determined by considering the profile where agent $2$ swaps $\left\{ b,c\right\} $.
That profile's disagreement parameter}\vspace{1mm} equals that of
the original profile,\footnote{Swapping $\left\{ b,c\right\} $ and $\left\{ b,d\right\} $ in agent
$1$'s ranking each increases the disagreement parameter by $1$,
whereas swapping $\left\{ b,a\right\} $ and $\left\{ b,c\right\} $
in agent $2$'s ranking each decreases it by $1$.} and agents $1$ and $3$ share the same ranking in it, so that profile
is determined by Remark \textcolor{blue}{\ref{rem:same param share ranking}.}
$\left(4,b\right)$ is determined by considering the profile where
agent $4$ ranks $bacd$. That profile's disagreement parameter is
less than that of the original profile.\footnote{The swap of $\left\{ c,a\right\} $ in agent $4$'s ranking decreases
the disagreement parameter by $3$, and if $dP_{4}a$, then the swap
of $\left\{ d,a\right\} $ decreases it even further.} Finally, $\left(2,c\right)$ is determined by complementarity.\vspace{1mm}
\end{itemize}
\item \begin{minipage}[t]{0.20\linewidth}{\itshape%
\begin{tabular}{|c|c|c|c|}
\hline 
a & \textbf{\textcolor{violet}{b}} & \textcolor{green}{a} & b\tabularnewline
\hline 
b & \textcolor{orange}{c} & \textcolor{violet}{c} & $\widetilde{x}$\tabularnewline
\hline 
\textcolor{green}{c} & \textcolor{red}{a} & x & $\widetilde{y}$\tabularnewline
\hline 
d & \textcolor{green}{d} & y & $\widetilde{z}$\tabularnewline
\hline 
\end{tabular}}\end{minipage}\hfill\modifiedjustifiedminipage{0.78}{When $dP_{4}c$,
we use this profile to determine $\left(2,b\right)$. $\left(2,a\right)$,
$\left(3,b\right)$, and $\left(4,c\right)$ are determined by efficiency.
$\left(1,c\right)$ is determined by considering the profile where
agent $1$ swaps $\left\{ a,b\right\} $. If $aP_{4}c$, that profile
is determined by the}\vspace{1mm} lemma; otherwise, its disagreement
parameter equals that of the original profile, and it is not supported.\footnote{When $cP_{4}a$, the swap of $\left\{ a,c\right\} $ in agent $2$'s
ranking increases the parameter by only $1$, while the swap $\left\{ a,b\right\} $
in agent $1$'s ranking decreases it by $1$. In the resulting profile,
agents $2$ and $4$ are not supported because of $\left\{ c,a\right\} $
and $\left\{ d,c\right\} $, respectively, and agent $3$ is not supported
because of the pair containing $b$ and the house immediately above
it.} $\left(2,d\right)$ is determined by considering the profile where
agent $2$ ranks $abcd$ and applying the lemma. $\left(3,a\right)$
and $\left(3,d\right)$ are determined by considering the profile
where agent $3$ ranks $abcd$.\footnote{We use the fact that agent $3$ cannot receive $b$ even when $bP_{3}d$,
together with Remark \textcolor{blue}{\ref{efficiency swaps}} to
justify that the same transition also determines $\left(3,d\right)$
when $dP_{3}b$.} That profile's disagreement parameter is at most that of the original
profile,\footnote{Swapping $\left\{ c,b\right\} $ in agent $3$'s ranking decreases
the disagreement parameter by $3$, which is at least as large as
the increase produced by swapping $\left\{ a,c\right\} $ in agent
$2$'s ranking. An additional swap of $\left\{ d,b\right\} $ in agent
$3$'s ranking, if it occurs, may reduce the parameter even further.} and agents $1$ and $3$ share the same ranking in it, so that profile
is determined by Remark \textcolor{blue}{\ref{rem:same param share ranking}.}
Consequently, $\left(2,b\right)$ is determined by complementarity.
\end{itemize}
\end{casenv}
\item When agent $3$ ranks $adxy$, the constraints require that $cP_{4}d$.
Thus, by Remark \textcolor{blue}{\ref{bacd agent 3 ranks a first so 4 ranks b first}},
the remaining possibilities for agent $4$'s ranking are of the form
$bc\widetilde{x}\widetilde{y}$, since agent $2$ ranks $bacd$.\vspace{1mm}

\begin{casenv}[leftmargin=0pt]

\item[] {\centering{\itshape%
\begin{tabular}{|c|c|c|c|}
\hline 
\uline{a} & b & \textcolor{green}{a} & \textcolor{green}{b}\tabularnewline
\hline 
b & a & \textcolor{green}{d} & \textcolor{green}{c}\tabularnewline
\hline 
\textcolor{brown}{c} & \textcolor{brown}{c} & \textcolor{green}{x} & \textcolor{green}{$\widetilde{x}$}\tabularnewline
\hline 
\textcolor{brown}{d} & \textcolor{brown}{d} & \textcolor{green}{y} & \textcolor{green}{$\widetilde{y}$}\tabularnewline
\hline 
\end{tabular}}\rlap{(\textcolor{blue}{\ref{bacd 1,2 c,d}})}\par}\vspace{1mm}In
profiles of this form, agent $3$ is not supported because of the
pair containing $b$ and the house immediately above it, and similarly
agent $4$ when the same reasoning is applied with $a$ in place of
$b$. It then suffices to determine $\left(1,a\right)$ in the following
profile.\vspace{1mm}

\begin{minipage}[t]{0.19\linewidth}{\itshape%
\begin{tabular}{|c|c|c|c|}
\hline 
\textbf{\textcolor{orange}{a}} & \textcolor{blue}{b} & a & \textcolor{blue}{b}\tabularnewline
\hline 
c & \textcolor{violet}{a} & d & c\tabularnewline
\hline 
\textcolor{red}{b} & \textcolor{green}{c} & x & $\widetilde{x}$\tabularnewline
\hline 
d & \textcolor{green}{d} & y & $\widetilde{y}$\tabularnewline
\hline 
\end{tabular}}\end{minipage}\hfill\modifiedjustifiedminipage{0.79}{$\left(1,b\right)$,
$\left(3,b\right)$, and $\left(4,a\right)$ are determined by efficiency,
and consequently $b$ is determined by ETA. $\left(2,c\right)$ and
$\left(2,d\right)$ are determined by considering the profile where
agent $2$ swaps $\left\{ b,a\right\} $. If $bP_{3}c$, that profile
is determined by the lemma;}\vspace{1mm} otherwise, its disagreement
parameter equals that of the original profile and it is not supported.\footnote{In this case, the swap of $\left\{ b,c\right\} $ in agent $1$'s
ranking increased the disagreement parameter by only $1$, which is
exactly offset by the decrease induced by the swap of $\left\{ b,a\right\} $
in agent $2$'s ranking. In the resulting profile, agents $1$ and
$3$ are not supported because of the pair $\left\{ c,b\right\} $,
and agent $4$ is not supported because of the pair containing $a$
and the house immediately above it.} $\left(3,a\right)$ is determined by considering the profile where
agent $3$ ranks $abcd$. That profile's disagreement parameter is
less than that of the original profile.\footnote{Here, the swaps of $\left\{ d,b\right\} $ and $\left\{ d,c\right\} $
in agent $3$'s ranking reduce the parameter by $3$ each, and the
increment caused by swapping $\left\{ c,b\right\} $ in agent $1$'s
ranking is less than the total reduction of these swaps.} $\left(1,a\right)$ is then determined by complementarity.

\end{casenv}
\end{casenv}
\item When agent $2$ ranks $badc$, agents $1$ and $2$ disagree on both
adjacent pairs $\left\{ a,b\right\} $ and $\left\{ c,d\right\} $,
and by Remark \textcolor{blue}{\ref{both adjacent, reversed order}},
agents $3$ and $4$ cannot satisfy the constraints.
\item When agent $2$ ranks $bcad$, by Remark \textcolor{blue}{\ref{x,y,z}},
we must have $L_{1}=\left(a,b\right)$, or $L_{2}=\left(c,a\right)$,
or $L_{1}=L_{2}=\left(b,c\right)$. Note that in the first two cases,
we have $R_{cd}=R_{bc}=1$, so by Remarks \textcolor{blue}{\ref{rem:R_xd>0 and 3 prefers d}}
and \ref{rem:R_xc>0, 2-b and 3-c} we may assume that agent $3$ ranks
either $a$ or $b$ first.
\begin{casenv}
\item When $L_{1}=\left(a,b\right)$, by Remark \textcolor{blue}{\ref{rem:2 and 3 prefer b, 4 not, and b in L_1}}
we may assume that agent $3$ does not rank $b$ first. Hence, agent
$3$ ranks $a$ first, then agent $2$ cannot receive $a$. Hence,
we must have $L_{2}\neq\left(c,a\right)$, but then $R_{ca}=2$, which
requires $cP_{3}a$, a contradiction.
\item When $L_{2}=\left(c,a\right)$, we may assume $L_{1}=\varnothing$
or $L_{1}=\left(c,d\right)$, since $L_{1}=\left(b,c\right)$ is redundant
by Remark \textcolor{blue}{\ref{both adjacent, same order}}, and
the case $L_{1}=\left(a,b\right)$ was already addressed. If agent
$3$ ranks $a$ first, then agent $2$ cannot receive $a$. Otherwise,
agent $3$ ranks $b$ first, but this violates the constraint $R_{ab}=2$.
\item When $L_{1}=L_{2}=\left(b,c\right)$, the constraints $R_{ca}=R_{ab}=2$
imply that both agents $3$ and $4$ rank $c$ above $a$ and $a$
above $b$. Hence, agent $3$ ranks $c$ first, and agent $4$ ranks
either $c$ or $d$ first, but in either case agent $1$ cannot receive
$c$.
\end{casenv}
\item When agent $2$ ranks $bcda$:
\begin{casenv}
\item When $L_{1}=L_{2}=\varnothing$, the constraints fix the rankings
of agents $3$ and $4$ to be $cdab$ and $dabc$, respectively, which
yields a degenerate profile.
\item When $L_{2}=\left(d,a\right)$, agent $2$ cannot receive $a$.
\item When $L_{1}=\left(a,b\right)$, we have already handled the case $L_{2}=\left(d,a\right)$,
and the cases $L_{2}=\left(b,c\right)$ and $L_{2}=\left(c,d\right)$
are redundant by Remark \textcolor{blue}{\ref{both adjacent, same order}},
so we may assume $L_{2}=\varnothing$. Moreover, since $R_{bc}=R_{cd}=1$,
by Remarks \textcolor{blue}{\ref{rem:R_xd>0 and 3 prefers d}}, \ref{rem:R_xc>0, 2-b and 3-c},
and \textcolor{blue}{\ref{rem:2 and 3 prefer b, 4 not, and b in L_1},}
we may assume that agent $3$ ranks $a$ first, but this violates
the constraint $R_{da}=2$.
\item When $L_{1}=L_{2}=\left(b,c\right)$, the constraints determine agent
$3$'s ranking to be $cdab$, and agent $4$ must rank $d$ first,
yielding a degenerate profile.
\item When $L_{1}=L_{2}=\left(c,d\right)$, the constraints require that
either agent $3$ or agent $4$ rank $d$ first, but then agent $1$
cannot receive $d$.
\end{casenv}
\item When agent $2$ ranks $bdac$, note that $\widetilde{R}_{da}^{2}=\widetilde{R}_{ab}^{1}=2$
and $\widetilde{R}_{bd}^{2}=1$. By Remark \textcolor{blue}{\ref{x,y,z}},
this implies that $L_{1}=\left(a,b\right)$, or $L_{2}=\left(d,a\right)$,
or $L_{2}=\left(b,d\right)$. Similarly, since $\widetilde{R}_{cd}^{1}=\widetilde{R}_{da}^{2}=2$
and $\widetilde{R}_{ac}^{2}=1$, Remark \textcolor{blue}{\ref{x,y,z}}
also implies that $L_{1}=\left(c,d\right)$, or $L_{2}=\left(d,a\right)$,
or $L_{2}=\left(a,c\right)$.
\begin{casenv}
\item When $L_{1}=\varnothing$ or $L_{1}=\left(b,c\right)$, the arguments
above imply that we must have $L_{2}=\left(d,a\right)$. However,
the constraints $R_{ab}=R_{cd}=2$ and $R_{ac}=1$ imply that either
agent $3$ or agent $4$ ranks $a$ first. In that case, agent $2$
cannot receive $a$.
\item When $L_{1}=\left(a,b\right)$, the arguments above imply that $L_{2}=\left(d,a\right)$
or $L_{2}=\left(a,c\right)$. Moreover, since $R_{bc}=1$ and $R_{cd}=2$,
by Remarks \textcolor{blue}{\ref{rem:R_xd>0 and 3 prefers d}, }\ref{rem:R_xc>0, 2-b and 3-c},
and \textcolor{blue}{\ref{rem:2 and 3 prefer b, 4 not, and b in L_1}},
we may assume that agent $3$ ranks $a$ first, but then agent $2$
cannot receive $a$.
\item When $L_{1}=\left(c,d\right)$ and $L_{2}=\left(b,d\right)$, the
constraints fix the rankings of agents $3$ and $4$ to be $cdab$
and $dabc$, respectively. This yields a degenerate profile.
\item Otherwise, $L_{1}=\left(c,d\right)$ and $L_{2}=\left(d,a\right)$.
Since $R_{bc}=R_{cd}=1$, by Remarks \textcolor{blue}{\ref{rem:R_xd>0 and 3 prefers d}}
and \ref{rem:R_xc>0, 2-b and 3-c}, we may assume that agent $3$
ranks either $a$ or $b$ first. If agent $3$ ranks $a$ first, agent
$2$ cannot receive $a$, and if he ranks $b$ first, this contradicts
$R_{ab}=2$.
\end{casenv}
\item When agent $2$ ranks $bdca$, Remark \textcolor{blue}{\ref{both adjacent, reversed order}}
implies that we must have $L_{1}=\left(c,d\right)$ and $L_{2}=\left(d,c\right)$.
However, even after these relaxations, the constraints still cannot
be satisfied, since $R_{ca}=R_{ab}=2$ and $R_{bc}=1$.
\item When agent $2$ ranks $caxy$, we may rename the houses so that his
ranking becomes $abcd$. Under this renaming, agent $1$'s ranking
becomes one that ranks $b$ first. Swapping the names of agents $1$
and $2$ then reduces this case to a previously analyzed one.
\item When agent $2$ ranks $cbad$, agents $1$ and $2$ disagree on both
adjacent pairs $\left\{ a,b\right\} $ and $\left\{ b,c\right\} $.
By Remark \textcolor{blue}{\ref{both adjacent, reversed order}},
agents $3$ and $4$ cannot satisfy the constraints.
\item When agent $2$ ranks $cbda$, Remark \textcolor{blue}{\ref{both adjacent, reversed order}}
implies that we must have $L_{1}=\left(b,c\right)$ and $L_{2}=\left(c,b\right)$.
However, even after these relaxations, the constraints still cannot
be satisfied, since $R_{da}=R_{ab}=2$ and $R_{bd}=1$.
\item When agent $2$ ranks $cdab$:
\begin{casenv}
\item When $L_{1}=L_{2}=\varnothing$, the constraints fix the rankings
of agents $3$ and $4$ to be $bcda$ and $dabc$, respectively, which
yields a degenerate profile.
\item When $L_{1}=\varnothing$ and $L_{2}=\left(d,a\right)$, since $R_{cd}=1$,
by Remark \textcolor{blue}{\ref{rem:R_xd>0 and 3 prefers d}}, we
may assume that agent $3$ does not rank $d$ first.\textcolor{blue}{{}
}If agent $3$ ranks $a$ first, then agent $2$ cannot receive $a$.
If agent $3$ ranks $b$ first, the constraints $R_{ab}=1$ and $R_{bc}=2$
imply that agent $4$ must rank $d$ first, and we obtain a degenerate
profile. Otherwise, agent $3$ ranks $c$ first, but then the constraint
$R_{bc}=2$ is not satisfied.
\item When $L_{1}=L_{2}=\left(a,b\right)$, the constraints $R_{bc}=2$,
$R_{cd}=1$, and $R_{da}=2$ force one of agents $3$ or $4$ to rank
$b$ first, but then agent $2$ cannot receive $b$.
\item When $L_{1}=L_{2}=\left(c,d\right)$, the constraints force one of
agents $3$ or $4$ to rank $d$ first, but then agent $1$ cannot
receive $d$.
\item When $L_{1}=\left(b,c\right)$, we may assume that $L_{2}=\varnothing$
or $L_{2}=\left(d,a\right)$. Since $R_{cd}=1$, by Remark \textcolor{blue}{\ref{rem:R_xd>0 and 3 prefers d}},
we may assume that agent $3$ does not rank $d$ first. If agent $3$
ranks $a$ first, agent $2$ cannot receive $a$. Thus we may assume
$L_{2}\neq\left(d,a\right)$. But then $R_{da}=2$ implies $dP_{3}a$,
a contradiction. If agent $3$ ranks $b$ first, then the constraint
$R_{ab}=1$ implies $aP_{4}b$, so agent $4$ must rank either $c$
or $d$ first. If he ranks $c$ first, agent $1$ cannot receive $c$;
if he ranks $d$ first, we obtain a degenerate profile. Otherwise,
agent $3$ ranks $c$ first, in which case agent $1$ cannot receive
$c$.
\end{casenv}
\item When agent $2$ ranks $cdba$, Remark \textcolor{blue}{\ref{both adjacent, reversed order}}
requires the relaxation $L_{2}=\left(b,a\right)$. However, this relaxation
is admissible only when agent $2$ can receive $a$, and in this case
agent $2$ cannot receive $a$.
\item When agent $2$ ranks $d$ first and does not rank $a$ last, we may
rename the houses so that his ranking becomes $abcd$. Under this
renaming, agent $1$'s ranking becomes one that ranks either $b$
or $c$ first. Swapping the names of agents $1$ and $2$ then reduces
this case to a previously analyzed one.
\item When agent $2$ ranks $dbca$:
\begin{casenv}
\item When $L_{i}=\varnothing$ for some $i\in\left\{ 1,2\right\} $, the
constraints force one of agents $3$ or $4$ to share the same ranking
with agent $i$.
\item When $L_{1}=\left(a,b\right)$ and $L_{2}=\left(d,b\right)$, since
$R_{cd}=R_{ca}=2$, we may assume that agent $3$ does not rank either
$a$ or $d$ first. If agent $3$ ranks $b$ first, $R_{ab}=1$ and
$R_{cd}=2$ force agent $4$ to rank $c$ first, and we obtain a degenerate
profile. Otherwise, agent $3$ ranks $c$ first, and then agent $4$
must rank either $c$ or $d$ first, but this violates $R_{bc}=1$
or $R_{cd}=2$, respectively.
\item When $L_{1}=L_{2}=\left(b,c\right)$, the constraints $R_{cd}=R_{ca}=R_{ab}=2$
imply that agents $3$ and $4$ must rank $c$ first. In that case,
agents $1$ and $2$ cannot receive $c$.
\item Otherwise, either $L_{1}=\left(c,d\right)$ or $L_{2}=\left(c,a\right)$.
However, agent $1$ cannot receive $d$ and agent $2$ cannot receive
$a$, so these relaxations cannot be applied.
\end{casenv}
\item Otherwise, agent $2$ ranks $dcba$. In this case, agents $1$ and
$2$ disagree on the adjacent pairs $\left\{ a,b\right\} $, $\left\{ b,c\right\} $,
and $\left\{ c,d\right\} $. By Remark \textcolor{blue}{\ref{both adjacent, reversed order}},
agents $3$ and $4$ cannot satisfy the constraints.
\end{casenv}

\bibliographystyle{plain}  
\bibliography{mybibliography}  

\end{document}